\documentclass[manuscript,screen]{acmart}

\AtBeginDocument{\providecommand\BibTeX{{\normalfont B\kern-0.5em{\scshape i\kern-0.25em b}\kern-0.8em\TeX}}}

\setcopyright{acmcopyright}
\copyrightyear{2018}
\acmYear{2018}
\acmDOI{10.1145/1122445.1122456}

\acmJournal{TOPLAS}
\acmVolume{1}
\acmArticle{1}
\acmYear{2021}
\acmMonth{1}
\acmDOI{} \startPage{1}

\usepackage{graphicx}
\usepackage{booktabs}   \usepackage{okexplain}
\usepackage{tikz-cd}
\usepackage{graphicx}
\usepackage{marvosym}
\usepackage{tikz}
\usetikzlibrary{decorations.pathreplacing,calc}
\newcommand{\ntikzmark}[2]{#2\thinspace\tikz[overlay,remember picture,baseline=(#1.base)]{\node[inner sep=0pt] (#1) {};}}

\newcommand{\makebrace}[3]{\begin{tikzpicture}[overlay, remember picture]
        \draw [decoration={brace,amplitude=0.5em},decorate]
        let \p1=(#1), \p2=(#2) in
        ({max(\x1,\x2)}, {\y1+0.8em}) -- node[right=0.6em] {#3} ({max(\x1,\x2)}, {\y2});
    \end{tikzpicture}
}

\usepackage{wrapfig}
\usepackage{caption}
\usepackage{subcaption}

\usepackage{makecell}  

\usepackage{listings}

\usepackage{svg}

\usepackage[shortlabels]{enumitem}

\usepackage{pdfcomment}

\usepackage{xcolor}
\usepackage{amsmath}
\usepackage{mathtools}
\usepackage{suffix}
\usepackage{amsfonts}
\usepackage{mathpartir}
\usepackage{enumitem}
\usepackage{stmaryrd}
\usepackage[all]{xy}
\usepackage{twoopt}

\usepackage{amsthm}

\makeatletter
\newtheorem*{rep@theorem}{\rep@title}
\newcommand{\newreptheorem}[2]{\newenvironment{rep#1}[1]{\def\rep@title{#2 ##1}\begin{rep@theorem}\def\@currentlabel{##1}}{\end{rep@theorem}}}
\makeatother
\newreptheorem{theorem}{Theorem}
\newreptheorem{corollary}{Corollary}
\newreptheorem{proposition}{Proposition}
\newtheorem*{theorem*}{Theorem}
\newtheorem*{corollary*}{Corollary}

\newcommand\computationtype[1]{\underline{#1}}
\makeatletter
\newcommand\@TyAlph[1]{\ifcase #1\or \tau\or \sigma\or \rho\else \@ctrerr \fi }
\newcommand\ty[1][1]{{\@TyAlph{#1}}}

\newcommand\@CTyAlph[1]{\computationtype{\ifcase #1\or \tau\or \sigma\or \rho\else \@ctrerr \fi}}
\newcommand\cty[1][1]{{\@CTyAlph{#1}}}

\newcommand\tvar[1][1]{{\@TyVarAlph{#1}}}
\newcommand\@TyVarAlph[1]{\ifcase #1\or \alpha\or \beta\or \gamma\else \@ctrerr \fi }
\newcommand\var[1][1]{{\@VarAlph{#1}}}
\newcommand\@VarAlph[1]{\ifcase #1\or x\or y\or z\or u\or v\or w\else \@ctrerr \fi }

\newcommand\trm[1][1]{{\@TermAlph{#1}}}
\newcommand\@TermAlph[1]{\ifcase #1\or t\or s\or r\else \@ctrerr \fi }

\newcommand\val[1][1]{\ifcase #1\or v\or w\or u\else \@ctrerr \fi }

\newcommand\op[1][1]{\ifcase #1\or \mathsf{op}\or \mathsf{op}'\or \mathsf{op}''\else \@ctrerr \fi }
\makeatother
\newcommand\lop{\mathsf{lop}}
\newcommand\Op{\mathsf{Op}}
\newcommand\LOp{\mathsf{LOp}}

\newcommand\cnst{\underline{c}}
\newcommand\lcnst{\underline{lc}}
\newcommand\zero{\underline{0}}
\newcommand\sigmoid{\varsigma}
\newcommand\tSum{\mathrm{sum}}

\newcommand\lPair[2]{\mathbf{lpair}(#1,#2)}
\newcommand\lFst{\mathbf{lfst}\,}
\newcommand\lSnd{\mathbf{lsnd}\,}
\newcommand\lcomp{;_{\ell}}

\newcommand\bp[1]{\boldsymbol{(}#1\boldsymbol{)}}
\newcommand\tUnit{\tTuple{}}

\newcommand\tPair[2]{\langle #1, #2\rangle}

\newcommand\tTriple[3]{\langle #1, #2, #3\rangle}
\newcommand\tTuple[1]{\langle #1\rangle}

\newcommand\fun[1]{\lambda #1.\,}

\newcommand\lfun[1]{\underline{\lambda} #1.\,}
\newcommand\lapp[2]{#1{\bullet} #2}
\newcommand\letin[3]{\mathbf{let}\,#1=#2\,\mathbf{in}\,#3}
\newcommand\pletin[4]{\letin{\tPair{#1}{#2}}{#3}{#4}}

\newcommand\lLambda{\underline{\Lambda}}

\newcommand\projection{\pi}
\newcommand\injection{\iota}
\newcommand\CMon{\mathbf{CMon}}

\newcommand\tensMatch[5][\,]{\mathbf{case}\,#2\,\mathbf{of}#1\copower{#3}{#4}\To#5}

\newcommand\tZipWith{\mathbf{zipWith}}

\newcommand\tFst{\mathbf{fst}\,}
\newcommand\tSnd{\mathbf{snd}\,}

\newcommand\tProj[1]{\mathbf{proj}_{#1}\,}
\newcommand\tCoProj[1]{\mathbf{coproj}_{#1}\,}
\newcommand\idx[2]{\mathbf{idx}(#1; #2)\,}

\newcommand\lvar{\mathsf{v}}

\newcommand\Cat{\mathbf{Cat}}

\newcommand\ctx{\Gamma}
\newcommand\tinf{\vdash}
\newcommand\Ginf[3][]{\ctx #1\tinf #2 : #3}
\newcommand\subst[2]{#1{}[#2]}
\newcommand\sfor[2]{^{#2}\!/\!_{#1}}

\newcommand\copower[3][]{{!}#2\otimes_{#1}#3}

\newcommand\creals{\underline{\mathbf{real}}}
\newcommand\reals{\mathbf{real}}

\newcommand\Unit{\mathbf{1}}

\WithSuffix\newcommand\t*{\boldsymbol{\mathop{*}}}

\newcommand\ListSym{\mathbf{List}}

\newcommand\MapSym{\mathbf{Copower}}

\newcommand\List[1]{\ListSym(#1)}
\newcommand\Map[2]{\MapSym(#1,#2)}
\newcommand\LinFunSym{\mathbf{LFun}}
\newcommand\LinFun[2]{\LinFunSym(#1,#2)}

\newcommand\To{\to}

\newcommand\bProd[2]{\bp{#1 \t* #2}}
\newcommand\tProd[3]{\bp{#1 \t* #2 \t* #3}}
\newcommand\EmptyList{\mathbf{[\,]}}

\newcommand\ListCons[2]{#1:: #2}
\newcommand\PlusList[2]{#1+\!\!+\, #2}
\newcommand\ListFold[5]{\mathbf{fold}\,#1\,\mathbf{over}\,#2\,\mathbf{in}\,#3\,\mathbf{from}\,#4=#5}

\newcommand\Dsynsymbol[1][]{\scalebox{0.8}{$\overrightarrow{\mathcal{D}}$}_{#1}}\newcommand\Dsyn[2][]{\Dsynsymbol[#1](#2)}

\newcommand\Dsynrevsymbol[1][]{\scalebox{0.8}{$\overleftarrow{\mathcal{D}}$}_{#1}}\newcommand\Dsynrev[2][]{\Dsynrevsymbol[#1](#2)}

\newcommand\CSyn{{\mathbf{CSyn}}}
\newcommand\LSyn{{\mathbf{LSyn}}}
\newcommand\Syn{\mathbf{Syn}}

\newcommand\tFromMaybe[1]{\mathrm{fromMaybe}}

\newcommand\tMap{\mathbf{map}}

\newcommand\freeeq[1]{\stackrel{\# #1}{=}}
\newcommand\beeq{\stackrel{\beta\eta}{=}}
\newcommand\bepeq{\!\stackrel{\beta\eta+}{=}\!}

\makeatletter
\newcommand{\pushright}[1]{\ifmeasuring@#1\else\omit$\displaystyle#1$\ignorespaces\fi}
\makeatother

\newcommand\explainr[1]{&\pushright{\color{gray}\scriptsize\{\;\textnormal{#1}\;\}}}

\newcommand\citeappx[1]{\ifx\fossacsversion\undefined Appx.~#1\else\cite[Appx.~#1]{vakar2020reverse}\fi}

\definecolor{shade}{RGB}{223,223,223}
\definecolor{unshade}{RGB}{255,255,255}

\usepackage{tcolorbox}
\newtcbox{\shadebox}{on line,arc=1pt, outer arc=2pt,colback=shade,colframe=shade,boxsep=0pt,left=1pt,right=1pt,top=2pt,bottom=2pt,boxrule=0pt,bottomrule=1pt,toprule=1pt}

\newtcbox{\unshadebox}{on line,arc=1pt, outer arc=2pt,colback=unshade,colframe=shade,boxsep=0pt,left=1pt,right=1pt,top=2pt,bottom=2pt,boxrule=0pt,bottomrule=1pt,toprule=1pt}

\newcommand\syncat[1]{\mspace{-25mu}\synname{#1}}
\newcommand\synname[1]{\qquad\text{#1}}

\newenvironment{syntax}[1][]{\(
\begin{array}[t]{#1l@{\quad\!\!}*3{l@{}}@{\,}l}
}{
\end{array}
\)}

\newcommand\gdefinedby{::=}
\newcommand\gor{\mathrel{\lvert}}

\newcommand{\wCpo}{\mathbf{\boldsymbol\omega CPO}}

\makeatletter
\def\MTrightharpoonupfill{\arrowfill@\relbar\relbar\rightharpoonup}
\def\MTleftharpoondownfill{\arrowfill@\leftharpoondown\relbar\relbar}
\def\MTleftharpoonupfill{\arrowfill@\leftharpoonup\relbar\relbar}
\def\MTrightharpoondownfill{\arrowfill@\relbar\relbar\rightharpoondown}
\newcommand*\xhookrightleftharpoons[2][]{\mathrel{\raise.22ex\hbox{$\lhook\joinrel\ext@arrow 0359\MTrightharpoonupfill{\phantom{#1}}{#2}$}\setbox0=\hbox{$\ext@arrow 3095\MTleftharpoondownfill{#1}{\phantom{\lhook\joinrel#2}}$}\kern-\wd0 \lower.22ex\box0}}
\newcommand*\xleftrighthookharpoons[2][]{\mathrel{\raise.22ex\hbox{$\ext@arrow 3095\MTleftharpoonupfill{\phantom{#1\mspace{15mu}}}{#2}$}\setbox0=\hbox{$\mathrel{\raise-.4837ex\hbox{$\lhook$}}\joinrel\ext@arrow 0359\MTrightharpoondownfill{#1}{\phantom{#2}}$}\kern-\wd0 \lower.22ex\box0}}
\makeatother

\newcommand\pair[2]{\parent{#1, #2}}
\newcommand\parent[1]{\left(#1\right)}
\WithSuffix\newcommand\pair-[2]{(#1, #2)}

\usepackage[bbgreekl]{mathbbol}

\newcommand{\lUnit}{\underline{\Unit}}

\newcommand{\CartSp}{\mathbf{CartSp}}

\newcommand{\Set}{\mathbf{Set}}

\newcommand\inv[1]{#1^{-1}}
\WithSuffix\newcommand\inv+[1]{\parent{#1}^{-1}}

\newcommand\initial{\mathbb{0}}
\newcommand\terminal{\mathbb{1}}

\newcommand\isomorphic\cong

\newcommand{\sem}[1]{\llbracket #1\rrbracket}
\newcommand{\semgl}[1]{\llparenthesis #1\rrparenthesis^f}
\newcommand{\semglrev}[1]{\llparenthesis #1\rrparenthesis^r}
\newcommand{\RR}{\mathbb{R}}
\newcommand{\NN}{\mathbb{N}}
\newcommand\cat[1]{\mathcal{#1}}
\newcommand\catC{\cat{C}}

\newcommand\catS{\cat{S}}
\newcommand\catL{\cat{L}}

\newcommand\Dsemsymbol[1][]{\mathcal{T}_{#1}}\newcommand\Dsem[2][]{\Dsemsymbol[#1]#2}
\newcommand\Dsemrevsymbol[1][]{\mathcal{T}^*_{#1}}
\newcommand\Dsemrev[2][]{\Dsemrevsymbol[#1]#2}
\newcommand\ev[1][]{\mathbf{ev}^{#1}}
\newcommand\evRsymbol[1][]{\mathrm{evR}^{#1}}
\newcommandtwoopt\evR[3][][]{\evRsymbol[#2]_{#1}(#3)}
\newcommand\lamRsymbol[1][]{\mathrm{lamR}^{#1}}
\newcommandtwoopt\lamR[3][][]{\lamRsymbol[#2]_{#1}(#3)}

\newcommand{\Gl}[1][]{\scalebox{0.8}{$\overrightarrow{\mathbf{SScone}}$}_{#1}}
\newcommand{\GlRev}[1][]{\scalebox{0.8}{$\overleftarrow{\mathbf{SScone}}$}_{#1}}

\newcommand{\sPair}[2]{( #1, #2 )}

\newcommand\transpose[1]{{#1}^{t}}

\newcommand\leval[1]{{\mathbf {leval}}_{#1}}
\newcommand\lsing[1]{\{(#1,-)\}}

\newcommand\lcopowfold[1][{}]{{\mathbf {lcopowfold}}\,{#1}}
\newcommand\lswap[1][{}]{{\mathbf {lswap}}_{#1}}
\newcommand\linearid[1][{}]{{\mathbf {lid}}_{#1}}
\newcommand\id[1][{}]{{\rm id}_{#1}}

\newcommand\xto\xrightarrow

\makeatletter
\newcommand*\bigcdot{\mathpalette\bigcdot@{.6}}
\newcommand*\bigcdot@[2]{\mathbin{\vcenter{\hbox{\scalebox{#2}{$\m@th#1\bullet$}}}}}
\makeatother

\newcommand\innerprod[2]{#1 \odot #2}

\newcommand\y{\mathbf{y}}

\newcommand\seq[2][]{\left(#2\right)_{#1}}
\newcommand\coseq[2][]{\left[#2\right]_{#1}}
\newcommand\set[1]{\left\{#1\right\}}

\newcommand\Domain[1]{\mathop{\rm Dom}\parent{#1}}

\renewcommand\lim{\mathrm{lim}}

\newcommand\ob[1]{\mathrm{ob}\,#1}
\renewcommand\circ{
  \PackageError{Matthijs}{Please use ; instead! Or if you'd rather use circ consistently, tell me and we'll switch the whole document over}{}}

\newcommand{\defeq}{\stackrel {\mathrm{def}}=}

\usepackage{todonotes}

\iftrue
\newcommand\mvremark[1]{\pdfmargincomment[color=yellow]{MV: #1}}
\newcommand\tsremark[1]{\pdfmargincomment[color=yellow]{TS: #1}}

\else
\newcommand\mvremark[1]{}
\newcommand\tsremark[1]{}

\fi

\makeatletter
\RequirePackage{zref}

\zref@newprop{oktheoremfreetext}{\oktheorem@parameter}

\newcommand\OKTheoremAddReferences[2]{
  \expandafter\newcommand\csname#1ref\endcsname[1]{#2~\ref{#1:##1}}
  \expandafter\newcommand\csname#1label\endcsname[1]{\label{#1:##1}}
  \WithSuffix\expandafter\newcommand\csname#1ref\endcsname*[1]{\ref{#1:##1}}
  \WithSuffix\expandafter\newcommand\csname#1label\endcsname+[1]{\hypertarget{#1+:##1}{}\zref@labelbyprops{#1:##1}{oktheoremfreetext}}
  \WithSuffix\expandafter\newcommand\csname#1ref\endcsname+[1]{\hyperlink{#1+:##1}{{{\let\ref\@refstar#2~\zref@extract{#1:##1}{oktheoremfreetext}}}}}
  \WithSuffix\expandafter\newcommand\csname#1ref\endcsname-[1]{\hyperlink{#1+:##1}{{\let\ref\@refstar{\zref@extract{#1:##1}{oktheoremfreetext}}}}}
}
\makeatother
\theoremstyle{definition}
\newtheorem{insight}{Insight}

\newtheorem*{exampleEnv*}{Example}
\newtheorem*{exampleEnv+}{Example \oktheorem@parameter}
\newenvironment{example*}{\begin{exampleEnv*}}{\qed\end{exampleEnv*}}
\newenvironment{example+}[1]{\def\oktheorem@parameter{#1}\begin{exampleEnv+}}{\qed\end{exampleEnv+}}
\OKTheoremAddReferences{example}{Example}
\newtheorem{thmx}{Theorem}

  \renewcommand\Domain[1]{\mathbf{Dom}(#1)}
  \newcommand\LDomain[1]{\mathbf{LDom}(#1)}
  \newcommand\CDomain[1]{\mathbf{CDom}(#1)}

\newcommand\ct[1]{\underline{#1}}
\newcommand\cRR{\ct{\RR}}

\newcommand\smooth{{differentiable}}

\newcommand\vGamma{\overline{\Gamma}} 
\allowdisplaybreaks

\begin{document}

\title{CHAD: Combinatory Homomorphic Automatic Differentiation}
\titlenote{This paper provides an extended version of \cite{vakar2020reverse},
augmenting it with
\begin{itemize}
\item examples and an interpretation of AD at higher-order types (notably in \S\ref{sec:key-ideas});
\item an explicit definition of the AD algorithm, directly phrased on a $\lambda$-calculus rather than categorical combinators (\S\ref{sec:ad-transformation});
\item a simplified semantics and correctness proof of the algorithm, based on sets rather than diffeological spaces, to make the paper more accessible (\S\ref{sec:semantics},\ref{sec:glueing-correctness});
\item proofs (\S\ref{sec:glueing-correctness});
\item an extended discussion of how to implement the proposed algorithm (\S\ref{sec:implementation});
\item a reference implementation of CHAD in Haskell (\S\ref{sec:implementation}), available at \href{https://github.com/VMatthijs/CHAD}{https://github.com/VMatthijs/CHAD} (under continuous improvement by Tom Smeding, Matthijs V\'ak\'ar, and others);
\item the concept of CHAD as a more broadly applicable technique for AD (and even other program analyses) on expressive functional languages (\S\ref{sec:scope});
\item various major rewrites throughout.
\end{itemize}
}

\author{Matthijs V\'ak\'ar}
\email{m.i.l.vakar@uu.nl}

\author{Tom Smeding}
\email{t.j.smeding@uu.nl}
\affiliation{\institution{Utrecht University}
\city{Utrecht}
\country{Netherlands}
}

\renewcommand{\shortauthors}{Matthijs V\'ak\'ar and Tom Smeding}

\begin{abstract}
We introduce Combinatory Homomorphic Automatic Differentiation (CHAD),
a principled, pure, provably correct define-then-run method for performing forward
and reverse mode automatic differentiation (AD) on programming languages with 
expressive features.
It implements AD as a compositional, type-respecting source-code transformation that generates purely functional code.
This code transformation is principled in the sense that it is the unique homomorphic (structure-preserving) extension to expressive languages of Elliott's well-known and unambiguous definitions of AD for a first-order functional language.
Correctness of the method follows by a (compositional) logical relations argument that shows that the semantics of the syntactic derivative is the usual calculus derivative of the semantics of the original program.

In their most elegant formulation, the transformations generate code with linear types.
However, the code transformations can be implemented in a standard functional language lacking linear types: while the correctness proof requires tracking linearity, the actual transformations do not.
In fact, even in a standard functional language, we can obtain all the type safety that linear types provide: we can implement all linear types used to type the transformations as abstract types, using a basic module system.

In this paper, we detail the method when applied to a simple higher-order language for manipulating statically sized arrays.
However, we explain how the methodology applies, more generally, to functional languages with other expressive features.
Finally, we discuss how the scope of CHAD extends beyond applications in AD to other dynamic program analyses that accumulate data in a commutative monoid.
\end{abstract}



\maketitle

\section{Introduction}\label{sec:introduction}
Automatic differentiation (AD) is a technique for
transforming code that implements a function $f$ into code 
that computes $f$'s derivative, essentially by 
using the chain rule for derivatives.
Due to its efficiency and numerical stability, 
AD is the technique of choice whenever we need to compute derivatives of functions
that are implemented as programs, particularly in 
high-dimensional settings.
Optimization and Monte-Carlo integration 
algorithms, such as gradient descent and Hamiltonian Monte-Carlo methods,
rely crucially on the calculation of derivatives.
These algorithms are used in virtually every machine learning 
and computational statistics application, and the calculation 
of derivatives is usually the computational bottleneck.
These applications explain the recent surge of interest in AD,
which has resulted in the proliferation of popular AD systems such as TensorFlow \cite{abadi2016tensorflow},
PyTorch \cite{paszke2017automatic}, and Stan Math \cite{carpenter2015stan}.

AD, roughly speaking, comes in two modes: forward mode and reverse mode.
When differentiating a function $\RR^n\to \RR^m$, forward mode tends to 
be more efficient if $m\gg n$, while reverse mode is generally more efficient
if $n\gg m$.
As most applications reduce to optimization or Monte-Carlo integration of an 
objective function $\RR^n\to\RR$ with $n$ very large (at the time of this paper, on the order of $10^4-10^7$),
reverse mode AD is in many ways the more interesting algorithm \cite{baydin2014automatic}.

However, reverse AD is also more complicated to understand and implement than forward AD.
Forward AD can be implemented as a structure-preserving 
program transformation, even on languages with complex features \cite{shaikhha2019efficient}.
Thus, it admits an elegant proof of correctness \cite{hsv-fossacs2020}.
By contrast, reverse AD is only well-understood as a compile-time source-code transformation that does not use a run-time interpreter
(also called \emph{define-then-run} style AD) 
on limited programming languages, such as first-order functional languages.
Typically, its implementations on more expressive languages that have features 
such as higher-order functions use interpreted \emph{define-by-run}
approaches. These approaches first build a computation graph at run time,
effectively evaluating the program until a straight-line first-order program is left,
and then they perform automatic differentiation on this new program \cite{paszke2017automatic,carpenter2015stan}.
First, such approaches have the severe downside that they can suffer from interpretation overhead.
Second, the differentiated code cannot benefit as much from existing optimizing compiler architectures.
As a result, these AD libraries need to be implemented using carefully hand-optimized code that, for example, does not contain any common subexpressions.
This implementation process is precarious and labour-intensive.
Furthermore, some whole-program optimizations that a compiler would detect
go entirely unused in such systems.

Similarly, correctness proofs of reverse AD have taken a define-by-run approach 
or have relied on non-standard
operational semantics, using forms of symbolic execution \cite{abadi-plotkin2020,mak-ong2020,brunel2019backpropagation}.
Most work that treats reverse AD as a source-code transformation 
does so using complex transformations that introduce
mutable state and/or 
non-local control flow \cite{pearlmutter2008reverse,wang2018demystifying}.
As a result, it is unclear whether and why such techniques are correct.
Furthermore, AD applications (e.g. in machine learning) tend to be run on parallel hardware,
which can be easier to target with purely functional code.
Another approach has been to compile high-level languages to a low-level 
imperative representation first and then perform AD at that level \cite{innes2018don}, using mutation and jumps.
This approach has the downside that we might lose important opportunities for compiler 
optimizations, such as map-fusion and embarrassingly parallel maps,
which we can exploit if we perform define-then-run AD on a high-level functional representation.

A notable exception to these define-by-run and non-functional approaches to 
AD is Elliott's work \cite{elliott2018simple}, which presents an elegant, purely functional, 
define-then-run version of reverse AD.
Unfortunately, their techniques are limited to first-order programs over tuples of real numbers.
The workshop paper \cite{vytiniotis2019differentiable} by Vytiniotis, Belov, Wei, Plotkin, and Abadi proposes two possible extensions of 
 Elliott's functional AD to accommodate higher-order functions. 
However, it does not address whether or why these extensions would be correct or establish
 a more general methodology for applying AD to languages 
with expressive features.

This paper introduces Combinatory Homomorphic Automatic Differentiation (CHAD) and its proof of correctness.
CHAD is based on the observation that Elliott's work in \cite{elliott2018simple} has a unique structure-preserving
extension that lets us perform AD on various expressive programming language features.
We see purely functional higher-order (parallel) array processing languages such as Accelerate \cite{chakravarty2011accelerating} and Futhark \cite{henriksen2017futhark} as particularly relevant platforms for the machine learning 
applications for which AD tends to be used.
With that in mind, we detail CHAD when applied to higher-order functional programs over (primitive) arrays of reals.
This paper makes the following contributions:
\begin{itemize}
\item We introduce CHAD, a categorical perspective on AD, which lets us see AD as a uniquely determined homomorphic (structure-preserving) functor from the syntax of its source programming language (\S
\ref{sec:language}) to the syntax of its target language (\S
\ref{sec:minimal-linear-language}).
\item We explain, from this categorical setting, precisely in what sense reverse AD is the ``mirror image'' of forward AD (\S\ref{sec:self-dualization}).
\item We detail how this technique lets us define purely functional define-then-run reverse mode AD on a higher-order language (\S\ref{sec:ad-transformation}).
\item We present an elegant proof of the semantic correctness of the resulting AD transformations, based on a 
semantic logical relations argument, 
demonstrating that the transformations calculate the derivatives of the program in the usual mathematical sense (\S
\ref{sec:semantics} and \S\ref{sec:glueing-correctness}).
\item We show that the AD definitions and correctness proof are extensible to higher-order primitives such as a
$\tMap$-operation over our primitive arrays (\S\ref{sec:higher-order-primitives}).
\item We show how our techniques are readily implementable in standard functional languages 
to give purely functional, principled, semantically correct, compositional, define-then-run reverse mode~AD (\S\ref{sec:implementation}).
\item Finally, we place CHAD in a broader context and explain how it applies, more generally, to dynamic 
program analyses that accumulate information in a commutative monoid (\S\ref{sec:scope}).
\end{itemize}
We start by giving a high-level overview of the main insights and theorems of this paper in \S\ref{sec:key-ideas}.

For a review of the basics of AD, we refer the reader to \cite{baydin2014automatic,margossian2019review}.
We discuss recent related work studying AD from a programming languages perspective in \S\ref{sec:related-work}.
\section{Key ideas}\label{sec:key-ideas}
We start by providing a high-level overview of the paper, highlighting the main insights and theorems underlying our contributions. 

\subsection{Aims of Automatic Differentiation}
The basic challenge that automatic differentiation aims to solve 
is the following.
We are given a program $\var:\reals^n\vdash \trm:\reals^m$
that takes an $n$-dimensional array of (floating-point) real numbers 
as input and produces an $m$-dimensional array of reals as output.
That is, $\trm$ computes some mathematical function $\sem{\trm}:\RR^n\to \RR^m$.
We want to transform the code of $\trm$ into:\footnote{The program transformations are called $\Dsyn{-}_2$ and $\Dsynrev{-}_2$ here. In \S\S\ref{ssec:pairing-sharing} we discuss that it is better to define our actual program transformations to have a slightly different type. The second half of those transformations (defined in \S\S\ref{ssec:first-order-chad}) corresponds to these $\Dsyn{-}_2$ and $\Dsynrev{-}_2$.}
\begin{itemize}
    \item  a program $\Dsyn{\trm}_2$
that computes the derivative $D\sem{\trm}:\RR^n\to \cRR^n\multimap \cRR^m$, in the case of forward AD;
    \item a program $\Dsynrev{\trm}_2$ that computes the transposed 
    derivative $\transpose{D\sem{\trm}}:\RR^n\to\cRR^m\multimap \cRR^n$, in the case of reverse AD.
\end{itemize}
Here, we write $\cRR^n$ for the space of (co)tangent vectors to $\RR^n$; we regard $\cRR^n$ as a commutative monoid under elementwise addition. We write $\multimap$ for a linear function type to emphasize that derivatives are linear in the sense of being monoid homomorphisms.

Furthermore, we have the following desiderata for these code transformations:
\begin{enumerate}
\item we want these code transformations to be defined compositionally, so we can easily extend the source programming language we apply the transformations to with new primitives;
\item we want these transformations to apply to a wide range of programming techniques, so we are not limited in our programming style even if we want our code to be differentiated;
\item we want the transformations to generate purely functional code 
so we can easily prove its correctness and deploy it on parallel hardware;
\item we want the code size of $\Dsyn{\trm}_2$ and $\Dsynrev{\trm}_2$ to grow linearly in the size of $\trm$, so we can apply the technique to large codebases;
\item we want the time complexity of $\Dsyn{\trm}_2$ and $\Dsynrev{\trm}_2$ to be proportional to that of $\trm$ and, generally, as low as possible;
this means that we can use forward AD to efficiently compute a column of the Jacobian matrix of partial derivatives, while reverse AD efficiently computes a row of the Jacobian.
\end{enumerate}
In this paper, we demonstrate how the CHAD technique of automatic differentiation satisfies desiderata (1)-(4) -- we leave (5) to future work.
It achieves this by taking seriously the mathematical structure of programming languages as freely generated categories and by observing that differentiation is compositional according to the chain rule.

\subsection{The Chain Rule -- Pairing and Sharing of Primals and Derivatives}\label{ssec:pairing-sharing}
To achieve desideratum (1) of compositionality, it is tempting to examine the 
chain rule, the key compositionality property of derivatives.
Given $f:\RR^n\to\RR^m$, we write
\begin{align*}
\Dsem{f}:\;&\RR^n\to \RR^m\times (\cRR^n\multimap \cRR^m)\\
&x\;\; \mapsto\sPair{f(x)}{\; \;v\;\,\mapsto Df(x)(v)}
\end{align*}
for the function that pairs up the primal function value $f(x)$ with the derivative $Df(x)$ of $f$ at $x$ that acts on tangent vectors $v$.
The chain rule then gives the following formula for the derivative of the composition $f;g$ of $f$ and $g$:
$$
\Dsem{(f;g)}(x) = \sPair{\Dsem[1]{g}(\Dsem[1]{f}(x))}{\Dsem[2]{f}(x);\Dsem[2]{g}(\Dsem[1]{f}(x))},
$$
where we write $\Dsem[1]{f}\defeq  \Dsem{f};\pi_1$ and $\Dsem[2]{f}\defeq  \Dsem{f};\pi_2$ for the first and second components of $\Dsem{f}$, respectively.
We make two observations:
\begin{enumerate}
\item the derivative of the composition $f;g$ depends not only on the derivatives of $g$ and $f$ but also on the primal value of $f$;
\item the primal value of $f$ is used twice: once in the primal value of $f;g$ and once in its derivative; we want to share these repeated subcomputations, to address desiderata (4) and (5).
\end{enumerate}
\begin{insight}\textit{
It is wise to \emph{pair up} computations of primal function values and derivatives and to \emph{share} computation between them if we want to calculate derivatives of functions compositionally and efficiently.}
\end{insight}

Similarly, we can pair up $f$'s transposed (adjoint) derivative $\transpose{Df}$, which propagates cotangent rather than tangent vectors:
\begin{align*}
    \Dsemrev{f} : \;&\RR^n \to \RR^m\times (\cRR^m\multimap \cRR^n)\\
    &x\;\;\mapsto \sPair{f(x)}{\;\;v\;\,\,\mapsto \transpose{Df}(x)(v)}.
    \end{align*}
It then satisfies the following chain rule, which follows from the usual chain rule above together with the fact that $\transpose{(A;B)}=\transpose{B};\transpose{A}$ for linear maps $A$ and $B$ (transposition is contravariant -- note the resulting reversed order of $\Dsemrev[2]{f}$ and $\Dsemrev[2]{g}$ for reverse AD):
$$
\Dsemrev{(f;g)}(x) = \sPair{\Dsemrev[1]{g}(\Dsemrev[1]{f}(x))}{\Dsemrev[2]{g}(\Dsemrev[1]{f}(x));\Dsemrev[2]{f}(x)}.
$$
Again, pairing and sharing the primal and (transposed) derivative computations 
is beneficial.

CHAD directly implements the operations $\Dsemsymbol$ and $\Dsemrevsymbol$ as source-code transformations $\Dsynsymbol$ and $\Dsynrevsymbol$ on a functional language to implement forward\footnote{For forward AD, we can also choose to implement instead
\begin{align*}
\Dsemsymbol'f:\; &(\RR^n\times \RR^n)\to (\RR^m\times \RR^m)\\
&(x, v) \mapsto (f(x), Df(x)(v))
\end{align*}
together with its chain rule as code transformations.
This leads to a different style of forward AD based on a \emph{dual numbers representation}.
\cite{hsv-fossacs2020} gives an analysis of this style of forward AD, similar 
to the treatment of reverse AD and (non-dual number) forward AD in this paper.
Although forward AD with dual numbers is more memory-efficient and preferable in 
practical implementations, it does not have an obvious reverse-mode variant.
See \S\S\ref{ssec:dual-numbers-forward-ad} for more discussion.
} and reverse mode AD, respectively.
These code transformations are defined compositionally through structural induction on the syntax, by exploiting the chain rules above combined with the 
categorical structure of programming languages.

\subsection{CHAD on a First-Order Functional Language}\label{ssec:first-order-chad}
Here, we outline how CHAD looks when applied to programs written in a first-order functional language.
We treat this material as known because it is essentially the algorithm of \cite{elliott2018simple}.
However, we present it in terms of a $\lambda$-calculus rather than categorical combinators, by applying the well-known mechanical translations between the two formalisms \cite{curien1986categorical}. We hope that this presentation makes the algorithm easier to apply in practice.

We consider a source programming language (see \S\ref{sec:language}) where we write $\ty,\ty[2],\ty[3]$ for types that are either statically sized arrays of $n$ real numbers $\reals^n$ or tuples $\ty\t*\ty[2]$ of types $\ty,\ty[2]$.
These types will be called \emph{first-order types} in this section.\footnote{In the rest of the paper, we also consider the unit type $\Unit$ a first-order type. These types are also called \emph{ground types}.}
We consider programs $\trm$ of type $\ty[2]$ in a typing context $\Gamma=\var_1:\ty_1,\ldots,\var_n:\ty_n$, where $\var_i$ are identifiers.
We write such typings of programs in a context as $\Gamma\vdash\trm:\ty[2]$.
As long as our language has certain primitive operations (which we represent schematically)
$$
\inferrule{\Ginf{\trm_1}{\reals^{n_1}}\quad\cdots\quad \Ginf{\trm_k}{\reals^{n_k}}}{\Ginf{\op(\trm_1,\ldots,\trm_k)}{\reals^m}}
$$
such as constants (as nullary operations), (elementwise) addition and multiplication of arrays,
inner products and certain non-linear functions such as sigmoid functions, we can write complex programs by sequencing together such operations.
Fig. \ref{fig:first-order-ad-example} (a) and (b) give examples of programs we can write, where we write $\reals$ for $\reals^1$ and indicate shared subcomputations with $\mathbf{let}$-bindings.

CHAD transforms the types and programs of this source language into types and programs of a suitably chosen target language (see \S\ref{sec:minimal-linear-language}) that is a superset of the source language.
CHAD associates the following types to each source-language type $\ty$:
\begin{itemize}
    \item forward mode primal values $\Dsyn{\ty}_1$;\\
    we define $\Dsyn{\reals^n}_1\defeq\reals^n$ and $\Dsyn{\ty\t*\ty[2]}_1\defeq\Dsyn{\ty}_1\t*\Dsyn{\ty[2]}_1$; that is, for now $\Dsyn{\ty}_1=\ty$; 
    \item reverse mode primal values $\Dsynrev{\ty}_1$;\\ we define $\Dsynrev{\reals^n}_1\defeq\reals^n$ and $\Dsynrev{\ty\t*\ty[2]}_1\defeq\Dsynrev{\ty}_1\t*\Dsynrev{\ty[2]}_1$; that is, for now $\Dsynrev{\ty}_1=\ty$; 
    \item forward mode tangent values $\Dsyn{\ty}_2$;\\
    we define $\Dsyn{\reals^n}_2\defeq\creals^n$ and $\Dsyn{\ty\t*\ty[2]}_2\defeq\Dsyn{\ty}_2\t*\Dsyn{\ty[2]}_2$;
    \item reverse mode cotangent values $\Dsynrev{\ty}_2$;\\
    we define $\Dsynrev{\reals^n}_2\defeq\creals^n$ and $\Dsynrev{\ty\t*\ty[2]}_2\defeq\Dsynrev{\ty}_2\t*\Dsynrev{\ty[2]}_2$.
\end{itemize}
The types $\Dsyn{\ty}_1$ and $\Dsynrev{\ty}_1$ of primals are Cartesian types, which we can think of as denoting sets,
while the types $\Dsyn{\ty}_2$ and $\Dsynrev{\ty}_2$ are linear types that denote commutative monoids.
That is, such linear types in our language need to have a commutative monoid structure $(\zero,+)$.
For example, $\creals^n$ is the commutative monoid over $\reals^n$ where $\zero$ is the zero vector and $(+)$ is elementwise addition of vectors.
Derivatives and transposed derivatives are then linear functions, that is, homomorphisms of this $(\zero,+)$-monoid structure.
As we will see, we use the monoid structure to initialize and accumulate (co)tangents in the definition of CHAD.

We extend these operations $\Dsynsymbol$ and $\Dsynrevsymbol$ to act not only on types but also on typing contexts $\Gamma=\var_1:\ty_1,\ldots,\var_n:\ty_n$ to produce primal contexts and (co)tangent types:
\begin{align*}
\Dsyn{\var_1:\ty_1,\ldots,\var_n:\ty_n}_1&=\var_1:\Dsyn{\ty_1}_1,\ldots, \var_n:\Dsyn{\ty_n}_1 \qquad\qquad \Dsyn{\var_1:\ty_1,\ldots,\var_n:\ty_n}_2=\Dsyn{\ty_1}_2\t*\cdots\t*\Dsyn{\ty_n}_2\\ 
\Dsynrev{\var_1:\ty_1,\ldots,\var_n:\ty_n}_1&=\var_1:\Dsynrev{\ty_1}_1,\ldots, \var_n:\Dsynrev{\ty_n}_1 \qquad\qquad 
\Dsynrev{\var_1:\ty_1,\ldots,\var_n:\ty_n}_2=\Dsynrev{\ty_1}_2\t*\cdots\t*\Dsynrev{\ty_n}_2.
\end{align*}
To each program $\Gamma\vdash\trm:\ty[2]$, CHAD then associates programs that calculate the forward-mode and reverse-mode derivatives $\Dsyn[\vGamma]{\trm}$ and $\Dsynrev[\vGamma]{\trm}$, whose definitions use the list $\vGamma$ of identifiers that occur in $\Gamma$:
\begin{align*}
&\Dsyn{\Gamma}_1\vdash \Dsyn[\vGamma]{\trm} : \Dsyn{\ty[2]}_1\t* (\Dsyn{\Gamma}_2\multimap \Dsyn{\ty[2]}_2)\\
&\Dsynrev{\Gamma}_1\vdash \Dsynrev[\vGamma]{\trm} : \Dsynrev{\ty[2]}_1\t* ( \Dsynrev{\ty[2]}_2\multimap \Dsynrev{\Gamma}_2).
\end{align*}
Since each program $\trm$ computes a \smooth{} function $\sem{\trm}$ between Euclidean spaces when all primitive operations $\op$ are \smooth{}, the key property that we prove for these code transformations is that they actually calculate derivatives:
\begin{thmx}[Correctness of CHAD, Thm. \ref{thm:AD-correctness}]\label{thm:chad-correctness}
For any well-typed program (where $\ty_i$ and $\ty[2]$ are first-order types, i.e.\ $\reals^n$ and tuples of such types) $$\var_1:\ty_1,\ldots,\var_n:\ty_n\vdash {\trm}:\ty[2]$$
we have
$\sem{\Dsyn{\trm}}=\Dsem{\sem{\trm}}\;\text{ and }\;\sem{\Dsynrev{\trm}}=\Dsemrev{\sem{\trm}}.$
\end{thmx}
Once we fix a semantics for the source and target languages, we can show that this theorem holds if we define $\Dsynsymbol$ and $\Dsynrevsymbol$ on programs using the chain rule.
The proof proceeds by straightforward induction on the syntax.

For example, we can correctly define reverse mode CHAD on a first-order language as follows (see \S\ref{sec:ad-transformation}):
\begin{flalign*}
    &\Dsynrev[\vGamma]{\op(\trm_1,\ldots,\trm_k)} &&\defeq && \pletin{\var_1}{\var_1'}{\Dsynrev[\vGamma]{\trm_1}}{\cdots
     \pletin{\var_k}{\var_k'}{\Dsynrev[\vGamma]{\trm_k}}{\\
    &&&&&\tPair{\op(\var_1,\ldots,\var_k)}{\lfun\lvar \letin{\lvar}{\transpose{D\op}(\var_1,\ldots,\var_k;\lvar)}{\lapp{\var_1'}{(\tProj{1}{\lvar})}+\cdots+\lapp{\var_k'}{(\tProj{k}{\lvar})}}}}}
\\
&\Dsynrev[\vGamma]{\var} && \defeq &&  \tPair{\var}{\lfun{\lvar} \tCoProj{\idx{\var}{\vGamma}}(\lvar)}
\\
    &
\Dsynrev[\vGamma]{\letin{\var}{\trm}{\trm[2]}}  &&
\defeq &&
\pletin{\var}{\var'}{\Dsynrev[\vGamma]{\trm}}{
    \pletin{\var[2]}{\var[2]'}{\Dsynrev[\vGamma,\var]{\trm[2]}}{
        \tPair{\var[2]}{\lfun\lvar 
        \letin{\lvar}{\lapp{\var[2]'}{\lvar}}{
            \tFst\lvar+\lapp{\var'}{(\tSnd \lvar)}
        }}
    }}
\\&
\Dsynrev[\vGamma]{\tPair{\trm}{\trm[2]}} && \defeq && 
\pletin{\var}{\var'}{\Dsynrev[\vGamma]{\trm}}{ 
\pletin{\var[2]}{\var[2]'}{\Dsynrev[\vGamma]{\trm[2]}}{
\tPair{\tPair{\var}{\var[2]}}{\lfun\lvar \lapp{\var'}{(\tFst\lvar)} + \lapp{\var[2]'}{(\tSnd \lvar)}}}}
\\&
\Dsynrev[\vGamma]{\tFst\trm} &&\defeq && 
\pletin{\var}{\var'}{\Dsynrev[\vGamma]{\trm}}
{\tPair{\tFst\var}{\lfun\lvar \lapp{\var'}{\tPair{\lvar}{\zero}}}}
\\&
\Dsynrev[\vGamma]{\tSnd\trm} && \defeq && 
\pletin{\var}{\var'}{\Dsynrev[\vGamma]{\trm}}
{\tPair{\tSnd\var}{\lfun\lvar \lapp{\var'}{\tPair{\zero}{\lvar}}}}.
\end{flalign*}
Here, we write $\lfun\lvar\trm$ for a linear function abstraction (merely a notational convention -- it can simply be thought of as a plain function abstraction) and $\lapp{\trm}{\trm[2]}$ for a linear function application of $\trm : \ty \multimap \ty[2]$ to the argument $\trm[2] : \ty$ (which again can be thought of as a plain function application).
Furthermore,
given a program
$\trm$ of tuple type ${\cty[2]_1}\t*{\cdots}\t*{\cty[2]_n}$, we write $\tProj{i}{\trm}$ for its $i$-th projection of type $\cty[2]_i$.
Similarly,
given a program $ \trm$ of linear type $\cty[2]_i$, we write $\tCoProj{i}(\trm)$ for the $i$-th coprojection $\tTuple{\zero,\ldots,\zero,\trm,\zero,\ldots,\zero}
$ of type ${\cty[2]_1}\t*{\cdots}\t*{\cty[2]_n}$ and we write $\idx{\var_i}{\var_1,\ldots,\var_n}=i$ for the index of an identifier in a list of identifiers.  
Finally, $\transpose{D\op}$ here is a linear operation that implements the transposed derivative of the primitive operation $\op$.
We note that we crucially need the commutative monoid structure on linear types to correctly define the reverse mode derivatives of programs that involve tuples (or $n$-ary operations for $n\neq 1$).
Intuitively, matrix transposition (of derivatives) flips the copying-deleting comonoid structure provided by tuples into the addition-zero monoid structure.
\begin{insight}\textit{
    In functional define-then-run reverse AD,
    we need to have a commutative monoid structure on types of cotangents
    to mirror the comonoid structure coming from tuples:
    copying fan-out in the original program gets translated into
    fan-in in the transposed derivative, for accumulating incoming cotangents.
    This leads to linear types of cotangents.
    }
\end{insight}

Furthermore, observe that CHAD pairs up primal and (co)tangent values and shares common subcomputations, as desired. 
We see that what CHAD achieves is a compositional and efficient reverse mode AD algorithm that computes the (transposed) derivatives of a composite program in terms of the (transposed) derivatives $\transpose{D\op}$ of the basic building blocks $\op$.
Finally, it does so in a way that satisfies desiderata (1)-(4).

\begin{figure}[!h]
    \begin{subfigure}[b]{0.45\linewidth}
\begin{lstlisting}[mathescape=true]
let y = 2 * x     
    z = x * y   
    w = cos z   
    v = $\langle$y,z,w$\rangle$ in
    v
\end{lstlisting}
\subcaption{Original program $x:\reals\vdash\trm:\reals\t*\reals\t*\reals$ computing a function $\sem{\trm}:\RR\to \RR^3; x\mapsto v$.}
\end{subfigure}\hspace{0.08\linewidth}
\begin{subfigure}[b]{0.45\linewidth}
\begin{lstlisting}[mathescape=true]
let y = x1 * x4 + 2 * x2
    z = y * x3
    w = z + x4
    v = sin w in
    v     
\end{lstlisting}
\subcaption{Original program $x1:\reals,x2:\reals,x3:\reals,x4:\reals\vdash\trm[2]:\reals$ computing a function $\sem{\trm[2]}:\RR^4\to \RR;\\ (x1,x2,x3,x4)\mapsto v$.}
    \end{subfigure}
\\
\quad\\
\begin{subfigure}[b]{0.45\linewidth}

\begin{tabular}{ll}
\begin{lstlisting}[mathescape=true]
let y  = 2 * x
    z  = x * y   
    w  = cos z    
    v  = $\langle$y,z,w$\rangle$ in
$\langle$v, $\lambda$x'.
    let y' = 2 * x'
        z' = x' * y + x * y'
        w' = -sin z * z'
        v' = $\langle$y',z',w'$\rangle$ in 
      v'$\rangle$
    


\end{lstlisting}&\hspace{-18pt}
\begin{tabular}{l}
    \ntikzmark{S}{ }\\
    \\
    \\
    \ntikzmark{E}{ }\makebrace{S}{E}{primals}\\\\
    \ntikzmark{S}{ }\\\\
    \\
    \\
    \ntikzmark{E}{ }\makebrace{S}{E}{tangents}\\\\
    \\
    \\
\end{tabular}\hspace{-30pt}
\end{tabular}
    \subcaption{Forward AD transformed program $x:\reals\vdash \Dsyn{\trm}:(\reals\t*\reals\t*\reals)\t*(\creals\multimap \creals \t*\creals\t*\creals)$ computing the derivative
    $\sem{\Dsyn{\trm}}=\Dsem{\sem{\trm}}:\RR\to \RR^3\times (\cRR\multimap \cRR^3); x\mapsto (v,x' \mapsto v')$.\\\\\\\\}
    \end{subfigure}\hspace{0.08\linewidth}
\begin{subfigure}[b]{0.45\linewidth}
    \begin{tabular}{ll}
\begin{lstlisting}[mathescape=true]
let y   = x1 * x4
            + 2 * x2
    z   = y * x3
    w   = z + x4
    v   = sin w in
$\langle$v, $\lambda$v'.
    let w'  = cos w * v'
        z'  = w'
        y'  = x3 * z'
        x1' = y' * x4
        x2' = 2 * y'
        x3' = y * z'
        x4' = x1 * y'
               + w' in 
        $\langle$x1',x2',x3',x4'$\rangle$$\rangle$
\end{lstlisting} &\hspace{-18pt}   \begin{tabular}{l}
    \ntikzmark{S}{ }\\
    \\
    \\
    \\
    \ntikzmark{E}{ }\makebrace{S}{E}{primals}\\\\
    \ntikzmark{S}{ }\\\\
    \\
    \\
    \\
    \\
    \\
    \\
    \ntikzmark{E}{ }\makebrace{S}{E}{cotangents}\\
\end{tabular}\hspace{-30pt}
\end{tabular}
\subcaption{Reverse AD transformed program $x1:\reals,x2:\reals,x3:\reals,x4:\reals\vdash \Dsynrev{\trm[2]}:\reals\t*(\creals\multimap\creals\t*\creals\t*\creals\t*\creals)$ computing the transposed derivative
$\sem{\Dsynrev{\trm[2]}}=\Dsemrev{\sem{\trm[2]}}:\RR^4\to \RR\times (\cRR\multimap \cRR^4); (x1,x2,x3,x4)\mapsto (v, v' \mapsto (x1',x2',x3',x4'))$.
We see that the copying of $x4$ is translated into the addition of cotangents in $x4'$.}
\end{subfigure}
\caption{\label{fig:first-order-ad-example}Forward and reverse AD illustrated on simple first-order functional programs.
}
\end{figure}

For example, Fig. \ref{fig:first-order-ad-example} (c) and (d) display the code 
that forward and reverse mode CHAD, respectively, generate for the source programs in (a) and (b).
This is the code that is actually generated by the CHAD code transformations 
in our Haskell implementation followed by some very basic simplifications that do not affect time complexity and whose only purpose here is to aid legibility.
For more information about how exactly this code relates to the output one gets when applying the forward and reverse AD macros in this paper to the source programs, see Appendix \ref{appx:impl-simpl}.

\subsection{Intermezzo: the Categorical Structure of CHAD}
While this definition of CHAD on a first-order language straightforwardly follows 
from the mathematics of derivatives,
it is not immediately clear how it should be extended to source languages with more expressive features such as higher-order functions.
Indeed, we do not typically consider derivatives of higher-order functions in 
calculus. 
In fact, it is not even clear what a tangent or cotangent to a function type should be, or, for that matter, what a primal associated with a value of function type is.
To solve this mystery, we employ some category theory.

Observe that the first-order source language we consider can be viewed as a category $\Syn$ with products (see \S\ref{sec:language}): its objects are types $\ty,\ty[2],\ty[3]$
and morphisms $\trm\in\Syn(\ty,\ty[2])$ are programs $\var:\ty\vdash \trm:\ty[2]$ modulo standard $\beta\eta$-program equivalence (identities are given by variables and composition is done through $\mathbf{let}$-bindings).
This category is freely generated by the objects $\reals^n$ and morphisms $\op$
in the sense that any consistent assignment of objects $F(\reals^n)$ and morphisms $F(\op)$ in a category with products $\catC$ extends to a unique product-preserving functor $F:\Syn\to \catC$.

Suppose that we are given a categorical model $\catL:\catC^{op}\to \Cat$ of linear logic (a so-called locally indexed category -- see, for example, \cite[\S\S\S 9.3.4]{levy2012call}),
where we think of the objects and morphisms of $\catC$ as the semantics of Cartesian types and their programs and of the objects and morphisms of $\catL$ as the semantics of linear types and their programs.
We observe that we can define categories $\Sigma_{\catC}\catL$ and $\Sigma_{\catC}\catL^{op}$ (their so-called Grothendieck constructions, or $\Sigma$-types, see \S\ref{sec:self-dualization}) with 
objects that are pairs $(A_1, A_2)$ with $A_1$ an object of $\catC$ and $A_2$ an object of $\catL$ and homsets
\begin{align*}
    \Sigma_{\catC}\catL((A_1,A_2),(B_1,B_2))&\defeq \catC(A_1,B_1)\times \catL(A_1)(A_2, B_2)\cong \catC(A_1,B_1\times (A_2\multimap B_2))\\
    \Sigma_{\catC}\catL^{op}((A_1,A_2),(B_1,B_2))&\defeq  \catC(A_1,B_1)\times \catL(A_1)(B_2, A_2)\cong \catC(A_1,B_1\times (B_2\multimap A_2)).
    \end{align*}
We prove that these categories have finite products, provided that some conditions are satisfied: namely, that $\catC$ has finite products and $\catL$ has indexed finite biproducts (or equivalently: has indexed finite products and is enriched over commutative monoids).
Indeed, then $\prod_{i\in I}(A_{1i},A_{2i})=(\prod_{i\in I}A_{1i}, \prod_{i\in I}A_{2i})$.
In other words, it is sufficient if our model of linear logic is biadditive.
In particular, the categorical model of linear logic that we can build from the syntax of our target language for CHAD, $\LSyn:\CSyn^{op}\to\Cat$, satisfies our conditions (in fact, it is the initial model that does so), so $\Sigma_\CSyn\LSyn$ and $\Sigma_\CSyn\LSyn^{op}$ have finite products.
By the universal property of the source language $\Syn$, we obtain a canonical definition of CHAD.
\begin{thmx}[CHAD from a universal property, Cor. \ref{cor:chad-definition}] \label{thm:chad-universal-property}
Forward and reverse mode CHAD are the unique structure-preserving functors
\begin{align*}
    &\Dsyn{-}:\Syn\to \Sigma_{\CSyn}\LSyn &&\Dsynrev{-}:\Syn\to \Sigma_{\CSyn}\LSyn^{op}
\end{align*}
from the syntactic category $\Syn$ of the source language to the (opposite) Grothendieck construction of the target language $\LSyn:\CSyn^{op}\to \Cat$
that send primitive operations $\op$ to their derivative $D\op$ and transposed 
derivative $\transpose{D\op}$, respectively.
\end{thmx}
\noindent The definitions that follow from this universal property reproduce the definitions of CHAD that we have given so far.
Intuitively, the linear types represent commutative monoids, implementing the idea that (transposed) derivatives are linear functions in the sense that $Df(x)(0)=0$
and $Df(x)(v+v')=Df(x)(v)+Df(x)(v')$.
We have seen that this commutative monoid structure is important when writing down the definitions of AD as a source-code transformation.

Since a higher-order language can be viewed as a freely generated \emph{Cartesian closed} category $\Syn$, it is tempting to find a suitable target language such that 
$\Sigma_\CSyn\LSyn$ and $\Sigma_\CSyn\LSyn^{op}$ are Cartesian closed.
Then, we can \emph{define} CHAD on this higher-order language via Thm. \ref{thm:chad-universal-property}.
\begin{insight}
    \textit{
            To understand how to perform CHAD on a source 
            language with a language feature $X$ (e.g., higher-order functions),
            we need to understand the categorical semantics of language feature $X$ (e.g., categorical exponentials) in 
            categories of the form $\Sigma_\catC\catL$ and $\Sigma_\catC\catL^{op}$.
            Giving sufficient conditions on a model of linear logic $\catL:\catC^{op}\to\Cat$ for such a semantics to exist yields a 
            suitable target language for CHAD as the initial such model $\LSyn:\CSyn^{op}\to\Cat$, with the definition of the algorithm 
            following from the universal property of the source language.
    }
\end{insight}

\subsection{Cartesian Closure of $\Sigma_\catC\catL$ and $\Sigma_\catC\catL^{op}$
and CHAD of Higher-Order Functions}
With this insight, we identify conditions on a locally indexed category $\catL:\catC^{op}\to \Cat$ that are enough to guarantee that $\Sigma_\catC\catL$ and $\Sigma_\catC\catL^{op}$ are Cartesian closed (see \S\ref{sec:self-dualization}).
\begin{thmx}[Cartesian Closure of $\Sigma_\catC\catL$ and $\Sigma_\catC\catL^{op}$, Thm. \ref{thm:grothendieck-ccc-covariant}, \ref{thm:grothendieck-ccc-contravariant}]\label{thm:ccc-groth}
Suppose that a locally indexed category $\catL:\catC^{op}\to \Cat$
supports (we are intentionally a bit vague here for the sake of legibility)
\begin{itemize}
    \item linear $\copower{(-)}{(-)}$-types (copowers);
    \item linear $(-)\Rightarrow (-)$-types (powers);
    \item Cartesian $(-)\multimap(-)$-types (types of linear functions);
    \item linear biproduct types (or equivalently, linear (additive) product types and enrichment of $\catL$ over commutative monoids);
    \item Cartesian tuple and function types. 
\end{itemize}
Then, $\Sigma_\catC\catL$ and $\Sigma_\catC\catL^{op}$ are Cartesian closed with, respectively, exponentials:
\begin{align*}
(A_1,A_2)\Rightarrow_{\Sigma_\catC\catL} (B_1, B_2) = (A_1\Rightarrow B_1\times (A_2\multimap B_2), A_1\Rightarrow B_2) \\
(A_1,A_2)\Rightarrow_{\Sigma_\catC\catL^{op}} (B_1, B_2) = (A_1\Rightarrow B_1\times (B_2\multimap A_2), \copower{A_1}{ B_2}).
\end{align*}
\end{thmx}
In particular, if we extend our target language with (linear) powers, (linear) copowers and (Cartesian) function types, then $\LSyn:\CSyn^{op}\to \Cat$ satisfies the conditions of Thm. \ref{thm:ccc-groth}, so we can extend Thm.
\ref{thm:chad-universal-property} to our higher-order source language.
In particular, we find the following definitions of CHAD for primals and (co)tangents to function types:
\begin{align*}
&\Dsyn{\ty\To\ty[2]}_1\defeq \Dsyn{\ty}_1\To \Dsyn{\ty[2]}_1\t*(\Dsyn{\ty}_2\multimap \Dsyn{\ty[2]}_2)
&& \Dsyn{\ty\To\ty[2]}_2\defeq\Dsyn{\ty}_1\To\Dsyn{\ty[2]}_2\\
&\Dsynrev{\ty\To\ty[2]}_1\defeq \Dsynrev{\ty}_1\To \Dsynrev{\ty[2]}_1\t*(\Dsynrev{\ty[2]}_2\multimap \Dsynrev{\ty}_2)
&& \Dsynrev{\ty\To\ty[2]}_2\defeq\copower{\Dsynrev{\ty}_1}{\Dsynrev{\ty[2]}_2}.
\end{align*}
Interestingly, we see that for higher-order programs, the primal transformations are no longer the identity.
Indeed, the primals $\Dsyn{\ty\To\ty[2]}_1$ and $\Dsynrev{\ty\To\ty[2]}_1$ of the function type $\ty\To\ty[2]$ store not only the primal function itself, but also its derivative with respect to its argument.
The other half of a function's derivative, namely the derivative with respect to the context variables over which it closes, is stored in the tangent space $\Dsyn{\ty\To\ty[2]}_2$ and cotangent space $\Dsynrev{\ty\To\ty[2]}_2$ of the function type $\ty\To\ty[2]$.
\begin{insight}
\textit{A forward (respectively, reverse) mode primal to a function type $\ty\to\ty[2]$ keeps 
track of both the function and its derivative with respect to its argument (respectively, transposed derivative).
For reverse AD, a cotangent at function type $\ty\to\ty[2]$ (to be propagated back to the enclosing context of the function) 
keeps track of the incoming cotangents $v$ of type $\Dsynrev{\ty[2]}_2$
for each primal $x$ of type $\Dsynrev{\ty}_1$ on which we call the function.
We store these pairs $(x,v)$ in the type $\copower{\Dsynrev{\ty}_1}{\Dsynrev{\ty[2]}_2}$
(which we will see is essentially a quotient of a list of
pairs of type $\Dsynrev{\ty}_1\t*\Dsynrev{\ty[2]}_2$).
Less surprisingly, for forward AD, a tangent at function type $\ty\to\ty[2]$ (propagated forward from the enclosing context of the function)
consists of a function
sending each argument primal of type $\Dsyn{\ty}_1$ to the outgoing tangent of type $\Dsyn{\ty[2]}_2$.}
\end{insight}
On programs, we obtain the following extensions of our definitions for reverse AD:
\begin{flalign*}&   
\Dsynrev[\vGamma]{\fun\var\trm} &\hspace{-12pt}\phantom{.}&\defeq  \;
\letin{\var[2]}{\fun\var\Dsynrev[\vGamma,\var]{\trm}}{
\tPair{\fun\var
\pletin{\var[3]}{\var[3]'}{\var[2]\,\var}{
\tPair{\var[3]}{\lfun\lvar\tSnd(\lapp{\var[3]'}{\lvar})}}}
{\lfun\lvar  \tensMatch{\lvar}{\var}{\lvar}{
    \tFst(\lapp{(\tSnd(\var[2]\,\var))}{\lvar})} }
}\\
&\Dsynrev[\vGamma]{\trm\,\trm[2]}&\hspace{-12pt}\phantom{.}&\defeq  \;
\pletin{\var}{\var'_{\text{ctx}}}{\Dsynrev[\vGamma]{\trm}}{
\pletin{\var[2]}{\var[2]'}{\Dsynrev[\vGamma]{\trm[2]}}{
\pletin{\var[3]}{\var'_{\text{arg}}}{\var\,\var[2]}{
\tPair{\var[3]}{\lfun\lvar  \lapp{\var'_{\text{ctx}}}{(!\var[2]\otimes \lvar)} + \lapp{\var[2]'}{(\lapp{\var'_{\text{arg}}}{\lvar})}}}}}.
\end{flalign*}
Regarding $\Dsynrev[\vGamma]{\fun\var\trm}$: suppose that $(\fun\var\trm) : \ty\To\ty[2]$.
Note then that we have $\Gamma, \var:\ty \vdash \trm:\ty[2]$ and hence $\trm$'s derivative has type $\Dsynrev{\Gamma}_1, \var:\Dsynrev{\ty}_1 \vdash \Dsynrev[\vGamma,\var]{\trm} : \Dsynrev{\ty[2]}_1 \t* (\Dsynrev{\ty[2]}_2 \multimap \Dsynrev{\Gamma}_2 \t* \Dsynrev{\ty}_2)$.
Calling the transposed derivative function for $\trm$ ($\var[3]'$ in the primal, $\tSnd(\var[2]\,\var)$ in the dual) therefore gives us \emph{both} halves of the transposed derivative (the derivative with respect to the function argument and the context variables, that is) of the function; we then select the appropriate components using projections.
Similarly, in $\Dsynrev[\vGamma]{\trm\,\trm[2]}$ we extract the transposed derivative 
$\var'_{\text{ctx}}$ of $\trm$ with respect to the context variables from the cotangent of $\trm$ and obtain the transposed derivative $\var'_{\text{arg}}$ of $\trm$ with respect to its function argument from $\trm$'s primal.
We combine these two halves of the transposed derivative with $\trm[2]$'s transposed derivative (which we get from its cotangent) to get the correct transposed derivative for the function application $\trm\,\trm[2]$.

\subsection{Proving CHAD Correct}\label{ssec:key-ideas-chad-correct}
With these definitions in place, we turn to the correctness of the source-code transformations.
To phrase correctness, we first need to construct a suitable semantics with an uncontroversial 
notion of semantic differentiation (see \S\ref{sec:semantics}).
We choose to work with a semantics in terms of the category $\Set$ of sets and functions\footnote{In \cite{vakar2020reverse}, we worked with a semantics in terms of diffeological spaces 
and \smooth{} functions, instead, to ensure that any first-order function 
is \smooth{}.
This choice separated the proof that every first-order denotation is \smooth{} from the proof that AD computes the correct derivative.
To make the presentation of this paper more accessible, we have chosen simply to work with sets and 
functions, and to prove differentiability of every first-order denotation
simultaneously with the proof that AD computes the correct derivative.}, noting that any function 
$f:\RR^n\to \RR^m$ has a unique derivative as long as $f$ is \smooth{}.
We will only be interested in this semantic notion of derivative of first-order functions for the sake 
of correctness of AD, and we will not concern ourselves with semantic derivatives of higher-order functions.
We interpret the required linear types in the category $\CMon$ of commutative monoids and homomorphisms.

By the universal properties of the syntax, we obtain canonical, structure-preserving (homomorphic) functors $\sem{-}:\Syn\to\Set$, $\sem{-}:\CSyn\to\Set$ and
$\sem{-}:\LSyn\to\CMon$ once we fix interpretations $\RR^n$ of $\reals^n$ 
and well-typed (\smooth{}) interpretations $\sem{\op}$ for each operation $\op$.
These functors define a concrete denotational semantics for our source and target languages.

Having constructed the semantics, we can turn to the correctness proof (of \S\ref{sec:glueing-correctness}).
Because calculus does not provide an unambiguous notion of derivative 
at function spaces, we cannot prove that the AD transformations
correctly implement mathematical derivatives by straightforward induction on the syntax.
Instead, we use a logical relations argument over the semantics.
\begin{insight}
\emph{Once we show that the (transposed) derivatives of primitive operations $\op$
are correctly implemented, correctness of (transposed) derivatives of all other programs
follows from a standard logical relations construction over the semantics
that relates a curve to its primal and (co)tangent curve.
By the chain rule for (transposed) derivatives, all CHAD-transformed programs respect the logical relations.
By basic calculus results, CHAD therefore must compute the (transposed) derivative.
}
\end{insight}
In \S\ref{sec:glueing-correctness}, we present an elegant high-level formulation of this correctness argument, using categorical logical relations techniques (subsconing).
To make this argument accessible to a wider audience of readers, we present a low-level description of the logical relations argument here.
The reader may note that these arguments look significantly different from the usual definitions of logical relations.
That difference is caused by the non-standard Cartesian closed structure of $\Sigma_\catC \catL$ and $\Sigma_\catC \catL^{op}$ and the proof is entirely 
standard when viewed from the higher level of abstraction that subsconing gives us.

We first sketch the correctness argument for forward mode CHAD.
By induction on the structure of types, writing $(f, f')$ for the product pairing of $f$ and $f'$, we construct a logical relation $P_{\ty}$ on types $\ty$ as
\begin{align*}
P_{\ty}&\subseteq (\RR^d\Rightarrow\sem{\ty})\times (\RR^d\Rightarrow(\sem{\Dsyn{\ty}_1}\times (\cRR^d\multimap\sem{\Dsyn{\ty}_2})))\\
P_{\reals^n} &\defeq \set{(f, g)\mid f\textnormal{ is \smooth{} and } g=\Dsem{f}}
\\
P_{\Unit} &\defeq \set{(x\mapsto (),x\mapsto ((), r\mapsto ()))}\\
P_{\ty\t*\ty[2]}&\defeq \set{((\sPair{f}{f'},(\sPair{g}{g'},x\mapsto r\mapsto \sPair{h(x)(r)}{h'(x)(r)})))\mid (f,(g,h))\in 
P_{\ty}, (f',(g',h'))\in P_{\ty[2]}}\\
P_{\ty\To\ty[2]}&\defeq \big\{(f,(g,h))\mid \forall (f',(g',h'))\in P_{\ty}.
(x\mapsto f(x)(f'(x)), (
x\mapsto \pi_1 (g(x)(g'(x))),\\
&\phantom{big\{(f,(g,h))\mid \forall (f',(g',h'))\in P_{\ty}.\quad} {x}\mapsto {r}\mapsto (\pi_2 (g(x)(g'(x))))(h'(x)(r)) + h(x)(r)(g'(x))))\in P_{\ty[2]}
\big\}.
\end{align*}
We extend the logical relation to typing contexts $\Gamma=\var_1:\ty_1,\ldots,\var_n:\ty_n$ as $P_\Gamma\defeq P_{\ty_1\t*\cdots\t*\ty_n}$.
Then, we establish the following fundamental lemma, which says that all well-typed source language programs $\trm$ respect 
the logical relation.
\begin{lemma}\label{lem:fundamental-fwd}
For any source language program $\Gamma\vdash \trm:\ty[2]$ and any $f:\RR^d\to \sem{\Gamma}$, $g:\RR^d\to\sem{\Dsyn{\Gamma}_1}$,
    $h:\RR^d\to \cRR^d\multimap \sem{\Dsyn{\Gamma}_2}$ such that $(f,(g,h))\in P_{\Gamma}$,
    we have that 
    $(f;\sem{\trm},
    (g; \sem{\Dsyn[\vGamma]{\trm}} ; \pi_1, x\mapsto r\mapsto \pi_2(\sem{\Dsyn[\vGamma]{\trm}}(g(x)))(h(x)(r))
    ))\in P_{\ty[2]}$.
\end{lemma}
\noindent  The proof proceeds by induction on the typing derivation of $\trm$.
The main remaining step in the argument is to note that any tangent vector at $\sem{\ty_1\t*\cdots\t*\ty_n}\cong\RR^N$,
for first-order $\ty_i$, can be 
represented by a curve $\RR\to \sem{\ty_1\t*\cdots\t*\ty_n}$.

Similarly, for reverse mode CHAD,
we define, by induction on the structure of types, a logical relation $P_{\ty}$ on types $\ty$ (and, as before, we also define $P_{\Gamma}=P_{\ty_1\t*\cdots\t*\ty_n}$ for typing contexts $\Gamma=\var_1:\ty_1,\ldots,\var_n:\ty_n$):
\begin{align*}
    P_{\ty}&\subseteq (\RR^d\To\sem{\ty})\times (\RR^d\To(\sem{\Dsynrev{\ty}_1}\times (\sem{\Dsynrev{\ty}_2}\multimap \cRR^d)))\\
P_{\reals^n} &\defeq \big\{(f, g)\mid f\textnormal{ is \smooth{} and }g=\Dsemrev{f}\big\}\\
P_{\Unit}&\defeq \set{(x\mapsto (),x\mapsto ((),v\mapsto 0))}\\
P_{\ty\t*\ty[2]}&\defeq \set{((\sPair{f}{f'},(\sPair{g}{g'},{x}\mapsto {v}\mapsto h(x)(\pi_1 v)+h'(x)(\pi_2 v))))\mid (f,(g,h))\in 
P_{\ty}, (f',(g',h'))\in P_{\ty[2]}}\\
P_{\ty\To\ty[2]}&\defeq \big\{(f,(g,h))\mid 
\forall (f',(g',h'))\in P_{\ty}.
(x\mapsto f(x)(f'(x)), (
x\mapsto \pi_1 (g(x)(g'(x))),\\ &\phantom{\defeq \big\{(f,(g,h))\mid 
\forall (f',(g',h'))\in P_{\ty}.\quad} {x}\mapsto {v}\mapsto h({x})(\copower{g'(x)}{v})
+h'(x)((\pi_2(g(x)(g'(x))))v))\in P_{\ty[2]}
\big\}
\end{align*}
Then, we establish the following fundamental lemma.
\begin{lemma}\label{lem:fundamental-rev}
For any source language program $\Gamma\vdash \trm :\ty[2]$ and any $f:\RR^d\to \sem{\Gamma}$, $g:\RR^d\to\sem{\Dsynrev{\Gamma}_1}$,
$h:\RR^d\to \sem{\Dsynrev{\Gamma}_2}\multimap\cRR^d$ such that $(f,(g,h))\in P_{\Gamma}$,
we have that
$(f;\sem{\trm},
(g; \sem{\Dsynrev[\vGamma]{\trm}} ; \pi_1,{x}\mapsto {v}\mapsto
h(x)( \pi_2(\sem{\Dsynrev[\vGamma]{\trm}}(g(x)))( v))))\in P_{\ty[2]}$.
\end{lemma}
\noindent The proof proceeds by induction on the typing derivation of $\trm$.
Correctness follows from the fundamental lemma by observing that any tangent vector at $\sem{\ty_1\t*\cdots\t*\ty_n}\cong\RR^N$,
for first-order $\ty_i$, can be 
represented by a curve $\RR\to \sem{\ty_1\t*\cdots\t*\ty_n}$.

We obtain our main theorem, Thm. \ref{thm:chad-correctness}, but now for our CHAD algorithms applied to a higher-order source language.

\subsection{A Practical Implementation in Haskell}
Next, we address the practicality of our method (in \S\ref{sec:implementation}).
The code transformations we employ are straightforward to implement and they are well-behaved in the sense that the derivative code they generate grows linearly in the size of the original source code.
However, the implementation of the required linear types presents a challenge.
Indeed, types such as $\copower{(-)}{(-)}$ and $(-)\multimap (-)$ are absent from
 languages such as Haskell and OCaml.
Fortunately, in this instance, we can implement them using abstract data types 
and a basic module system:
\begin{insight}\emph{
Under the hood, $\copower{\ty}{\ty[2]}$ can consist of a list of values of 
type $\ty\t*\ty[2]$.
Its API ensures that the list order and the difference between
$\PlusList{xs}{[(\trm, \trm[2]), (\trm, \trm[2]')]}$ and
$\PlusList{xs}{[(\trm, \trm[2]+\trm[2]')]}$ (or $xs$ and $\PlusList{xs}{[(\trm, 0)]}$)
cannot be observed; in this sense, 
it is a quotient type.
Meanwhile, $\ty\multimap \ty[2]$  
can be implemented as a standard function 
type $\ty\To\ty[2]$ with a limited API that enforces that we can 
only ever construct linear functions; in this sense, it is a subtype.}
\end{insight}
This idea leads to our reference implementation of CHAD in Haskell (available at \href{https://github.com/VMatthijs/CHAD}{https://github.com/VMatthijs/CHAD}),
which generates perfectly standard simply typed functional code that is given extra type safety by the linear types, implemented as abstract types.
To illustrate what our method does in practice, we consider two programs 
of our higher-order source language, shown in Fig. \ref{fig:higher-order-ad-example} (a) and (b), that we may want to differentiate.
The forward and reverse mode derivatives that our CHAD implementation generates for these programs are listed in Fig. \ref{fig:higher-order-ad-example} (c) and (d),
again modulo minor simplifications that aid legibility but have no significant run-time implications.\footnote{For information on the exact simplifications performed, see Appendix \ref{appx:impl-simpl}.}

\begin{figure}[!t]
    \begin{subfigure}[b]{0.45\linewidth}
\begin{lstlisting}[mathescape=true]
let f  = $\lambda$z. x * z + 1
    zs = replicate x
    ys = map f zs in
    ys     
\end{lstlisting}
\subcaption{Original program $\var:\reals\vdash \trm:\reals^n$ computing a function $\sem{\trm}:\RR\to \RR^n; x\mapsto ys$.
\texttt{replicate x} copies \texttt{x} $n$ times.}
\end{subfigure}\hspace{0.08\linewidth}
\begin{subfigure}[b]{0.45\linewidth}
\begin{lstlisting}[mathescape=true]
let f  = $\lambda$x2i. x1 * x2i
    ys = map f x2
    w  = sum ys in
    w
\end{lstlisting}
\subcaption{Original program $\var1:\reals,\var2:\reals^n\vdash \trm[2]:\reals$ computing a function $\sem{\trm[2]}:\RR\times \RR^n\to\RR; (x1,x2)\mapsto w$.}
    \end{subfigure}
\\
\quad\\
\begin{subfigure}[b]{\linewidth}

\begin{tabular}{ll}
\begin{lstlisting}[mathescape=true]
let [*f   = $\color{blue}\lambda$z.$\color{blue}\langle$x * z + 1,$\color{blue}\lambda$z'.x * z'$\color{blue}\rangle$*]
    zs  = replicate x
    [*ys  = map ($\color{blue}\lambda$z.*]<*fst*>[*$\color{blue}$(f z)) zs*] in
    $\langle$ys, $\lambda$x'.
    let [*f'  = $\color{blue}\lambda$z.x' * z*]
        zs' = replicate x'
        [*ys' = map f' zs + zipWith ($\color{blue}\lambda$z.*]<*snd*>[*(f z)) zs zs'*] in 
        ys'$\rangle$
\end{lstlisting}&\hspace{-12pt}
\begin{tabular}{l}
    \ntikzmark{S}{ }
    \\
    \\
    \ntikzmark{E}{ }\makebrace{S}{E}{primals}\\
    \\
    \ntikzmark{S}{ }
    \\
    \\
    \\
    \ntikzmark{E}{ }\makebrace{S}{E}{tangents}\\
\end{tabular}\hspace{-30pt}
\end{tabular}
    \subcaption{Forward AD transformed program $x:\reals\vdash \Dsyn{\trm}:\reals^n\t*(\creals\multimap \creals^n)$ computing the derivative
    $\sem{\Dsyn{\trm}}=\Dsem{\sem{\trm}}:\RR\to \RR^n\times (\cRR\multimap \cRR^n); x\mapsto (ys, x' \mapsto ys')$.}
    \end{subfigure}\\
    \quad\\
\begin{subfigure}[b]{\linewidth}
    \begin{tabular}{ll}
\begin{lstlisting}[mathescape=true]
let [*f   = $\color{blue}\lambda$x2i.$\color{blue}\langle$x1 * x2i,$\color{blue}\lambda$y'.x1 * y'$\color{blue}\rangle$*]
    [*ys  = map ($\color{blue}\lambda$x2i.*]<*fst*>[*(f x2i)) x2*]
    w   = sum ys in
    $\langle$w, $\lambda$w'.
    let ys' = replicate w'
        [*f'  = zip x2 ys'*]
        [*x1' = sum (map ($\color{blue}\lambda\langle$x2i,y'$\color{blue}\rangle$.y' * x2i) f')*]
        [*x2' = zipWith ($\color{blue}\lambda$x2i.*]<*snd*>[*(f x2i)) x2 ys'*] in
        $\langle$x1',x2'$\rangle\rangle$
\end{lstlisting} &\hspace{-12pt}   \begin{tabular}{l}
    \ntikzmark{S}{ }
    \\
    \\
    \ntikzmark{E}{ }\makebrace{S}{E}{primals}\\
    \\
    \ntikzmark{S}{ }
    \\
    \\
    \\ 
    \\
    \ntikzmark{E}{ }\makebrace{S}{E}{cotangents}\\
\end{tabular}\hspace{-30pt}
\end{tabular}
\subcaption{Reverse AD transformed program $x1:\reals,x2:\reals^n\vdash \Dsynrev{\trm[2]}:\reals\t* (\creals\multimap \creals\t*\creals^n)$ computing the transposed derivative
$\sem{\Dsynrev{\trm[2]}}=\Dsemrev{\sem{\trm[2]}}:\RR\times \RR^n\to \RR\times (\cRR\multimap \cRR\times \cRR^n); (x1,x2)\mapsto (w,w' \mapsto (x1',x2'))$.}
\end{subfigure}
\caption{\label{fig:higher-order-ad-example}Forward and reverse AD illustrated on higher-order functional array processing programs.
The parts of the programs that involve AD on higher-order functions are marked 
in blue.
Observe that in (c) (respectively, (d)), the primal value associated with the function 
$f$ from program (a) (respectively, (b)) computes both the original function $f$ and its 
derivative (respectively, transposed derivative) with respect to its argument $z$ (respectively, $x2i$).
In (c), the tangent $f'$ to $f$ is \emph{produced} by propagating \emph{forward} the tangent $x'$ to the context variable $x$
that $f$ captures, using $f$'s derivative with respect to $x$.
This lets us correctly propagate forward the contributions to $ys'$ from $f$'s dependence both on its argument $z$ and on its context variable $x$.
Dually, in (d), the cotangent to $f$, which we construct from the cotangent $ys'$, is \emph{consumed} by propagating it \emph{backward} to the cotangent $x1'$ of the context variable $x1$ that $f$ captures, using $f$'s transposed derivative with respect to $x1$. Meanwhile, the adjoint $x2'$ is constructed using the part of the primal of $f$ that captures $f$'s transposed derivative with respect to $x2i$.
}
\end{figure}

In \S\ref{sec:implementation}, we also phrase the correctness proof of the AD transformations in elementary terms,
so that it holds in the applied setting where we use abstract types to implement linear types.
We show that our correctness results are meaningful because they are stated using a denotational semantics 
that is adequate with respect to the standard operational semantics.
Furthermore, to stress the applicability of our method, we show in \S\ref{sec:higher-order-primitives} that it extends to 
higher-order (primitive) operations, such as $\tMap$.

Finally, in \S\ref{sec:scope}, we zoom out and reflect on how this method generalizes.
The crux of CHAD lies in the following~steps:
\begin{itemize}
\item view the source language as a freely generated category $\Syn$ with some appropriate structure $\catS$ (such as Cartesian closure, coproducts, (co)inductive types, iteration), generated from 
objects $\reals^n$ and morphisms $\op$;
\item find a suitable target language $\LSyn$ (with linear types arising from the effect 
of commutative monoids) for the translation such that
$\Sigma_\CSyn \LSyn$ and $\Sigma_\CSyn \LSyn^{op}$ are categories with the structure $\catS$;
in our experience, this is possible for most common choices of $\catS$ corresponding 
to programming language constructs;
\item then, by the universal property of $\Syn$, we obtain unique structure-preserving (homomorphic) functors
$\Dsynsymbol : \Syn\to\Sigma_\CSyn \LSyn$ and $\Dsynrevsymbol:\Syn\to\Sigma_\CSyn \LSyn^{op}$
defining forward and reverse mode AD transformations, as soon as we fix their action on $\op$ (and $\reals^n$) to implement the derivative of the operations;
\item the correctness of these AD methods follows by a standard 
categorical logical relations argument as the subscones $\Gl$ and $\GlRev$ 
also tend to be categories with the structure $\catS$ for most choices of $\catS$.
\end{itemize}
\begin{insight}\emph{
The definition and correctness proof of forward and reverse AD on expressive programming languages follow automatically, by viewing the algorithms as structure-preserving functors $\Dsynsymbol : \Syn\to\Sigma_\CSyn \LSyn$ and $\Dsynrevsymbol:\Syn\to\Sigma_\CSyn \LSyn^{op}$.}
\end{insight}
We conclude by observing that, in this sense, CHAD is not specific to 
automatic differentiation at all.
We can choose generators other than $\reals^n$ and $\op$ for $\Syn$ and 
 different mappings of these generators under $\Dsynsymbol$ and $\Dsynrevsymbol$. Doing so lets CHAD derive various other dynamic program analyses that accumulate data in a commutative monoid, together with their correctness proofs by logical relations (see \S\S\ref{ssec:other-analyses}).

 \section{$\lambda$-Calculus as a Source Language for Automatic Differentiation}\label{sec:language}
As a source language for our AD translations,
we can begin with a standard, simply typed $\lambda$-calculus that has 
ground types $\reals^n$ of statically sized\footnote{Here, we work with statically sized arrays to simplify the theoretical development. However, in our implementation, we show that CHAD applies equally well to types of varying dimension such as dynamically sized arrays.} arrays of $n$ real numbers, for all $n\in\NN$,
and sets $\Op_{n_1,...,n_k}^m$ of primitive operations $\op$
for all $k, m, n_1,\ldots, n_k\in\NN$. These operations will be interpreted
as \smooth{}\footnote{Observe that this restriction does not meaningfully exclude functions that are 
\smooth{} almost everywhere, such as ReLU, because such functions can be approximated by 
\smooth{} functions.
Given how coarse an approximation real numbers already are to floating-point arithmetic,
the distinction between everywhere \smooth{} and almost-everywhere \smooth{} is not meaningful in practice.
} functions $(\RR^{n_1}\times\ldots\times\RR^{n_k})\To\RR^m$.
Examples to keep in mind for $\op$ include
\begin{itemize}
  \item constants $\cnst\in \Op_{}^n$ for each $c\in \RR^n$, for which 
   we slightly abuse notation and write $\cnst()$ as $\cnst$;
    \item elementwise addition and product $(+),(*)\!\in\!\Op_{n,n}^n$
    and matrix-vector product $(\star)\!\in\!\Op_{n\cdot m, m}^n$;
    \item operations for summing all the elements in an array: $\tSum\in\Op_{n}^1$;
    \item some non-linear functions such as the sigmoid function $\sigmoid\in \Op_{1}^1$.
\end{itemize}
We intentionally present operations in a schematic way, as primitive operations tend to form a collection that grows as needed as an AD library develops.
The precise operations needed will depend on the applications.
In statistics and machine learning applications, $\Op$
tends to include
a mix of multi-dimensional linear algebra operations and mostly one-dimensional non-linear 
functions.
A typical library for use in machine learning would work with multi-dimensional arrays 
(sometimes called ``tensors'').
We focus here on one-dimensional arrays because the details of how to represent arrays
are orthogonal to the concerns of our development.

The types $\ty,\ty[2],\ty[3]$ and terms $\trm,\trm[2],\trm[3]$ of our AD source language are as follows:

\noindent\begin{syntax}
    \ty, \ty[2], \ty[3] & \gdefinedby & & \syncat{types}                          \\
    &\gor& \reals^n                      & \synname{real arrays}\\
    &\gor\quad\, & \Unit & \synname{nullary product}\\
    &&&\\
    \trm, \trm[2], \trm[3] & \gdefinedby & & \syncat{terms}    \\
    &    & \var                          & \synname{identifier} \\
    & \gor & \letin{\var}{\trm}{\trm[2]} & \synname{$\mathbf{let}$-bindings}\\
    &\gor& \op(\trm_1,\ldots,\trm_k)
& \synname{operations}                      \\
\end{syntax}\qquad\qquad\qquad
  \begin{syntax}
    &\gor\quad\,& \ty_1\t* \ty_2 & \synname{binary product} \\
  &\gor& \ty \To \ty[2]              & \synname{function}      \\
  &&&\\
  &&&\\
  &\gor& \tUnit\ \gor \tPair{\trm}{\trm[2]} & \synname{product tuples}\\& \gor & \tFst{\trm}\ \gor\tSnd{\trm}& \synname{product projections}\\
&\gor& \fun \var    \trm &\synname{function abstraction}\\
&\gor &\trm\,\trm[2] & \synname{function application}
\end{syntax} \\ 

The typing rules are in Fig.~\ref{fig:types1}.
We use the usual conventions for free and bound variables and write 
the capture-avoiding substitution of $\var$ with $\trm[2]$ in $\trm$ as $\subst{\trm}{\sfor{\var}{\trm[2]}}$.
We employ the usual syntactic sugar
$\fun{\tPair{\var}{\var[2]}}\trm\defeq \fun{\var[3]}{\subst{\trm}{\sfor{\var}{\tFst\var[3]},\sfor{\var[2]}{\tSnd\var[3]}}}$,
and we write $\reals$ for $\reals^1$.
\begin{figure}[!t]
\framebox{\begin{minipage}{0.98\linewidth}\noindent\hspace{-24pt}\[
  \begin{array}{c}
    \inferrule{
      ((\var : \ty) \in \ctx)
    }{
      \Ginf \var\ty
    }
    \qquad 
    \inferrule{\Ginf{\trm}{\ty}\quad \Ginf[,\var:\ty]{\trm[2]}{\ty[2]}}
    {\Ginf {\letin{\var}{\trm}{\trm[2]}}{\ty[2]}}\qquad

  \inferrule{\set{
    \Ginf {\trm_i} {\reals^{n_i}}\mid i=1,\ldots,k}\quad
     (\op\in\Op^m_{n_1,\ldots,n_k})
  }{
    \Ginf {\op(\trm_1,\ldots,\trm_k)}{\reals^m
  }}\qquad
\inferrule{
  ~
}{
  \Ginf{\tUnit}{\Unit}
}\\\\
  \inferrule{
    \Ginf {\trm}{\ty}\quad
    \Ginf {\trm[2]}{\ty[2]}
  }{
    \Ginf{\tPair{\trm}{\trm[2]}}{\ty\t* \ty[2] }
  }
  \qquad
    \inferrule{
    \Ginf{\trm}{\ty\t*\ty[2] }
  }{
    \Ginf{\tFst\trm}{\ty}
  }
  \qquad  \inferrule{
    \Ginf{\trm}{\ty\t*\ty[2] }
  }{
    \Ginf{\tSnd\trm}{\ty[2]}
  }
 \qquad

  \inferrule{
    \Ginf[, \var : \ty]{\trm}{\ty[2]}
  }{
    \Ginf{\fun{\var}\trm}{\ty\To\ty[2]}
  }
\qquad
  \inferrule{
    \Ginf{\trm}{\ty[2]\To\ty}
    \\
    \Ginf{\trm[2]}{\ty[2]}
  }{
    \Ginf{\trm\, \trm[2]}{\ty}
  }
\end{array}
\]
 \end{minipage}}
    \caption{Typing rules for the AD source language.\label{fig:types1}}\;
  \end{figure}
As Fig. \ref{fig:beta-eta} displays, we consider the standard $\beta\eta$-equational theory for our language, 
where equations hold on pairs of terms of the same type in the same context.
\begin{figure}[!t]
  \framebox{\begin{minipage}{0.98\linewidth}\hspace{-24pt} \[\begin{array}{l}
    \letin{\var}{\trm}{\trm[2]} = \subst{\trm[2]}{\sfor{\var}{\trm}}\quad\;\,
    \trm =\tUnit\quad\;\,
    \tFst \tPair{\trm}{\trm[2]} = \trm \quad\;\,
    \tSnd \tPair{\trm}{\trm[2]} = \trm[2]\quad\;\,
  \trm = \tPair{\tFst\trm}{\tSnd\trm}\quad\;\,
  (\fun{\var}{\trm})\,\trm[2] = \subst{\trm}{\sfor{\var}{\trm[2]}} \quad\;\,
        \trm \freeeq{\var} \fun{\var}{\trm\,\var}
\end{array}\]

  \end{minipage}}
 \caption{Standard $\beta\eta$-laws for products and functions. We write $\freeeq{\var_1,\ldots,\var_n}$ to indicate that the variables $\var_1,\ldots,\var_n$ need to be fresh in the left-hand side.
 As usual, we only distinguish terms up to $\alpha$-renaming of bound variables.\label{fig:beta-eta}\;
}
 \end{figure}
We could consider further equations for our operations, but we do not, as we will not need~them.

This standard $\lambda$-calculus is widely known to be equivalent to the free 
Cartesian closed category $\Syn$ generated by the objects $\reals^n$ and the morphisms $\op$
(see \cite{lambek1988introduction}).
\begin{itemize}
  \item  $\Syn$ has types $\ty,\ty[2],\ty[3]$ as objects;
  \item  $\Syn$ has morphisms $\trm\in\Syn(\ty,\ty[2])$ that are in one-to-one correspondence with terms $\var:\ty\vdash \trm:\ty[2]$ up to 
  $\beta\eta$-equivalence (which includes $\alpha$-equivalence);
  \item identities are represented by $\var:\ty\vdash \var:\ty$;
  \item composition of $\var:\ty\vdash \trm:\ty[2]$ and $\var[2]:\ty[2]\vdash\trm[2]:\ty[3]$ is represented by $\var:\ty\vdash \letin{\var[2]}{\trm}{\trm[2]}:\ty[3]$;
  \item $\Unit$ and $\ty\t*\ty[2]$ represent nullary and binary products, while $\ty\To\ty[2]$ is the categorical exponential.
\end{itemize}
$\Syn$ has the following well-known universal property.
\begin{proposition}[Universal property of $\Syn$]\label{prop:universal-property-syn}
For any Cartesian closed category $(\catC,\terminal,\times,\Rightarrow)$,
we obtain a unique Cartesian closed functor $F:\Syn\to\catC$, once we choose objects $F(\reals^n)$ of $\catC$ and, for each $\op\in\Op^m_{n_1,\ldots,n_k}$, make
well-typed choices of $\catC$-morphisms\;\;
$
F(\op):(F(\reals^{n_1})\times \ldots\times F(\reals^{n_k}))\To F(\reals^m).
$
\end{proposition}
 \section{Linear $\lambda$-Calculus as an Idealised AD Target Language}\label{sec:minimal-linear-language} 
As a target language for our AD source-code transformations, 
we consider a language that extends the language of \S\ref{sec:language}
with limited linear types.
We could opt to work with a full linear logic as in \cite{barber1996dual} or 
\cite{benton1994mixed}.
Instead, however, we will only include the bare minimum of
linear type formers that we actually need to phrase the AD transformations.
The resulting language is closely related to, but more minimal than, the Enriched Effect Calculus 
of \cite{egger2009enriching}.
We limit our language in this way because we want to stress that the resulting 
code transformations can easily be implemented in existing functional languages 
such as Haskell or OCaml.
As we discuss in \S\ref{sec:implementation}, the idea is to use a module system
to implement the required linear types as abstract data types.

In our idealised target language, we consider \emph{linear types} (also known as computation types)
$\cty$, $\cty[2]$, $\cty[3]$,
in addition to the \emph{Cartesian types} (also known as value types) $\ty$, $\ty[2]$, $\ty[3]$ that
we have considered so far.
We think of Cartesian types as denoting sets and linear types as 
denoting sets equipped with an algebraic structure.
The Cartesian types will be used to represent sets of primals.
The relevant algebraic structure on linear types, in this instance, turns out to be 
that of a commutative monoid, as this algebraic structure is needed to 
formulate automatic differentiation algorithms.
Indeed, we will use the linear types to denote sets of (co)tangent vectors.
These (co)tangents form a commutative monoid under addition.

Concretely, we extend the types and terms of our language as follows:\\
 \begin{syntax}
    \cty, \cty[2], \cty[3] & \gdefinedby & &  \syncat{linear types}\\     
  &\gor & \creals^n & \synname{real array}\\
  & \gor & \lUnit & \synname{unit type}\\
  &&&\\
  \ty, \ty[2], \ty[3] & \gdefinedby & & \syncat{Cartesian types}               \\
  &\gor& \ldots                      & \synname{as in \S\ref{sec:language}}\\
  &&&\\
  \trm, \trm[2], \trm[3] & \gdefinedby  & & \syncat{terms}             \\
&\gor& \ldots                      & \synname{as in \S\ref{sec:language}} \\
& \gor & \lvar & \synname{linear identifier}\\
&\gor & \letin{\lvar}{\trm}{\trm[2]} & \synname{linear $\mathbf{let}$-binding}\\
\end{syntax}
~\qquad\qquad\!\!\!
\begin{syntax}
  & \gor & \cty \t* \cty[2] & \synname{binary product}\\
  & \gor & \ty \To \cty[2] & \synname{power}\\
  &\gor &\copower{\ty }{\cty[2]}& \synname{copower}\\
  &&&\\
  &\gor\quad\, & \cty\multimap \cty[2] & \synname{linear function}\\
  &&&\\
  &&&\\
&\gor & \lop(\trm_1,\ldots,\trm_k;\trm[2]) & \synname{linear operation}\\
  &\gor\;\; & \copower{\trm}{\trm[2]}\,\gor \tensMatch{\trm}{\var[2]}{\lvar}{\trm[2]} & \synname{copower intro/elim}\\
  &\gor & \lfun{\lvar}{\trm}\,\gor  \lapp\trm{\trm[2]}  & \synname{abstraction/application}\\
  &\gor & \zero\,\gor \trm+\trm[2] & \synname{monoid structure.}
\end{syntax}
 \\
 \\
We work with linear operations $\lop\in\LOp_{n_1,...,n_k;
n'_1,\ldots,n'_l}^{m_1,\ldots,m_r}$,
which are intended to represent \smooth{} functions
$$
(\RR^{n_1}\times \cdots\times \RR^{n_k}\times \RR^{n'_1}\times \cdots\times \RR^{n'_l})\to \RR^{m_1}\times\cdots\times\RR^{m_r}
$$
that are linear (in the sense of 
respecting $\zero$ and $+$) in the last $l$ 
arguments but not in the first $k$.
We 
 write $$
\LDomain{\lop}\defeq \creals^{n'_1}\t*\ldots\t* \creals^{n'_l}
\qquad\text{and}\qquad
\CDomain{\lop}\defeq \creals^{m_1}\t*\ldots\t* \creals^{m_r}$$
for $\lop\in\LOp_{n_1,...,n_k;
n'_1,\ldots,n'_l}^{m_1,\ldots,m_r}$.
These operations can include dense and sparse matrix-vector multiplications, for example.
Their purpose is to serve as primitives to
implement derivatives~$D\op(\var_1,\ldots,\var_k;\lvar)$ and transposed derivatives
$ \transpose{D\op}(\var_1,\ldots,\var_k;\lvar)$ 
of the
operations $\op$ from the source language as terms with free variables $\var_1,\ldots,\var_k,\lvar$ that are linear in $\lvar$.
In fact, one can also opt to directly include, in $\LOp$, primitive linear operations for the derivatives of each (Cartesian) operation $\op\in\Op$:
\begin{align*}
  &D\op\in \LOp_{n_1,...,n_k;n_1,\ldots,n_k}^m
  &\transpose{D\op}\in \LOp_{n_1,...,n_k;m}^{n_1,\ldots,n_k}.
  \end{align*}

In addition to the judgement $\Ginf \trm\ty$, which we encountered in \S\ref{sec:language}, we now consider an additional judgement 
$\Ginf[;\lvar:\cty]\trm{\cty[2]}$.
While we think of the former as denoting a function 
between sets, we think of the latter as a function 
from the set that $\Gamma$ denotes to the set of 
monoid homomorphisms from the denotation of $\cty$ to that of $\cty[2]$.

Figs.~\ref{fig:types1} and~\ref{fig:minimal-linear-types} display the typing rules of our language.
\begin{figure}[!t]
\fbox{\begin{minipage}{0.98\linewidth}\hspace{-24pt}\noindent\[
  \begin{array}{c}
    \inferrule{~}{\Ginf[;\lvar:\cty]{\lvar}{\cty}}
    \qquad 
    \inferrule{\Ginf{\trm}{\ty}\quad 
    \Ginf[,\var:\ty;\lvar:{\cty[2]}]{\trm[2]}{\cty[3]}}
    {\Ginf[;\lvar:{\cty[2]}]{\letin{\var}{\trm}{\trm[2]}}{\cty[3]}}\qquad
    \inferrule{\Ginf[;\lvar:\cty]{\trm}{\cty[2]}\quad 
    \Ginf[;\lvar:{\cty[2]}]{\trm[2]}{\cty[3]}}{
      \Ginf[;\lvar:\cty]{\letin{\lvar}{\trm}{\trm[2]}}{\cty[3]}}
    \\
    \\ 
    \inferrule{
        \set{\Ginf {\trm_i} {\reals^{n_i}}\mid i=1,\ldots,k}
        \quad
        \Ginf[;\lvar:\cty]{\trm[2]}{\LDomain{\lop}}\quad 
         (\lop\in\LOp^{m_1,\ldots,m_r}_{n_1,\ldots,n_k;n'_1,\ldots,n'_l})
      }{
        \Ginf[;\lvar:\cty] {\lop(\trm_1,\ldots,\trm_k;\trm[2])}{\CDomain{\lop}}
      }\\ 
      \\ 
    \inferrule{~}{\Ginf[;\lvar:\cty]\tUnit \lUnit}
    \qquad 
    \inferrule{\Ginf[;\lvar:\cty]\trm{\cty[2]}\quad 
    \Ginf[;\lvar:\cty]{\trm[2]}{\cty[3]}}{
        \Ginf[;\lvar:\cty]{\tPair{\trm}{\trm[2]}}{\cty[2]\t* \cty[3]}
    }\qquad
    \inferrule{\Ginf[;\lvar:\cty]\trm{\cty[2]\t*\cty[3]}}{\Ginf[;\lvar:\cty]{\tFst\trm}{\cty[2]}}
    \qquad
    \inferrule{\Ginf[;\lvar:\cty]\trm{\cty[2]\t*\cty[3]}}{\Ginf[;\lvar:\cty]{\tSnd\trm}{\cty[3]}}
    \\ \\ 
    \inferrule{\Ginf[{,\var[2]:\ty[2];\lvar:\cty}]\trm{\cty[3]}}{\Ginf[{;\lvar:\cty}]{\fun{\var[2]}\trm}{\ty[2]\To\cty[3]}}
    \qquad 
    \inferrule{\Ginf[;\lvar:\cty]\trm{\ty[2]\To\cty[3]}\quad 
    \Ginf{\trm[2]}{\ty[2]}}{\Ginf[;\lvar:\cty]{\trm\,\trm[2]}{\cty[3]}}
    \qquad
    \inferrule{\Ginf{\trm}{\ty[2]}\quad \Ginf[;\lvar:\cty]{\trm[2]}{\cty[3]}}
    {\Ginf[;\lvar:\cty]{\copower{\trm}{\trm[2]}}{\copower{\ty[2]}{\cty[3]}}}
    \\ \\ 
    \inferrule{\Ginf[;\lvar:\cty]{\trm}{\copower{\ty[2]}{\cty[3]}}\quad  
    \Ginf[{,\var[2]:\ty[2];\lvar:\cty[3]}]{\trm[2]}{\cty[3]'}}
    {\Ginf[;\lvar:\cty]{\tensMatch{\trm}{\var[2]}{\lvar}{\trm[2]}}{\cty[3]'}}
    \qquad 
\inferrule{\Ginf[;\lvar:\cty]\trm{\cty[2]}}{\Ginf{\lfun{\lvar}\trm}{\cty\multimap \cty[2]}}\\ \\  
\inferrule{\Ginf\trm{\cty[3]\multimap\cty[2]}\quad 
\Ginf[;\lvar:\cty]{\trm[2]}{\cty[3]}}{\Ginf[;\lvar:\cty]{\lapp\trm{\trm[2]}}{\cty[2]}}
\qquad 
\inferrule{~}{\Ginf[;\lvar:\cty]{\zero}{\cty[2]}}\qquad 
\inferrule{\Ginf[;\lvar:\cty]{\trm}{\cty[2]}\quad\Ginf[;\lvar:\cty]{\trm[2]}{\cty[2]} }
{\Ginf[;\lvar:\cty]{\trm+\trm[2]}{\cty[2]}}
\end{array}
\]
 \end{minipage}}
  \caption{Typing rules for the idealised AD target language with linear types, which we consider on top of the rules of Fig. \ref{fig:types1}.\label{fig:minimal-linear-types}}
\end{figure}
We consider the $\beta\eta+$-equational theory 
of Fig. \ref{fig:beta-eta} and \ref{fig:minimal-linear-beta-eta} for our language,
where equations hold on pairs of terms of the same type in the same context.
It includes $\beta\eta$-rules as well as commutative monoid and homomorphism laws.
\begin{figure}[!t]
    \fbox{\begin{minipage}{0.98\linewidth}\begin{align*}
       &\letin{\lvar}{\trm}{\trm[2]} = \subst{\trm[2]}{\sfor{\lvar}{\trm}}\\
       &\tensMatch{\copower{\trm}{\trm[2]}}{\var}{\lvar}{\trm[3]}=\subst{\trm[3]}{\sfor{\var}{\trm},\sfor{\lvar}{\trm[2]}}&&
       \subst{\trm}{\sfor{\var}{\trm[2]}}\freeeq{\var[2],\lvar}\tensMatch{\trm[2]}{\var[2]}{\lvar}{\subst{\trm}{\sfor{\var}{\copower{\var[2]}{\lvar}}}}\\
       &\lapp{(\lfun{\lvar}{\trm})}{\trm[2]} = \subst{\trm}{\sfor{\lvar}{\trm[2]}} && \trm = \lfun{\lvar}{\lapp{\trm}{\lvar}}
       \\
       &\trm + \zero = \trm\qquad \zero + \trm = \trm&& (\trm+\trm[2])+\trm[3]=\trm+(\trm[2]+\trm[3])\qquad 
       \trm + \trm[2] = \trm[2] + \trm\\
       &(\Ginf[;\lvar:\cty]{\trm}{\cty[2]})\Rightarrow \subst{\trm}{\sfor{\lvar}{\zero}}=\zero&& (\Ginf[;\lvar:\cty]{\trm}{\cty[2]})\Rightarrow
       \subst{\trm}{\sfor{\lvar}{\trm[2]+\trm[3]}}=\subst{\trm}{\sfor{\lvar}{\trm[2]}}+\subst{\trm}{\sfor{\lvar}{\trm[3]}}
\end{align*}
      \end{minipage}}
  \caption{Equational rules for the idealised, linear AD language, which we use on top of the 
  rules of Fig. \ref{fig:beta-eta}. In addition to standard $\beta\eta$-rules for $\copower{(-)}{(-)}$- and $\multimap$-types,
  we add rules making $(\zero,+)$ into a commutative monoid on the terms of 
  each linear type as well as rules that say that terms of linear types are homomorphisms in their linear variable.
\label{fig:minimal-linear-beta-eta}
  }
  
\end{figure}
 \section{Semantics of the Source and Target Languages}\label{sec:semantics}

\subsection{Preliminaries}
\subsubsection{Category theory}
We assume familiarity with categories, functors, natural transformations,
and their theory of (co)limits and adjunctions.
We write:
\begin{itemize}
\item unary, binary, and $I$-ary products
as $\terminal$, $X_1\times X_2$, and $\prod_{i\in I}X_i$, writing
$\projection_i$ for the projections and
$()$, $\sPair{x_1}{x_2}$, and $\seq[i\in I]{x_i}$ for the tupling maps;
\item unary, binary, and $I$-ary coproducts
as $\initial$, $X_1 + X_2$, and $\sum_{i\in I}X_i$, writing
$\injection_i$ for the injections and
$[]$, $[{x_1},{x_2}]$, and $\coseq[i\in I]{x_i}$ for the cotupling
maps;
\item exponentials as $Y\Rightarrow X$, writing $\Lambda$ and $\ev$ for currying
 and evaluation.
\end{itemize}

\subsubsection{Commutative Monoids}
A \emph{monoid} $(|X|,0_X,+_X)$ consists of a set $|X|$ with an element 
$0_X\in |X|$ and a function $(+_X):|X|\times |X|\to |X|$
such that $0_X +_X x = x= x+_X 0_X$ for any $x\in |X|$ and 
$x +_X (x' +_X x'')= (x+_Xx')+_X x''$ for any $x,x',x''\in|X|$.
A \emph{monoid} $(|X|,0_X,+_X)$ is called commutative if $x+_X x' =x' +_X x$ 
for all $x,x'\in|X|$.
Given monoids $X$ and $Y$, a function $f:|X|\to|Y|$ is called 
a \emph{homomorphism of monoids} if $f(0_X)=0_Y$ and $f(x+_Xx')=f(x)+_Y f(x')$.
We write $\CMon$ for the category of commutative monoids and their homomorphisms.
We will frequently simply write $0$ for $0_X$ and $+$ for $+_X$, if $X$ is clear from context.
We will sometimes write $\sum_{i=1}^nx_i$ for $((x_1+x_2)+\ldots)\ldots +x_n$.
\begin{example}
The real numbers $\cRR$ form a commutative monoid with $0$ and $+$ equal to the 
number $0$ and ordinary addition.
\end{example}
\begin{example}\label{ex:prod-monoid}
    Given commutative monoids $\seq[i\in I]{X_i}$, we can form the \emph{product monoid}
    $\prod_{i\in I}X_i$ with underlying set 
    $\prod_{i\in I}|X_i|$, $0=\seq[i\in I]{0_{X_i}}$ and 
    $\seq[i\in I]{x_i}+\seq[i\in I]{y_i}\defeq \seq[i\in I]{x_i+y_i}$.
    Given a set $I$ and a commutative monoid $X$, we can form the \emph{power monoid} $I\Rightarrow X\defeq \prod_{i\in I} X$
    as the $I$-fold self-product monoid.
\end{example}
\noindent Ex. \ref{ex:prod-monoid} gives the categorical product in $\CMon$.
We can, for example, construct a commutative monoid structure on any Euclidean space 
$\cRR^k \defeq \set{0,\ldots, k-1}\Rightarrow \cRR$ by combining the one on $\cRR$ with the power monoid structure.

\begin{example}\label{ex:coprod-monoid}
  Given commutative monoids $\seq[i\in I]{X_i}$, we can form the \emph{coproduct monoid}
  $\sum_{i\in I}X_i$ with underlying set 
  $\set{\seq[i\in I]{x_i}\in \prod_{i\in I}{|X_i|}\mid \set{j\in I\mid x_j\neq 0_{X_j}}\textnormal{ is finite}}$, $0=\seq[i\in I]{0_{X_i}}$ and 
  $\seq[i\in I]{x_i}+\seq[i\in I]{y_i}\defeq \seq[i\in I]{x_i+y_i}$.
  Given a set $I$ and a commutative monoid $X$, we can form the \emph{copower monoid} $\copower{I}{X}\defeq \sum_{i\in I} X$
  as the $I$-fold self-coproduct monoid.
We will often write $\copower{i}{x}\defeq \seq[j\in I]{\text{if }j=i\text{ then }x\text{ else } 0_X}\in \copower{I}{X}$.
\end{example}
\noindent Ex. \ref{ex:coprod-monoid} gives the categorical coproduct in $\CMon$.

\begin{example}\label{ex:hom-monoid}
  Given commutative monoids $X$ and $Y$, we can form the commutative monoid 
  $X\multimap Y$ of \emph{homomorphisms from $X$ to $Y$}.
  We define $|X\multimap Y|\defeq \CMon(X,Y)$, $0_{X\multimap Y}\defeq (x\mapsto 0_Y)$, and $f+_{X\multimap Y} g\defeq (x\mapsto f(x)+_Y g(x))$.
\end{example}
\noindent Ex. \ref{ex:hom-monoid} gives the categorical internal hom in $\CMon$.
Commutative monoid homomorphisms $\copower{I}X\to Y$ are in one-to-one correspondence with 
functions $I\to |X\multimap Y|$.

Finally, a category $\catC$ is called $\CMon$-enriched if
we have a commutative monoid structure on each homset $\catC(C,C')$ 
and function composition gives monoid homomorphisms $\catC(C,C')\to \catC(C',C'')\multimap 
\catC(C,C'')$.
In a category $\catC$ with finite products, these products are well-known to be biproducts (i.e. simultaneously products and coproducts) if and only if $\catC$ 
is $\CMon$-enriched (for more details, see, for example \cite{fiore2007differential}):
 define $[]\defeq 0$ and $[f,g]\defeq \pi_1;f + \pi_2;g$
and, conversely, $0\defeq []$ and $f + g\defeq \sPair{\id}{\id};[f,g]$.

\subsection{Abstract Denotational Semantics}
By the universal property of $\Syn$ (Prop. \ref{prop:universal-property-syn}), the language of \S\ref{sec:language} has a 
canonical interpretation in any Cartesian closed category $(\catC,\terminal,\times,\Rightarrow)$,
once we fix $\catC$-objects $\sem{\reals^n}$ to interpret $\reals^n$ and 
 $\catC$-morphisms
$\sem{\op}\in\catC(\sem{\Domain{\op}},\sem{\reals^m})$
to interpret $\op\in \Op_{n_1,...,n_k}^m$.
That is, any Cartesian closed category with such a choice of objects and morphisms is a \emph{categorical model of the source language} of \S\ref{sec:language}.
We interpret types $\ty$ and contexts $\Gamma$ as $\catC$-objects $\sem{\ty}$ and~$\sem{\Gamma}$:
$$\sem{\var_1:\ty_1,\ldots,\var_n:\ty_n}\defeq \sem{\ty_1}\times\ldots\times \sem{\ty_n}\qquad
\sem{\Unit}\defeq \terminal\qquad
\sem{\ty\t*\ty[2]}\defeq \sem{\ty}\times \sem{\ty[2]}\qquad
 \sem{\ty\To\ty[2]}\defeq \sem{\ty}\Rightarrow \sem{\ty[2]}.$$ 
We interpret terms $\Ginf\trm\ty$ as morphisms $\sem{\trm}$ in $\catC(\sem{\Gamma},\sem{\ty})$:
\[
\begin{array}{lll}
  \sem{\op(\trm_1,\ldots,\trm_k)}\defeq (\sem{\trm_1},\ldots,\sem{\trm_k});\sem{\op}
\\
\sem{\var_1:\ty_1,\ldots,\var_n:\ty_n\vdash \var_k:\ty_k}\defeq \pi_k\qquad\quad&
\sem{\letin{\var}{\trm}{\trm[2]}}\defeq (\id, \sem{\trm});\sem{\trm[2]}
\\
\sem{\tUnit}\defeq()\qquad\quad&
\sem{\tPair\trm{\trm[2]}}\defeq\sPair{\sem{\trm}}{\sem{\trm[2]}}\\
\sem{\tFst\trm}\defeq \sem{\trm};\pi_1\qquad\quad\sem{\tSnd\trm}\defeq \sem{\trm};\pi_2 \qquad \quad&
\sem{\fun{\var}\trm}\defeq \Lambda(\sem{\trm})\qquad\quad&
\sem{\trm\,\trm[2]}\defeq \sPair{\sem{\trm}}{\sem{\trm[2]}};\ev.
\end{array}
\]

We discuss how to extend $\sem{-}$ to apply to the full target language 
of \S\ref{sec:minimal-linear-language} by defining an appropriate notion of 
\emph{categorical model for the target language} of \S\ref{sec:minimal-linear-language}.
\begin{definition}[Categorical model of the target language]
By a categorical model of the target language, we mean the following data:
\begin{itemize}
\item A categorical model $\catC$ of the source language.
\item A locally indexed category (see, for example, \cite[\S\S\S 9.3.4]{levy2012call}) $\catL:\catC^{op}\to \Cat$,
i.e.\footnote{A locally $\catC$-indexed category $\catL$ can be equivalently defined as a category $\catL$ enriched over the presheaf category $[\catC^{op},\Set]$.
We prefer to consider locally indexed categories as special cases of indexed categories, instead, as 
CHAD's natural generalization to data types of varying dimension, such as unsized arrays or sum types, 
requires us to work with more general (non-locally) indexed categories \cite{lucatellivakar2021chad}.} a (strict) contravariant functor from $\catC$ to the category $\Cat$ of 
categories, such that $\ob\catL(C)=\ob\catL(C')$ and $\catL(f)(L)=L$ for any object $L$ of $\ob\catL(C)$ 
and any $f:C'\to C$ in $\catC$.
\item $\catL$ is \emph{biadditive}: each category $\catL(C)$ has (chosen) finite biproducts
 $(\terminal,\times)$
and $\catL(f)$ preserves them, for any $f:C'\to C$ in $\catC$, 
in the sense that $\catL(f)(\terminal)=\terminal$ and $\catL(f)(L\times L')=
\catL(f)(L)\times\catL(f)(L')$.
\item $\catL$ \emph{supports $\copower{(-)}{(-)}$-types and 
$\Rightarrow$-types}: $\catL(\pi_1)$ has a left adjoint $\copower[C]{C'}{-}$ and a right 
adjoint functor
$C'\Rightarrow_C -$, for each product projection $\pi_1:C\times C'\to C$
in $\catC$, satisfying a Beck-Chevalley condition\footnote{This condition says that the types $\copower[C]{C'}{L}$ and 
$C'\Rightarrow_C L$ do not depend on $C$.
We need to add this condition to match the syntax of the target language, in which copowers and powers only depend on two argument types.}:
$\copower[C]{C'}{L}=\copower[C'']{C'}{L}$ and 
$C'\Rightarrow_C L=C' \Rightarrow_{C''} L$
for any $C,C''\in\ob \catC$.
We simply write $\copower{C'}{L}$ and $C'\Rightarrow L$.
We write $\Phi$ and $\Psi$ for the natural isomorphisms
$\catL(C)(\copower{C'}{L}, L')\xto{\cong} \catL(C\times C')(L,L')$ 
and $\catL(C\times C')(L,L')\xto{\cong}\catL(C)(L, C'\Rightarrow L')$.
\item $\catL$ \emph{supports Cartesian $\multimap$-types}:
the functor $\catC^{op}\to \Set$;
$C\mapsto \catL(C)(L,L')$ is representable for any objects $L,L'$
of $\catL$.
That is,
we have objects $L\multimap L'$ of $\catC$ with
isomorphisms $ \lLambda:\catL(C)(L,L') \xto{\cong}\catC(C,L\multimap L')$, 
natural in $C$.
\item $\catL$ interprets primitive types and operations: we have a choice $\sem{\creals^n}\in\ob\catL$ to interpret $\creals^n$
and, for each $\lop\in\LOp_{n_1,...,n_k;n'_1,\ldots, n'_l}^{m_1,\ldots,m_r}$, compatible $\catL$-morphisms
$\sem{\lop}$ in $\catL(\sem{\reals^{n_1}}\times\cdots\times \sem{\reals^{n_k}})(\sem{\LDomain{\lop}}, \sem{\CDomain{\lop}})$.
\end{itemize}
\end{definition}
\noindent In particular, any biadditive model of intuitionistic linear/non-linear logic \cite{mellies2009categorical,fiore2007differential,benton1994mixed}
is such a categorical model, as long as we choose interpretations for primitive types and operations.

Next, we turn to the interpretation of our target language in such models, which gives an operational intuition of the different components of a categorical model.
We can interpret linear types $\cty$ as 
objects $\sem{\cty}$ of $\catL$:
\begin{align*}
   & \sem{\lUnit}\defeq \terminal \qquad \sem{\cty\t*\cty[2]}\defeq \sem{\cty}\times \sem{\cty[2]}
   \qquad \sem{\ty\To\cty[2]}\defeq \sem{\ty}\Rightarrow \sem{\cty[2]}\qquad 
   \sem{\copower{\ty}{\cty[2]}}\defeq \copower{\sem{\ty}}{\sem{\cty[2]}}.
    \end{align*}
We can interpret $\cty\multimap \cty[2]$ as the 
$\catC$-object $\sem{\cty\multimap \cty[2]}\defeq \sem{\cty}\multimap \sem{\cty[2]}$.
Finally, we can interpret terms $\Ginf\trm\ty$ as morphisms 
$\sem{\trm}$ in $\catC(\sem{\Gamma},\sem{\ty})$ and terms $\Ginf[;\lvar:\cty]\trm{\cty[2]}$
as $\sem{\trm}$ in $\catL(\sem{\Gamma})(\sem{\cty},\sem{\cty[2]})$:
\begin{align*}
&\sem{\lop(\trm_1,\ldots,\trm_k;\trm[2])}\defeq \sem{\trm[2]};\catL((\sem{\trm_1},\ldots,\sem{\trm_k}))(\sem{\lop})\\
&\sem{\Ginf[;\lvar:\cty]\lvar\cty}\defeq \id[\sem{\cty}]
\qquad \sem{\letin{\var}{\trm}{\trm[2]}}\defeq \catL((\id,\sem{\trm}))(\sem{\trm[2]})\qquad \sem{\letin{\lvar}{\trm}{\trm[2]}}\defeq \sem{\trm};\sem{\trm[2]}\\
& \sem{\tUnit}\defeq()\qquad \sem{\tPair\trm{\trm[2]}}\defeq\sPair{\sem{\trm}}{\sem{\trm[2]}}\qquad
\sem{\tFst\trm}\defeq\sem{\trm};\pi_1\qquad\sem{\tSnd\trm}\defeq \sem{\trm};\pi_2 \\
&\sem{\fun{\var}\trm}\defeq \Psi(\sem{\trm}) \qquad 
\sem{\trm\,\trm[2]}\defeq \catL(\sPair{\id}{\sem{\trm[2]}})(\inv\Psi(\sem{\trm})) \\
&\sem{\copower{\trm}{\trm[2]}}\defeq \catL(\sPair{\id}{\sem{\trm}})(\Phi(\id) ); (\copower{\sem{\ty[2]}}{\sem{\trm[2]}} )
\qquad\sem{\tensMatch{\trm}{\var[2]}{\lvar}{\trm[2]}}\defeq \sem{\trm};\inv\Phi(\sem{\trm[2]}) \\ & 
\sem{\lfun{\lvar}\trm}\defeq \lLambda({\sem{\trm}}) \qquad 
\sem{\lapp{\trm}{\trm[2]}}\defeq \sem{\trm[2]};\inv{\lLambda}(\sem{\trm})\qquad 
\sem{\zero}\defeq []\qquad \sem{\trm+\trm[2]}\defeq \sPair{\id}{\id};[\sem{\trm},\sem{\trm[2]}].
\end{align*}
Observe that we interpret $\zero$ and $+$ using the biproduct structure of $\catL$.

\begin{proposition}
The interpretation $\sem{-}$ of the language of 
\S\ref{sec:minimal-linear-language} in 
categorical models
is both sound and complete 
with respect to the $\beta\eta+$-equational theory: $\trm\bepeq \trm[2]$ iff 
$\sem{\trm}=\sem{\trm[2]}$ in each such~model.
\end{proposition}
The proof is a minor variation of syntax-semantics correspondences developed in detail in chapters 3 and 5 of 
\cite{vakar2017search}, where we use the well-known result 
that finite products in a category are biproducts iff the category is enriched over commutative monoids \cite{fiore2007differential}.
Soundness follows by case analysis on the $\beta\eta+$-rules.
Completeness follows by the construction of the syntactic model 
$\LSyn:\CSyn^{op}\to\Cat$:
\begin{itemize}
    \item $\CSyn$ extends its full subcategory $\Syn$ with Cartesian $\multimap$-types;
    \item Objects of $\LSyn(\ty)$ are linear types $\cty[2]$ of our target language.
    \item Morphisms in $\LSyn(\ty)(\cty[2],\cty[3])$ are terms
    $\var:\ty;\lvar:\cty[2]\vdash \trm:\cty[3]$ modulo $(\alpha)\beta\eta+$-equivalence.
    \item Identities in $\LSyn(\ty)$ are represented by the terms 
    $\var:\ty;\lvar:\cty[2]\vdash\lvar:\cty[2]$.
    \item Composition of $\var:\ty;\lvar:\cty[2]_1\vdash \trm:\cty[2]_2$
    and $\var:\ty;\lvar:\cty[2]_2\vdash \trm[2]:\cty[2]_3$ in $\LSyn(\ty)$
    is represented by $\var:\ty;\lvar:\cty[2]_1\vdash \letin{\lvar}{\trm}{\trm[2]}:\cty[2]_3$.
    \item Change of base $\LSyn(\trm):\LSyn(\ty)\to\LSyn(\ty')$ along 
    $(\var':\ty'\vdash \trm:\ty)\in \CSyn(\ty',\ty)$ is defined  
    $\LSyn(\trm)(\var:\ty;\lvar:\cty[2]\vdash \trm[2]:\cty[3])\defeq 
    \var':\ty';\lvar:\cty[2]\vdash \letin{\var}{\trm}{\trm[2]}~:~\cty[3]
    $.
    \item All type formers are interpreted in the expected way, based on their notation,
    using introduction and elimination rules for the required structural isomorphisms.
\end{itemize}

\subsection{Concrete Denotational Semantics}
\subsubsection{Sets and Commutative Monoids}
Throughout this paper, we have a particularly simple instance of the abstract semantics of our languages in mind,
as we intend to interpret $\reals^n$ as the usual Euclidean space 
$\RR^n$ (considered as a set) and to interpret each program $\var_1:\reals^{n_1},\ldots,\var_k:\reals^{n_k}\vdash\trm:{\reals^m}$ as a
function $\RR^{n_1}\times\ldots\times \RR^{n_k}\to \RR^m$.
Similarly, we intend to interpret $\creals^n$ as the commutative monoid $\cRR^n$ and each program 
$\var_1:\reals^{n_1},\ldots,\var_k:\reals^{n_k};\var[2]:\creals^{m}\vdash \trm : \creals^r$
as a function $\RR^{n_1}\times\ldots\times \RR^{n_k}\to \cRR^m\multimap \cRR^r$.
That is, we will work with a concrete denotational semantics in terms of sets and commutative monoids.

Some readers will immediately recognize that the free-forgetful adjunction $\Set\leftrightarrows \CMon$ gives a 
model of full intuitionistic linear logic \cite{mellies2009categorical}.
In fact, since $\CMon$ is $\CMon$-enriched, the model is biadditive \cite{fiore2007differential}.

However, we do not need such a rich type system.
For us, the following suffices.
Define $\CMon(X)$, for $X\in\ob\Set$, to have the objects of $\CMon$ 
and homsets $\CMon(X)(Y,Z)\defeq \Set(X,Y\multimap~Z)$.
Identities are defined as $x\mapsto (y\mapsto y)$ and composition 
$f;_{\CMon(X)}g$ is given by $x\mapsto (f(x);_{\CMon}g(x))$.
Given $f\in\Set(X,X')$, we define change-of-base $\CMon(X')\to\CMon(X)$
as $\CMon(f)(g)\defeq f;_{\Set}g$.
$\CMon(-)$ defines a locally indexed category.
By taking $\catC=\Set$ and $\catL(-)=\CMon(-)$, we 
obtain a concrete instance of our abstract semantics.
Indeed, we have natural isomorphisms
\begin{align*}
  &\CMon(X)(\copower{X'}{Y}, Z)\xto{\Phi} 
\CMon(X\times X')(Y,Z)
&&
\CMon(X\times X')(Y, Z)\xto{\Psi} 
\CMon(X)(Y,X'\Rightarrow Z)\\
    &\Phi(f)(x,x')(y)\defeq f(x)(\copower{x'}{y}) && \Psi(f)(x)(y)(x')\defeq f(x,x')(y)\\[-4pt] 
&\inv\Phi(f)(x)(\sum_{i=1}^n(\copower{x'_i}{y_i}))\defeq 
\sum_{i=1}^n f(x,x'_i)(y_i)&&
\inv\Psi(f)(x,x')(y)\defeq f(x)(y)(x').\\[-18pt]
\end{align*}

The main motivating examples of morphisms in this category are derivatives.
Recall that the \emph{derivative at $x$}, $Df(x)$, and \emph{transposed derivative at $x$},
$\transpose{Df}(x)$, 
of a \smooth{} function $f:\RR^n\to\RR^m$ are defined as the unique functions 
$Df(x):\RR^n\to \RR^m$ and $\transpose{Df}(x):\RR^m\to\RR^n$ satisfying
$$
Df(x)(v)=\lim_{\delta\to 0}\frac{f(x+\delta\cdot v)-f(x)}{\delta}
\qquad 
\innerprod{\transpose{Df}(x)(w)}{v}=\innerprod{w}{Df(x)(v)}
,
$$
where we write $\innerprod{v}{v'}$ for the inner product
$\sum_{i=1}^n (\pi_i v)\cdot (\pi_i v')$ 
of vectors $v,v'\in\RR^n$.
Now, for \smooth{} $f:\RR^n\to \RR^m$, 
 $Df$ and $\transpose{Df}$
give maps in $\CMon(\RR^n)(\cRR^n,\cRR^m)$ and 
$\CMon(\RR^n)(\cRR^m,\cRR^n)$, respectively.
Indeed, derivatives $Df(x)$ of $f$ at $x$ are linear functions,
as are transposed derivatives $\transpose{Df}(x)$.
When $f$ is twice differentiable, both depend differentiably on $x$.
Note that the derivatives are not merely linear in the sense of preserving $0$ and 
$+$.
They are also multiplicative in the sense that
$(Df)(x)(c\cdot v)=c\cdot (Df)(x)(v)$.
{We could have captured this property by working with vector spaces rather than commutative monoids.
However, we will not need this property to phrase or establish correctness of AD.}
Therefore, we restrict our attention to the more straightforward structure of 
commutative monoids.

Defining $\sem{\creals^n}\defeq \cRR^n$ and interpreting
each $\lop\in\LOp$ as the (\smooth{}) function
$\sem{\lop}:(\RR^{n_1}\times \ldots\times \RR^{n_k})\To 
(\cRR^{n'_1}\times \ldots\times \cRR^{n'_l})\multimap (\cRR^{m_1}\times\cdots\times \cRR^{m_r})$  
that it is intended to represent, we obtain a canonical interpretation 
of our target language in $\CMon$.

\subsection{Operational Semantics}
In this section, we describe an operational semantics for our 
source and target languages.
We consider call-by-value evaluation, but similar results can be obtained for call-by-name evaluation\footnote{In fact, we 
conjecture that our target language is pure in the sense that reductions are confluent.
}.
We present this semantics in big-step style.
Finally, we show that our denotational semantics are adequate with respect to this 
 operational semantics, thereby showing that the 
denotational semantics are sound tools for reasoning about our programs.

We consider the following program \emph{values}, where we write $\cnst$ for $\op()$ and $\lcnst(\lvar)$ for $\lop(;\lvar)$:\\
\begin{syntax}
    \val, \val[2], \val[3] & \gdefinedby & & \syncat{values}                          \\
    &\gor& \var\\
    &\gor & \cnst\\
    &\gor & \tUnit\\
    &\gor & \tPair{\val}{\val[2]}
\end{syntax}\qquad\qquad\qquad
  \begin{syntax}
    &\gor &\lvar\\
    &\gor\quad\,&   \fun{\var}{\trm}\\
    &\gor & \lfun{\lvar}{\trm}\\
    &\gor & \copower{\val_1}{\val[2]_1}+ (\copower{\val_2}{\val[2]_2}+(\cdots  \copower{\val_n}{\val[2]_n})\cdots)\\
    &\gor & \lcnst(\lvar).
\end{syntax} \\
We then define the big-step reduction relation $\trm\Downarrow \val$,
which says that a program
$\trm$ evaluates to the value $\val$, in Fig. \ref{fig:operational-semantics}.
To define this semantics, 
we assume that our languages contain, at least, nullary operations $\op=\cnst$ 
for all constants $c\in\RR^n$ and nullary linear operations $\lcnst$ for all 
linear maps (matrices) $lc\in \cRR^n\multimap \cRR^m$.
For all operations $\op$ and linear operations $\lop$, we assume that an intended semantics $\sem{\op}$ and $\sem{\lop}$ is specified as (functions on) vectors of reals.
\begin{figure}[!t]
  \fbox{\parbox{\linewidth}{\begin{align*}
    \inferrule{\text{$\val$ is not a linear map between tuples of real arrays}}
    {\val\Downarrow\val}
    &&
    \inferrule{     (\Ginf[;\lvar:\creals^{n_1}\t*\cdots\t*\creals^{n_k}]\lvar{\creals^{n_1}\t*\cdots\t*\creals^{n_k}})\!}
    {\lvar\Downarrow \underline{id}(;\lvar)}
    &&
    \inferrule{~}
    {\lcnst(\lvar)\Downarrow\lcnst(\lvar)}
\end{align*}
\begin{align*}
  \inferrule{\trm\Downarrow \val\quad \subst{\trm[2]}{\sfor{\var}{\val}}\Downarrow \val[2]}
  {\letin{\var}{\trm}{\trm[2]}\Downarrow \val[2]}
  &&
  \inferrule{\trm_1\Downarrow \cnst_1\quad\cdots\quad \trm_n\Downarrow \cnst_n}
  {\op(\trm_1,\ldots,\trm_n)\Downarrow \underline{\sem{\op}(c_1,\ldots,c_n)}}
  &&
  \inferrule{\trm\Downarrow \val\quad\trm[2]\Downarrow\val[2]}
  {\tPair{\trm}{\trm[2]}\Downarrow \tPair{\val}{\val[2]}} 
  &&
    \inferrule{\trm\Downarrow\tPair{\val}{\val[2]}}
    {\tFst{\trm}\Downarrow\val}
    && 
    \inferrule{\trm\Downarrow\tPair{\val}{\val[2]}}
    {\tSnd{\trm}\Downarrow\val[2]}
    \end{align*}
    \begin{align*}
      \inferrule{\trm\Downarrow\fun{\var}\trm[3]\quad 
      \trm[2]\Downarrow \val\quad 
      \subst{\trm[3]}{\sfor{\var}{\val}}\Downarrow \val[2]}
      {\trm\,\trm[2]\Downarrow \val[2]}
      &&
    \inferrule{\trm\Downarrow \val\quad \subst{\trm[2]}{\sfor{\lvar}{\val}}\Downarrow \val[2]}
    {\letin{\lvar}{\trm}{\trm[2]}\Downarrow \val[2]}
    &&
    \inferrule{\trm_1\Downarrow \cnst_1\quad \cdots\quad 
    \trm_n\Downarrow \cnst_n\quad 
    \trm[2]\Downarrow \lcnst(\lvar)
    }
    {\lop(\trm_1,\ldots,\trm_n;\trm[2])\Downarrow \underline{\sem{\lop}(c_1,\ldots,c_n;lc)}(\lvar)}
    &&
    \inferrule{\trm\Downarrow \val\quad 
    \trm[2]\Downarrow\val[2]}
    {\copower{\trm}{\trm[2]}\Downarrow \copower{\val}{\val[2]}}
    \end{align*}
    \begin{align*}
    \inferrule{\trm\Downarrow \copower{\val_1}{\val[2]_1}+ (\copower{\val_2}{\val[2]_2}+(\cdots  \copower{\val_n}{\val[2]_n})\cdots) \quad
    \set{\subst{\trm[2]}{\sfor{\var[2]}{\val_i},\sfor{\lvar}{\val[2]_i}}\Downarrow \val[3]_i}_{i=1}^n\qquad 
    \val[3]_1+ (\val[3]_2+(\cdots  \val[3]_n)\cdots)\Downarrow \val[3]
    }
    {\tensMatch{\trm}{\var[2]}{\lvar}{\trm[2]}\Downarrow \val[3]}
    \end{align*}
    \begin{align*}
      \inferrule{\trm\Downarrow \lfun{\lvar}\val\quad 
      \trm[2]\Downarrow\val[2]\quad 
      \subst{\val}{\sfor{\lvar}{\val[2]}}\Downarrow \val[3]}
      {\lapp{\trm}{\trm[2]}\Downarrow\val[3]}
      &&
      \inferrule{(\Ginf[;\lvar:\creals^{n_1}\t*\cdots\t*\creals^{n_k}]\zero{\creals^{m_1}\t*\cdots\t*\creals^{m_l}})}
      {\zero\Downarrow \zero(;\lvar)}
      \end{align*}
    \begin{align*}
      \inferrule{\trm_1\Downarrow \lcnst_1(\lvar)
      \quad \trm_2\Downarrow \lcnst_2(\lvar)}
      {\trm_1+\trm_2\Downarrow \underline{(lc_1+lc_2 )}(\lvar)}
      &&
      \inferrule{\trm_1\Downarrow \tUnit \quad \trm_2\Downarrow \tUnit 
      }
      {\trm_1+\trm_2\Downarrow \tUnit }
      \quad
      \inferrule{\trm_1\Downarrow \tPair{\val_1}{\val[2]_1} \quad \trm_2\Downarrow \tPair{\val_2}{\val[2]_2}
      \quad 
      \val_1+\val_2\Downarrow \val
      \quad 
      \val[2]_1+\val[2]_2\Downarrow \val[2]
      }
      {\trm_1+\trm_2\Downarrow \tPair{\val}{\val[2]} }
    \end{align*}
    \begin{align*}
      \inferrule{\trm_1\Downarrow \fun\var\trm[2]_1\quad 
      \trm_2\Downarrow\fun\var\trm[2]_2}
      {\trm_1+\trm_2\Downarrow \fun\var \trm[2]_1 + \trm[2]_2}
    \end{align*}
    \begin{align*}
      \inferrule{\trm_1\Downarrow \copower{\val_{1}}{\val[2]_1}+ (\copower{\val_2}{\val[2]_2}+(\cdots  \copower{\val_n}{\val[2]_n})\cdots) 
      \qquad 
      \trm_2\Downarrow \copower{\val_{n+1}}{\val[2]_{n+1}}+ (\copower{\val_{n+2}}{\val[2]_{n+2}}+(\cdots  \copower{\val_{n+m}}{\val[2]_{n+m}})\cdots)}
      {\trm_1+\trm_2\Downarrow \copower{\val_1}{\val[2]_1}+ (\copower{\val_2}{\val[2]_2}+(\cdots  \copower{\val_{n+m}}{\val[2]_{n+m}})\cdots)}
    \end{align*}
 }}
\caption{\label{fig:operational-semantics} The big-step call-by-value operational semantics $\trm\Downarrow \val$ for the 
source and target languages.
In the first rule, we intend to indicate that $\val\Downarrow \val$ unless $\val$ is a linear function 
between tuples of real arrays, i.e. unless it has a judgement of the form $\Gamma;\lvar:\creals^{n_1}\t*\cdots\t*\creals^{n_k} \vdash \val:{\creals^{m_1}\t*\cdots\t*\creals^{m_l}}$.}
\end{figure}
As a side note, we observe that this operational semantics has the following basic properties:
\begin{lemma}[Subject Reduction, Termination, Determinism]
If $\Ginf\trm\ty$ then there is a unique value $\val$ such that $\trm\Downarrow \val$. Then, $\Ginf\val\ty$.
Similarly, if $\Ginf[;\lvar:\cty]\trm{\cty[2]}$, then there is a unique value $\val$ such that $\trm\Downarrow \val$. Then, $\Ginf[;\lvar:\cty]\val{\cty[2]}$.
\end{lemma}
Subject reduction and termination are proved by a standard logical relations argument similar to those in \cite{DBLP:journals/corr/abs-1907-11133}.
Determinism follows by noting that all rules in the definition of $\Downarrow$ have conclusions $\trm\Downarrow \val$ with disjoint $\trm$. 

\noindent In fact, since every well-typed program $\trm$ has a unique value $\val$ such that 
$\trm\Downarrow \val$, we write $\Downarrow \trm$ for this $\val$.

We assume that only first-order types are observable (i.e., have decidable equality on their values):\\
\begin{syntax}
    \phi, \psi& \gdefinedby & & \syncat{first-order Cartesian types}                          \\
    &\gor& \reals^n\\
    &\gor & \Unit\\
    &\gor & \phi\t*\psi \end{syntax}\qquad\qquad\qquad
  \begin{syntax}
    \phi, \psi& \gdefinedby & & \syncat{first-order linear types}                          \\
    &\gor& \creals^n\\
    &\gor & \lUnit\\
    &\gor & \underline{\phi}\t*\underline{\psi}.
\end{syntax} \\

We define \emph{program contexts $C[\_]$} to be programs $C[\_]$ 
that use the variable $\_$ exactly once.
We call such program contexts \emph{of first-order type} if they satisfy the typing judgement $\_:\ty\vdash C[\_]:\phi$ for first-order Cartesian type $\phi$
or $\_:\ty;\lvar:\underline{\psi}\vdash C[\_]:\underline{\phi}$ for first-order linear types $\underline{\psi}$ and $\underline{\phi}$.
We write $C[\trm]$ for the {capturing} substitution of $\trm$ for $\_$ in $C[\_]$.
This operational semantics and notion of observable types lead us to define \emph{observational equivalence} (also known as contextual equivalence) $\trm\approx\trm[2]$ of programs $\cdot\vdash \trm,\trm[2]:\ty$, where we say that $\trm\approx\trm[2]$ holds if $ \Downarrow C[\trm]=\Downarrow C[\trm[2]]$ for all program contexts of first-order type.
Similarly, we call two programs $\cdot;\lvar:\cty\vdash \trm,\trm[2]:\cty[2]$ of linear type observationally equivalent (write also $\trm\approx \trm[2]$) if $\lfun{\lvar}\trm\approx\lfun{\lvar}\trm[2]$.

Note that we consider values $\cdot;\lvar:\underline{\psi}\vdash\val:\underline{\phi}$ for first-order linear types $\underline{\psi}$ and $\underline{\phi}$
to be observable, since linear functions between finite-dimensional spaces are finite-dimensional objects that can be fully observed by evaluating them on a (finite) basis for their domain type $\underline{\psi}$.
Indeed, such values $\val$ are always of the form $\lcnst(\lvar)$ for some $lc:\cRR^n\to \cRR^m$, and hence are effectively matrices.

We first show two standard lemmas.
\begin{lemma}[Compositionality of $\sem{-}$]
For any two terms $\Ginf{\trm,\trm[2]}\ty$ and any type-compatible 
program context $C[\_]$ we have that 
$\sem{\trm}=\sem{\trm[2]}$ implies $\sem{C[\trm]}=\sem{C[\trm[2]]}$.
\end{lemma}
\noindent This is proved by induction on the structure of terms.
\begin{lemma}[Soundness of $\Downarrow$]
If $\trm$ is well-typed, we have that $\sem{\trm}=\sem{\Downarrow \trm}$.
\end{lemma}
\noindent This is proved by induction on the definition of $\Downarrow$:
note that every operational rule is also an equation in the semantics.
Then, adequacy follows.
\begin{theorem}[Adequacy]
If $\sem{\trm}=\sem{\trm[2]}$,
it follows that $\trm\approx\trm[2]$.
\end{theorem}
\begin{proof}
Suppose that $\sem{\trm}=\sem{\trm[2]}$ and let $C[\_]$ be a type-compatible program context of first-order type.
Then,\\ $\sem{\Downarrow C[\trm]}=\sem{C[\trm]}=\sem{C[\trm[2]]}=\sem{\Downarrow C[\trm[2]]}$
by the previous two lemmas.
Finally, as values of observable types are easily seen to be 
faithfully (injectively) interpreted in our denotational semantics, it follows that $\Downarrow C[\trm]=\Downarrow C[\trm[2]]$.
Therefore, $\trm\approx\trm[2]$.
\end{proof}
\noindent That is, the denotational semantics is a sound means for proving 
observational equivalences of the operational semantics.

 \section{Pairing Primals with (Co)Tangents, Categorically}\label{sec:self-dualization}
In this section, we show that any categorical model $\catL:\catC^{op}\to\Cat$
of our target language gives rise to two Cartesian closed categories 
$\Sigma_{\catC}\catL$ and $\Sigma_{\catC}\catL^{op}$.
We believe that these observations of Cartesian closure are novel.
Surprisingly, they are highly relevant for obtaining a principled understanding 
of AD on a higher-order language: the former for forward AD, and the latter for reverse AD.
Applying these constructions to the syntactic category $\LSyn:\CSyn^{op}\to \Cat$ of our target language, 
we produce
a canonical definition of the AD macros as the canonical interpretation of the $\lambda$-calculus 
in the Cartesian closed categories $\Sigma_{\CSyn}\LSyn$ and $\Sigma_{\CSyn}\LSyn^{op}$.
In addition, when we apply this construction to the denotational semantics $\CMon:\Set^{op}\to\Cat$
and invoke a categorical logical relations technique, known as \emph{subsconing},
we find an elegant correctness proof of the source-code transformations.
The abstract construction delineated in this section is in many ways the theoretical crux of this paper.

\subsection{Grothendieck Constructions on Strictly Indexed Categories}\label{ssec:grothendieck-construction}
Recall that for any strictly indexed category, i.e., a (strict) functor $\catL:\catC^{op}\to\Cat$, 
we can consider its total category (or Grothendieck construction) $\Sigma_\catC \catL$,
which is a fibred category over $\catC$ (see \cite[sections A1.1.7, B1.3.1]{johnstone2002sketches}).
We can view it as a $\Sigma$-type of categories, which 
generalizes the Cartesian product.
Concretely, its objects are pairs $(A_1,A_2)$ of objects $A_1$ of $\catC$ and 
$A_2$ of $\catL(A_1)$.
Its morphisms $(A_1,A_2)\to (B_1,B_2)$ are pairs $\sPair{f_1}{f_2}$ of a morphism $f_1:A_1\to{} B_1$ in
$\catC$ 
and a morphism $f_2:A_2\to \catL(f_1)(B_2)$ in $\catL(A_1)$.
Identities are $\id[(A_1,A_2)]\defeq (\id[A_1], \id[A_2])$
and composition is $(f_1,f_2);(g_1,g_2)\defeq (f_1;g_1, f_2; \catL(f_1)(g_2))$.
Furthermore, given a strictly indexed category $\catL:\catC^{op}\to \Cat$,
we can consider its fibrewise dual category $\catL^{op}:\catC^{op}\to \Cat$,
which is defined as the composition $\catC^{op}\xto{\catL}\Cat\xto{op}\Cat$.
Thus, we can apply the same construction to $\catL^{op}$ to obtain a category~$\Sigma_{\catC}\catL^{op}$.

\subsection{Structure of $\Sigma_{\catC}\catL$ and $\Sigma_{\catC}\catL^{op}$ for Locally Indexed Categories}
\S\S\ref{ssec:grothendieck-construction} applies, in particular, to the locally indexed categories of \S\ref{sec:semantics}.
In this case, we will analyze the categorical structure of $\Sigma_{\catC}\catL$ and $\Sigma_{\catC}\catL^{op}$.
For reference, we first give a concrete description.

$\Sigma_{\catC}\catL$ is the following category:
\begin{itemize}
\item objects are pairs $(A_1,A_2)$ of objects $A_1$ of $\catC$ and $A_2$ of $\catL$;
\item morphisms $(A_1,A_2)\to (B_1,B_2)$ are pairs $(f_1,f_2)$ with $f_1:A_1\to B_1\in\catC$ and $f_2:A_2\to B_2\in \catL(A_1)$;
\item composition of $(A_1,A_2)\xto{(f_1,f_2)}(B_1,B_2)$
and $(B_1,B_2)\xto{(g_1,g_2)}(C_1,C_2)$ is given by 
$(f_1;g_1, f_2;\catL(f_1)(g_2))$ and identities $\id[(A_1,A_2)]$ are $(\id[A_1],\id[A_2])$.
\end{itemize}

$\Sigma_{\catC}\catL^{op}$ is the following category:
\begin{itemize}
    \item objects are pairs $(A_1,A_2)$ of objects $A_1$ of $\catC$ and $A_2$ of $\catL$;
    \item morphisms $(A_1,A_2)\to (B_1,B_2)$ are pairs $(f_1,f_2)$ with $f_1:A_1\to B_1\in\catC$ and $f_2:B_2\to A_2\in \catL(A_1)$;
    \item composition of $(A_1,A_2)\xto{(f_1,f_2)}(B_1,B_2)$ and 
    $(B_1,B_2)\xto{(g_1,g_2)}(C_1,C_2)$ is given by 
    $(f_1;g_1, \catL(f_1)(g_2);f_2)$ and identities $\id[(A_1,A_2)]$ are $(\id[A_1],\id[A_2])$.
\end{itemize}

These categories are relevant to automatic differentiation for the following reason.
Let us write $\CartSp$ for the category of Cartesian spaces $\RR^n$ and \smooth{} 
functions between them.
Observe that for any categorical model $\catL:\catC^{op}\to\Cat$ of the target language, 
\begin{align*}
&\Sigma_\catC\catL((A_1,A_2),(B_1,B_2))=\catC(A_1,B_1)\times \catL(A_1)(A_2,B_2)\cong \catC(A_1, B_1\times (A_2\multimap B_2))    \\
&\Sigma_\catC\catL^{op}((A_1,A_2),(B_1,B_2))=\catC(A_1,B_1)\times \catL(A_1)(B_2,A_2)\cong \catC(A_1, B_1\times (B_2\multimap A_2)).
\end{align*}
Then, observing that the composition in these $\Sigma$-types of categories is precisely the chain rule,
we see that the paired-up derivative $\Dsemsymbol$ and transposed derivative $\Dsemrevsymbol$ of \S\S\ref{ssec:pairing-sharing} define functors 
\[
\Dsemsymbol: \CartSp\to \Sigma_\Set\CMon\qquad \qquad\qquad\qquad\Dsemrevsymbol:\CartSp\to\Sigma_\Set\CMon^{op}.
\]
As we will see in \S\ref{sec:ad-transformation}, we can implement (higher-order extensions of) these functors as code transformations
\[
\Dsynsymbol: \Syn\to \Sigma_\CSyn\LSyn\qquad \qquad\qquad\qquad\Dsynrevsymbol:\Syn\to\Sigma_\CSyn\LSyn^{op}.
\]

As we will see, we can derive these code transformations by examining the categorical structure present in $\Sigma_{\catC}\catL$ and 
$\Sigma_{\catC}\catL^{op}$ for categorical models $\catL:\catC^{op}\to \Cat$ of the target language in the sense of \S\ref{sec:semantics}.
We believe the existence of this categorical structure is a novel observation.
We will make heavy use of it to define our AD algorithms and to prove them~correct.
\begin{theorem}\label{thm:grothendieck-ccc-covariant}
For a categorical model $\catL:\catC^{op}\to \Cat$ of the target language, $\Sigma_{\catC}\catL$ has:
\begin{itemize}
    \item terminal object 
$\terminal=(\terminal,\terminal)$ and binary products 
$(A_1,A_2)\times (B_1,B_2)=(A_1\times B_1,A_2\times  B_2)$;
\item exponentials $(A_1,A_2)\Rightarrow (B_1,B_2)=
(A_1\Rightarrow  (B_1\times (A_2\multimap B_2)), A_1\Rightarrow B_2).$
\end{itemize}
\end{theorem}
\begin{proof}
    We have (natural) bijections 
    \begin{align*}
    &\Sigma_{\catC}\catL((A_1,A_2), (\terminal,\terminal))
    =
    \catC(A_1,\terminal)\times \catL(A_1)(A_2,\terminal)
    \cong\terminal\times \terminal
    \cong \terminal\explainr{$\terminal$ terminal in $\catC$ and $\catL(A_1)$}\\
    &\\[-6pt]
    &\Sigma_{\catC}\catL((A_1,A_2), (B_1\times C_1,B_2\times  C_2))
    =\catC(A_1,B_1\times C_1)\times \catL(A_1)(A_2, B_2\times C_2)\hspace{-40pt}\;\\
    &\cong \catC(A_1,B_1)\times \catC(A_1,C_1)\times \catL(A_1)(A_2, B_2)\times 
    \catL(A_1)(A_2, C_2)\explainr{$\times $ product in $\catC$ and $\catL(A_1)$}\\
    &\cong \Sigma_{\catC}\catL((A_1,A_2),(B_1,B_2))\times \Sigma_{\catC}\catL((A_1,A_2), (C_1,C_2))\hspace{-40pt}\;\\
    &\\[-6pt]
    &\Sigma_{\catC}\catL((A_1,A_2)\times (B_1,B_2), (C_1,C_2)) =
        \Sigma_{\catC}\catL((A_1\times B_1,A_2\times  B_2), (C_1,C_2))\hspace{-40pt}\; \\
        &=\catC(A_1\times B_1, C_1)\times \catL(A_1\times B_1)(A_2\times B_2, C_2)\\
        &\cong
        \catC(A_1\times B_1, C_1)\times \catL(A_1\times B_1)(A_2, C_2)\times  \catL(A_1\times B_1)(B_2,C_2)\explainr{$\times$ coproducts in $\catL(A_1\times B_1)$}\\
        &\cong
        \catC(A_1\times B_1, C_1)\times \catL(A_1)(A_2, B_1\Rightarrow C_2)\times  \catL(A_1\times B_1)(B_2,C_2)
        \explainr{$\Rightarrow$-types in $\catL$}\\
        &\cong
        \catC(A_1\times B_1, C_1)\times \catL(A_1)(A_2, B_1\Rightarrow C_2)\times  \catC(A_1\times B_1,B_2\multimap C_2)
        \explainr{Cartesian $\multimap$-types}\\
        &\cong\catC(A_1\times B_1, C_1\times  (B_2\multimap C_2))\times 
        \catL(A_1)(A_2, B_1\Rightarrow C_2)\explainr{$\times$ is product in $\catC$}\\
        &\cong \catC(A_1, B_1\Rightarrow (C_1\times  (B_2\multimap C_2)))\times 
        \catL(A_1)(A_2, B_1\Rightarrow C_2)
        \explainr{$\Rightarrow$ is exponential in $\catC$}\\
        &= \Sigma_{\catC}\catL((A_1,A_2), (B_1\Rightarrow (C_1\times  (B_2\multimap C_2)), B_1\Rightarrow C_2))\\
        &= \Sigma_{\catC}\catL((A_1,A_2), (B_1,B_2)\Rightarrow (C_1,C_2)).\end{align*}\\[-22pt]
    \end{proof}
    We observe that
    we need $\catL$ to have biproducts (equivalently: to be $\CMon$-enriched) in order to show Cartesian closure.
    Furthermore, we need linear $\Rightarrow$-types and Cartesian $\multimap$-types to construct exponentials.
    Codually, we also obtain the Cartesian closure of $\Sigma_\catC \catL^{op}$.
    However, for concreteness, we give the proof explicitly.
\begin{theorem}\label{thm:grothendieck-ccc-contravariant}
    For a categorical model $\catL:\catC^{op}\to \Cat$ of the target language, $\Sigma_{\catC}\catL^{op}$ has:\begin{itemize}\item  terminal object $\terminal=(\terminal,\terminal)$ and binary products
    $(A_1,A_2)\times (B_1,B_2)=(A_1\times B_1,A_2\times  B_2)$;
    \item exponentials $(A_1,A_2)\Rightarrow (B_1,B_2)=
    (A_1\Rightarrow  (B_1\times (B_2\multimap A_2)), \copower{A_1}{B_2}).$
    \end{itemize}
    \end{theorem}
    \begin{proof}
        We have (natural) bijections
        \begin{align*}
            &\Sigma_{\catC}\catL^{op}((A_1,A_2), (\terminal,\terminal))
            =
            \catC(A_1,\terminal)\times \catL(A_1)(\terminal,A_2)
            \cong\terminal\times \terminal
            \cong \terminal
            \explainr{$\terminal$ terminal in $\catC$, initial in $\catL(A_1)$}\\
&\\[-6pt]
        &\Sigma_{\catC}\catL^{op}((A_1,A_2), (B_1\times C_1,B_2\times  C_2))
        =\catC(A_1,B_1\times C_1)\times \catL(A_1)(B_2\times C_2,A_2)\hspace{-60pt}\;\\
        &\cong
        \catC(A_1,B_1)\times \catC(A_1,C_1)\times \catL(A_1)(B_2,A_2)\times 
        \catL(A_1)(C_2,A_2)
        \explainr{$\times $ product in $\catC$, coproduct in $\catL(A_1)$}\\
        &= \Sigma_{\catC}\catL^{op}((A_1,A_2),(B_1,B_2))\times \Sigma_{\catC}\catL^{op}((A_1,A_2), (C_1,C_2))
        \\
        \\[-6pt]
            &\Sigma_{\catC}\catL^{op}((A_1,A_2)\times (B_1,B_2), (C_1,C_2)) =
            \Sigma_{\catC}\catL^{op}((A_1\times B_1,A_2\times  B_2), (C_1,C_2))\hspace{-60pt}\; \\
            &=\catC(A_1\times B_1, C_1)\times \catL(A_1\times B_1)(C_2, A_2\times B_2)\\
            &\cong
            \catC(A_1\times B_1, C_1)\times \catL(A_1\times B_1)(C_2, A_2)\times \catL(A_1\times B_1)(C_2, B_2)
            \explainr{$\times$ is product in $\catL(A_1\times B_1)$}\\
            &\cong            \catC(A_1\times B_1, C_1)\times\catC(A_1\times B_1,C_2\multimap B_2)\times 
            \catL(A_1\times B_1)(C_2, A_2)\hspace{-40pt}\;
             \explainr{Cartesian $\multimap$-types}\\
            &\cong
            \catC(A_1\times B_1, C_1\times  (C_2\multimap B_2))\times 
            \catL(A_1\times B_1)(C_2, A_2)
            \explainr{$\times$ is product in $\catC$}\\
            &\cong
            \catC(A_1,B_1\Rightarrow ( C_1\times  (C_2\multimap B_2)))\times 
            \catL(A_1\times B_1)(C_2, A_2)
            \explainr{$\Rightarrow$ is exponential in $\catC$}\\
            &\cong
            \catC(A_1,B_1\Rightarrow ( C_1\times  (C_2\multimap B_2)))\times 
            \catL(A_1)(\copower{B_1}{C_2}, A_2)
            \explainr{$\copower{(-)}{(-)}$-types} \\
            &=
            \Sigma_{\catC}\catL^{op}((A_1,A_2), (B_1\Rightarrow ( C_1\times  (C_2\multimap B_2)), \copower{B_1}{C_2}))\\
            &= 
            \Sigma_{\catC}\catL^{op}((A_1,A_2), (B_1,B_2)\Rightarrow (C_1,C_2)).\\[-22pt]
        \end{align*}
    \end{proof}
    Observe that we need the biproduct structure of $\catL$ to construct finite products in
    $\Sigma_{\catC}\catL^{op}$.
    Furthermore, we need Cartesian $\multimap$-types 
    and $\copower{(-)}{(-)}$-types, but not biproducts, to construct exponentials.

    Interestingly, the exponentials in $\Sigma_\catC\catL$ and $\Sigma_\catC\catL^{op}$ are not 
    fibred over $\catC$ (unlike their products, for example). Indeed $(A_1, A_2)\Rightarrow (B_1, B_2)$ has first 
    component not equal to $A_1\Rightarrow B_1$.
    In the context of automatic differentiation, this has the consequence that primals associated with
    values $f$ of function type 
    are not equal to $f$ itself. Instead, as we will see, they include both a copy of $f$ and a copy of its (transposed)
    derivative.
    These primals at higher-order types can be contrasted with the situation at first-order types, where values are equal to their associated primal,
    as a result of the finite products being fibred.

 \section{Novel AD Algorithms as Source-Code~Transformations}
\label{sec:ad-transformation}
As $\Sigma_{\CSyn}\LSyn$ and $\Sigma_{\CSyn}\LSyn^{op}$
are both Cartesian 
closed categories by Theorems \ref{thm:grothendieck-ccc-covariant} and \ref{thm:grothendieck-ccc-contravariant},
the universal property of the source language (Prop. \ref{prop:universal-property-syn})
gives us the following definition of forward and reverse mode CHAD as canonical homomorphic functors.
\begin{corollary}[Canonical definition of CHAD]\label{cor:chad-definition}
Once we fix compatible definitions 
$\Dsyn{\reals^n}$ and $\Dsyn{\op}$ (resp. $\Dsynrev{\reals^n}$ and $\Dsynrev{\op}$), 
we obtain a unique structure-preserving functor 
$$\Dsyn{-}:\Syn\to \Sigma_\CSyn \LSyn\qquad\qquad (\text{resp.\quad} \Dsynrev{-}:\Syn\to \Sigma_\CSyn \LSyn^{op}).$$
\end{corollary}
In this section, we discuss
\begin{itemize}
\item the interpretation of the above functors as a 
type-respecting code transformation;
\item how to give the basic definitions $\Dsyn{\reals^n}$, $\Dsynrev{\reals^n}$,  $\Dsyn{\op}$ and $\Dsynrev{\op}$;
\item what the induced AD 
definitions $\Dsyn{\trm}$ and $\Dsynrev{\trm}$ are for arbitrary source language programs $\trm$;
\item some consequences of the sharing of subexpressions that we have employed 
when defining the code transformations.
\end{itemize}

\subsection{Some Notation}
In the rest of this section, we use the following syntactic sugar:
\begin{itemize}
\item a notation for (linear) $n$-ary tuple types: $\tProd{\cty_1}{\ldots}{\cty_n}\defeq \bProd{\bProd{\bProd{\cty_1}{\cty_2}\cdots}{\cty_{n-1}}}{\cty_n}
$;
\item a notation for $n$-ary tuples: $\tTriple{\trm_1}{\cdots}{\trm_n}\defeq \tPair{\tPair{\tPair{\trm_1}{\trm_2}\cdots}{\trm_{n-1}}}{\trm_n}$;
\item given $\Gamma;\lvar:\cty\vdash \trm:\tProd{\cty[2]_1}{\cdots}{\cty[2]_n}$, we write $\Gamma;\lvar:\cty\vdash \tProj{i}(\trm):\cty[2]_i$ for the 
obvious $i$-th projection of $\trm$, which is constructed by repeatedly applying  
$\tFst$ and $\tSnd$ to $\trm$;
\item given $\Gamma;\lvar:\cty\vdash \trm:\cty[2]_i$, we write the $i$-th coprojection $\Gamma;\lvar:\cty\vdash 
\tCoProj{i}(\trm)\defeq \tTuple{\zero,\ldots,\zero,\trm,\zero,\ldots,\zero}
:\tProd{\cty[2]_1}{\cdots}{\cty[2]_n}$;
\item for a list $\var_1,\ldots,\var_n$ of distinct identifiers, we write $\idx{\var_i}{\var_1,\ldots, \var_n}\defeq i$ for the index of the identifier $\var_i$ in this list;
\item a $\mathbf{let}$-binding for tuples: $\pletin{\var}{\var[2]}{\trm}{\trm[2]}\defeq 
\letin{\var[3]}{\trm}{\letin{\var}{\tFst\var[3]}{\letin{\var[2]}{\tSnd\var[3]}{\trm[2]}}},$ where $\var[3]$ is a fresh variable.
\end{itemize}
Furthermore, all variables used in the source-code transformations below are assumed to be freshly chosen.

\subsection{$\Dsyn{-}$ and $\Dsynrev{-}$ as Type-Respecting Code Transformations}
Writing out the definitions of the categories $\Syn$, 
$\Sigma_\CSyn \LSyn$, $\Sigma_\CSyn \LSyn^{op}$, 
$\Dsyn{-}$ and $\Dsynrev{-}$ provide 
for each type $\ty$ of the source language (\S\ref{sec:language}) the following types in the target language (\S\ref{sec:minimal-linear-language}):
\begin{itemize}
    \item a Cartesian type $\Dsyn{\ty}_1$ 
    of forward mode primals;
    \item a linear type $\Dsyn{\ty}_2$ of forward mode tangents;
    \item a Cartesian type $\Dsynrev{\ty}_1$ 
    of reverse mode primals;
    \item a linear type $\Dsynrev{\ty}_2$ of reverse mode cotangents.
\end{itemize}
We can extend the actions of $\Dsyn{-}$ and $\Dsynrev{-}$ to typing contexts
$\Gamma=\var_1:\ty_1,\ldots,\var_n:\ty_n$ as
\begin{align*}
    &\Dsyn{\Gamma}_1\defeq \var_1:\Dsyn{\ty_1}_1,\ldots,\var_n:\Dsyn{\ty_n}_1\qquad
    \;&&\text{(a Cartesian typing context)}\\
    &\Dsyn{\Gamma}_2\defeq \tProd{\Dsyn{\ty_1}_2}{\cdots}{\Dsyn{\ty_n}_2}\qquad\;&&\text{(a linear type)}\\
    &\Dsynrev{\Gamma}_1\defeq \var_1:\Dsynrev{\ty_1}_1,\ldots,\var_n:\Dsynrev{\ty_n}_1\qquad
    \;&&\text{(a Cartesian typing context)}\\
    &\Dsynrev{\Gamma}_2\defeq \tProd{\Dsynrev{\ty_1}_2}{\cdots}{\Dsynrev{\ty_n}_2}\qquad\;&&\text{(a linear type)}.
    \end{align*}
Similarly, $\Dsyn{-}$ and $\Dsynrev{-}$ associate with each source-language 
program $\Gamma\vdash\trm:\ty$ the following programs in the target language (\S\ref{sec:minimal-linear-language}): 
\begin{itemize}
    \item a forward mode primal computation $\Dsyn{\Gamma}_1\vdash\Dsyn[\vGamma]{\trm}_1:\Dsyn{\ty}_1$;
    \item a forward mode tangent computation 
    $\Dsyn{\Gamma}_1;\lvar:\Dsyn{\Gamma}_2\vdash \Dsyn[\vGamma]{\trm}_2:\Dsyn{\ty}_2$;
    \item a reverse mode primal computation $\Dsynrev{\Gamma}_1\vdash\Dsynrev[\vGamma]{\trm}_1:\Dsynrev{\ty}_1$;
    \item a reverse mode cotangent computation 
    $\Dsynrev{\Gamma}_1;\lvar:\Dsynrev{\ty}_2\vdash \Dsynrev[\vGamma]{\trm}_2:\Dsynrev{\Gamma}_2$.
\end{itemize}
Here, we write $\vGamma$ for the list of identifiers $\var_1,\ldots,\var_n$ that occur in the typing context $\Gamma=\var_1:\ty_1,\ldots,\var_n:\ty_n$.
As we will see later, we need to know these context identifiers to define the code transformation.
Equivalently, we can pair up the primal and (co)tangent computations
as 
\begin{itemize}
\item a combined forward mode primal and tangent computation $\Dsyn{\Gamma}_1\vdash \Dsyn[\vGamma]{\trm}:\Dsyn{\ty}_1\t*(\Dsyn{\Gamma}_2\multimap\Dsyn{\ty}_2)$, where $\Dsyn[\vGamma]{\trm}\bepeq \tPair{\Dsyn[\vGamma]{\trm}_1}{\lfun{\lvar}\Dsyn[\vGamma]{\trm}_2}$;
\item a combined reverse mode primal and cotangent computation $\Dsynrev{\Gamma}_1\vdash \Dsynrev[\vGamma]{\trm}:\Dsynrev{\ty}_1\t*(\Dsynrev{\ty}_2\multimap \Dsynrev{\Gamma}_2)$, where $\Dsynrev[\vGamma]{\trm}\bepeq \tPair{\Dsynrev[\vGamma]{\trm}_1}{\lfun{\lvar}\Dsynrev[\vGamma]{\trm}_2}$.
\end{itemize}
We prefer to work with these combined primal and (co)tangent code transformations because doing so allows us to share common subexpressions between the primal and (co)tangent computations using $\mathbf{let}$-bindings.
Indeed, note that the universal property of $\Syn$ only defines the code transformations $\Dsyn{-}$ and $\Dsynrev{-}$ up to $\bepeq$.
In writing down the definitions of CHAD on programs, we make sure to 
choose sensible representatives of these $\beta\eta+$-equivalence classes
that share common subexpressions through $\mathbf{let}$-bindings.
While these $\mathbf{let}$-bindings naturally do not affect
correctness of the transformation, they let us avoid 
code explosion at compile time and unnecessary recomputation at run time.

Finally, because they are defined from a universal property,
our code transformations automatically 
respect equational reasoning in the sense that 
$\Gamma\vdash \trm\beeq \trm[2]:\ty$ implies that $\Dsyn[\vGamma]{\trm}\bepeq\Dsyn[\vGamma]{\trm[2]}$ and $\Dsynrev[\vGamma]{\trm}\bepeq\Dsynrev[\vGamma]{\trm[2]}$.

\subsection{The Basic Definitions:  $\Dsyn{\reals^n}$, $\Dsynrev{\reals^n}$,  $\Dsyn{\op}$ and $\Dsynrev{\op}$}
In \S\ref{sec:minimal-linear-language}, we have assumed that there are suitable terms (for example, linear operations) 
\begin{align*}
&\var_1:\reals^{n_1}, \cdots,\var_k: \reals^{n_k}\;\;;\;\; \lvar:\creals^{n_1}\t* \cdots\t* \creals^{n_k} &\hspace{-5pt}&\vdash D\op(\var_1,\ldots,\var_k;\lvar)\hspace{-5pt} &:\; &\creals^m\\
&\var_1:\reals^{n_1}, \cdots,\var_k: \reals^{n_k}\;\;;\;\; \lvar: \creals^m &\hspace{-5pt}&\vdash \transpose{D\op}(\var_1,\ldots,\var_k;\lvar)\hspace{-5pt}& :\;&\creals^{n_1}\t* \cdots\t* \creals^{n_k}
\end{align*} to represent the forward and 
reverse mode derivatives of the primitive operations $\op\in\Op_{n_1,...,n_k}^m$.
Using these, we define
\begin{flalign*}
    &\Dsyn{\reals^n}_1 &&\defeq&& \reals^n\\
    & {\Dsyn{\reals^n}_2} && \defeq && \creals^n\\
   \\
        &\Dsyn[\vGamma]{\op(\trm_1,\ldots,\trm_k)} && \defeq && \pletin{\var_1}{\var_1'}{\Dsyn[\vGamma]{\trm_1}}{\cdots\pletin{\var_k}{\var_k'}{\Dsyn[\vGamma]{\trm_k}}{\\
       && &&&\tPair{\op(\var_1,\ldots,\var_k)}{\lfun\lvar D\op(\var_1,\ldots,\var_k;\tTriple{\lapp{\var_1'}{\lvar}}{\ldots}{\lapp{\var_k'}{\lvar}})}}}
        \\
            &\Dsynrev{\reals^n}_1 && \defeq && \reals^n\\
            & \Dsynrev{\reals^n}_2 && \defeq && \creals^n\\
       \\
            &\Dsynrev[\vGamma]{\op(\trm_1,\ldots,\trm_k)} && \defeq && \pletin{\var_1}{\var_1'}{\Dsynrev[\vGamma]{\trm_1}}{\cdots
            \pletin{\var_k}{\var_k'}{\Dsynrev[\vGamma]{\trm_k}}{\\
            &&&&&\tPair{\op(\var_1,\ldots,\var_k)}{\lfun\lvar \letin{\lvar}{\transpose{D\op}(\var_1,\ldots,\var_k;\lvar)}{\lapp{\var_1'}{(\tProj{1}{\lvar})}+\cdots+\lapp{\var_k'}{(\tProj{k}{\lvar})}}}}}
        \end{flalign*}
    These basic definitions of CHAD for primitive operations implement the well-known multivariate chain rules for (transposed) derivatives of \S\S\ref{ssec:pairing-sharing}.

    For the AD transformations to be correct, it is important that these derivatives of language
    primitives are implemented correctly in the sense that
    $$
\sem{\var_1,\ldots,\var_k;\lvar\vdash D\op(\var_1,\ldots,\var_k;\lvar)}=D\sem{\op}\qquad \sem{\var_1,\ldots,\var_k;\lvar\vdash \transpose{D\op}(\var_1,\ldots,\var_k;\lvar)}=\transpose{D\sem{\op}}.
    $$
    For example, for elementwise multiplication $(*)\in\Op_{n,n}^n$, which we interpret as the usual elementwise product $\sem{(*)}\defeq (*): \RR^n\times \RR^n\to\RR^n$, we need, by the product rule for differentiation, that
\begin{align*}&\sem{D(*)(\var_1,\var_2;\lvar)}((a_1, a_2), (b_1, b_2))&&=&a_1 * b_2 + a_2 * b_1\\
&\sem{\transpose{D(*)}(\var_1,\var_2;\lvar)}((a_1, a_2),b)&&=&(a_2 * b, a_1 * b).\end{align*}
    By Prop. \ref{prop:universal-property-syn}, the extensions of the AD transformations $\Dsynsymbol$ and $\Dsynrevsymbol$ 
    to the full source language are now canonically determined, as the unique 
    Cartesian closed functors that extend these basic definitions.

    \subsection{The Implied Forward Mode CHAD Definitions}\label{ssec:forward-chad-defs}
We define the types of (forward mode) primals $\Dsyn{\ty}_1$ and tangents $\Dsyn{\ty}_2$ associated with a type $\ty$ as follows:
\begin{flalign*}
&\Dsyn{\Unit}_1 &&\defeq&& \Unit&& \Dsyn{\Unit }_2 &&\defeq&& \lUnit \\
&\Dsyn{\ty\t*\ty[2]}_1 &&\defeq&& \Dsyn{\ty}_1\t*\Dsyn{\ty[2]}_1&& \Dsyn{\ty\t*\ty[2]}_2 &&\defeq&&  \Dsyn{\ty}_2\t*\Dsyn{\ty[2]}_2\\
&\Dsyn{\ty\To\ty[2]}_1 &&\defeq&& \Dsyn{\ty}_1\To(\Dsyn{\ty[2]}_1\t* (\Dsyn{\ty}_2\multimap 
\Dsyn{\ty[2]}_2))\qquad\qquad\qquad\qquad&&\Dsyn{\ty\To\ty[2]}_2 &&\defeq&& \Dsyn{\ty}_1\To\Dsyn{\ty[2]}_2.
\end{flalign*}
Observe that the type of primals associated with a function type is not equal to the original type.
This is a consequence of the non-fibred nature of the exponentials in the $\Sigma$-type category $\Sigma_\CSyn\LSyn$ (\S\ref{sec:self-dualization}).

For programs $\trm$, we define their efficient CHAD transformation $\Dsyn[\vGamma]{\trm}$ 
 as follows:
    \begin{flalign*}
    &\Dsyn[\vGamma]{\var} &&\defeq  && \tPair{\var}{\lfun{\lvar}\tProj{\idx{\var}{\vGamma}}(\lvar)}
    \\
    &\Dsyn[\vGamma]{\letin{\var}{\trm}{\trm[2]}} &&\defeq &&\pletin{\var}{\var'}{\Dsyn[\vGamma]{\trm}}{ \pletin{\var[2]}{\var[2]'}{\Dsyn[\vGamma,\var]{\trm[2]}}{\tPair{\var[2]}{\lfun\lvar
    \lapp{\var[2]'}{\tPair{\lvar}{\lapp{\var'}{\lvar}}}}}}
    \\
    &\Dsyn[\vGamma]{\tUnit}  &&\defeq &&\tPair{\tUnit}{\lfun\lvar\tUnit}
    \\
    &\Dsyn[\vGamma]{\tPair{\trm}{\trm[2]}} &&\defeq &&
    \pletin{\var}{\var'}{\Dsyn[\vGamma]{\trm}}{ \pletin{\var[2]}{\var[2]'}{\Dsyn[\vGamma]{\trm[2]}}{\tPair{\tPair{\var}{\var[2]}}{\lfun\lvar \tPair{\lapp{\var'}\lvar}{\lapp{\var[2]'}\lvar}}}}
    \\&
    \Dsyn[\vGamma]{\tFst\trm} &&\defeq &&
    \pletin{\var}{\var'}{\Dsyn[\vGamma]{\trm}}
    {\tPair{\tFst\var}{\lfun\lvar \tFst(\lapp{\var'}{\lvar})}}
    \\&
    \Dsyn[\vGamma]{\tSnd\trm} &&\defeq &&
    \pletin{\var}{\var'}{\Dsyn[\vGamma]{\trm}}
    {\tPair{\tSnd\var}{\lfun\lvar \tSnd(\lapp{\var'}{\lvar})}}
    \\
        &
    \Dsyn[\vGamma]{\fun\var\trm}&&\defeq &&
        \letin{\var[2]}{\fun\var\Dsyn[\vGamma,\var]{\trm}}{\tPair{\fun\var
            \pletin{\var[3]}{\var[3]'}{\var[2]\,\var}{
            \tPair{\var[3]}{\lfun\lvar\lapp{\var[3]'}{\tPair{\zero}{\lvar}}}}}{\lfun\lvar\fun\var\lapp{(\tSnd(\var[2]\,\var))}{\tPair{\lvar}{\zero}}}
        }
    \\&
    \Dsyn[\vGamma]{\trm\,\trm[2]} &&\defeq &&
    \pletin{\var}{\var'_{\text{ctx}}}{\Dsyn[\vGamma]{\trm}}{
    \pletin{\var[2]}{\var[2]'}{\Dsyn[\vGamma]{\trm[2]}}{
    \pletin{\var[3]}{\var'_{\text{arg}}}{\var\,\var[2]}{
    \tPair{\var[3]}{\lfun\lvar (\lapp{\var'_{\text{ctx}}}{\lvar})\,\var[2]+
    \lapp{\var'_{\text{arg}}}{(\lapp{\var[2]'}{\lvar})}}}}}.
    \end{flalign*}
We explain and justify these transformations in the next subsection after discussing the transformations for reverse CHAD.

\subsection{The Implied Reverse Mode CHAD Definitions}\label{ssec:reverse-chad-defs}
We define the types of (reverse mode) primals $\Dsynrev{\ty}_1$ and cotangents $\Dsynrev{\ty}_2$ associated with a type $\ty$ as follows:
\begin{align*}
&\Dsynrev{\Unit}_1 &&\defeq&& \Unit &&  \Dsynrev{\Unit }_2 \defeq \lUnit\\
&\Dsynrev{\ty\t*\ty[2]}_1 &&\defeq&& \Dsynrev{\ty}_1\t*\Dsynrev{\ty[2]}_1
&& \Dsynrev{\ty\t*\ty[2]}_2 &&\defeq&&  \Dsynrev{\ty}_2\t*\Dsynrev{\ty[2]}_2\\
&\Dsynrev{\ty\To\ty[2]}_1 &&\defeq&& \Dsynrev{\ty}_1\To(
\Dsynrev{\ty[2]}_1\t*(\Dsynrev{\ty[2]}_2\multimap 
\Dsynrev{\ty}_2))\qquad\qquad\qquad\qquad
&&\Dsynrev{\ty\To\ty[2]}_2 &&\defeq&&  \copower{\Dsynrev{\ty}_1}{\Dsynrev{\ty[2]}_2}.
\end{align*}
Again, we associate a non-trivial type of primals to function types because exponentials are not fibred in $\Sigma_\CSyn\LSyn^{op}$ (\S\ref{sec:self-dualization}).

For programs $\trm$, we define their efficient CHAD transformation $\Dsynrev[\vGamma]{\trm}$ 
as follows:
\begin{flalign*}
&\Dsynrev[\vGamma]{\var} &&\defeq &&  \tPair{\var}{\lfun{\lvar} \tCoProj{\idx{\var}{\vGamma}}(\lvar)}
\\
    &
\Dsynrev[\vGamma]{\letin{\var}{\trm}{\trm[2]}}  && 
\defeq &&
\pletin{\var}{\var'}{\Dsynrev[\vGamma]{\trm}}{
    \pletin{\var[2]}{\var[2]'}{\Dsynrev[\vGamma,\var]{\trm[2]}}{
        \tPair{\var[2]}{\lfun\lvar 
        \letin{\lvar}{\lapp{\var[2]'}{\lvar}}{
            \tFst\lvar+\lapp{\var'}{(\tSnd \lvar)}
        }}
    }}
\\& 
\Dsynrev[\vGamma]{\tUnit}  &&\defeq && \tPair{\tUnit}{\lfun\lvar\zero}\\&
\Dsynrev[\vGamma]{\tPair{\trm}{\trm[2]}} &&\defeq && 
\pletin{\var}{\var'}{\Dsynrev[\vGamma]{\trm}}{ 
\pletin{\var[2]}{\var[2]'}{\Dsynrev[\vGamma]{\trm[2]}}{
\tPair{\tPair{\var}{\var[2]}}{\lfun\lvar \lapp{\var'}{(\tFst\lvar)} + \lapp{\var[2]'}{(\tSnd \lvar)}}}}
\\&
\Dsynrev[\vGamma]{\tFst\trm} &&\defeq && 
\pletin{\var}{\var'}{\Dsynrev[\vGamma]{\trm}}
{\tPair{\tFst\var}{\lfun\lvar \lapp{\var'}{\tPair{\lvar}{\zero}}}}
\\&
\Dsynrev[\vGamma]{\tSnd\trm} &&\defeq && 
\pletin{\var}{\var'}{\Dsynrev[\vGamma]{\trm}}
{\tPair{\tSnd\var}{\lfun\lvar \lapp{\var'}{\tPair{\zero}{\lvar}}}}
\\&   
\Dsynrev[\vGamma]{\fun\var\trm} &&\defeq && 
\letin{\var[2]}{\fun\var\Dsynrev[\vGamma,\var]{\trm}}{\\ &&&&&
\tPair{\fun\var
\pletin{\var[3]}{\var[3]'}{\var[2]\,\var}{
\tPair{\var[3]}{\lfun\lvar\tSnd(\lapp{\var[3]'}{\lvar})}}}
{\lfun\lvar  \tensMatch{\lvar}{\var}{\lvar}{
    \tFst(\lapp{(\tSnd(\var[2]\,\var))}{\lvar})} }
}
\\&
\Dsynrev[\vGamma]{\trm\,\trm[2]} &&\defeq && 
\pletin{\var}{\var'_{\text{ctx}}}{\Dsynrev[\vGamma]{\trm}}{
\pletin{\var[2]}{\var[2]'}{\Dsynrev[\vGamma]{\trm[2]}}{
\pletin{\var[3]}{\var'_{\text{arg}}}{\var\,\var[2]}{
\tPair{\var[3]}{\lfun\lvar  \lapp{\var'_{\text{ctx}}}{(!\var[2]\otimes \lvar)} + \lapp{\var[2]'}{(\lapp{\var'_{\text{arg}}}{\lvar})}}}}}.
\end{flalign*}
We now explain and justify the forward and reverse CHAD transformations.
The transformations for variables, tuples, and projections implement the well-known multivariate calculus facts about (transposed) derivatives of \smooth{} functions into and out of products of spaces.
The transformations for $\mathbf{let}$-bindings add to that the chain rules for $\Dsemsymbol$  and $\Dsemrevsymbol$ of \S\S\ref{ssec:pairing-sharing}. 
The transformations for $\lambda$-abstractions split the derivative of a closure $\fun{\var}\trm$ into the derivative $\var[3]'$ with respect to the function argument $\var$ and the derivative $\tSnd(\var[2]\,\var)$ with respect to the captured context variables; they store $\var[3]'$ together with the primal computation $\var[3]$ of $\fun{\var}\trm$ in the primal associated with the closure and they store $\tSnd(\var[2]\,\var)$ in the (co)tangent associated with the closure.
Conversely, the transformations for evaluations extract those two components of the (transposed) derivative $\var'_{\text{ctx}}$ (with respect to context variables) and $\var'_{\text{arg}}$ (with respect to the function argument) from the (co)tangent and primal, respectively, and recombine them to correctly propagate (co)tangent contributions from both sources.

\subsection{Sharing of Common Subexpressions}
Through careful use of $\mathbf{let}$-bindings, we have taken care to ensure that the CHAD code transformations we specified have the following good property:
for every program former $C[\trm_1,\ldots,\trm_n]$ that takes $n$ subprograms
$\trm_1$,\ldots, $\trm_n$ (for example, function application $\trm\,\trm[2]$
takes two subprograms $\trm$ and $\trm[2]$), we have that $\Dsyn[\vGamma]{C[\trm_1,\ldots,\trm_n]}$ uses $\Dsyn[\vGamma_i]{\trm_i}$ exactly once in its definition for each subprogram $\trm_i$, for some list of identifiers $\vGamma_i$.
Similarly, $\Dsynrev[\vGamma]{C[\trm_1,\ldots,\trm_n]}$ uses $\Dsynrev[\vGamma_i]{\trm_i}$ exactly once in its definition for each subprogram $\trm_i$, which demonstrates the following.
\begin{corollary}[No code explosion]
The code sizes of the forward and reverse CHAD-transformed programs $\Dsyn[\vGamma]{\trm}$ and $\Dsynrev[\vGamma]{\trm}$ both grow linearly in the size 
of the original source program $\trm$. 
\end{corollary}
This compile-time complexity property is crucial if we are to keep 
compilation times and executable sizes manageable when performing AD 
on large codebases.

Of course, our use of $\mathbf{let}$-bindings has the additional run-time benefit that repeated subcomputations are performed only once 
and their stored results are shared, rather than recomputed whenever their results are needed.
We have taken care to avoid any unnecessary computation in this way, which we hope will benefit the performance of CHAD in practice.
However, we leave a proper complexity and practical performance analysis to future work.
 \section{Proving Reverse and Forward AD Semantically Correct}\label{sec:glueing-correctness}
In this section, we show that the CHAD code transformations 
described in \S\ref{sec:ad-transformation} correctly compute 
mathematical derivatives (Thm. \ref{thm:AD-correctness}).
The proof consists mainly of an (open) logical relations argument over the
semantics in the Cartesian closed categories
$\Set\times \Sigma_{\Set}\CMon$ and $\Set\times \Sigma_{\Set}\CMon^{op}$.
The intuition behind the proof is as follows:
\begin{itemize}
    \item the logical relations relate \smooth{} functions $\RR^d\to\sem{\ty}$ to associated primal and (co)tangent functions;
    \item the semantics $\sem{\trm}\times \sem{\Dsyn[\vGamma]{\trm}}$ 
    and $\sem{\trm}\times \sem{\Dsynrev[\vGamma]{\trm}}$  of forward and reverse mode CHAD respect the logical relations;
    \item therefore, by basic results in calculus, they must equal the derivative and transposed derivative of $\sem{\trm}$.
\end{itemize}
This logical relations proof can be phrased in elementary terms, but the resulting argument is 
technical and would be hard to discover.
Instead, we prefer to phrase it in terms of a categorical subsconing 
construction, a more abstract and elegant perspective on logical relations.
We discovered the proof by taking this categorical perspective,
and, while we have verified the elementary argument (see \S\S\ref{ssec:key-ideas-chad-correct}),
we would not otherwise have found it.

\subsection{Preliminaries}
\subsubsection{Subsconing}\label{sssec:subsconing}
Logical relations arguments provide a powerful proof technique for demonstrating
properties of typed programs.
The arguments proceed by induction on the structure of types.
Here, we briefly review the basics of categorical logical relations arguments,
or \emph{subsconing constructions}.
We restrict to the level of generality that we need here, but we point out 
that the theory applies much more generally.

Consider a Cartesian closed category $(\catC,\terminal,\times,\Rightarrow)$.
Suppose that we are given a functor 
$F:\catC\to\Set$ to the category $\Set$ of sets and functions
that preserves finite products in the sense that $F(\terminal)\cong\terminal$
and $F(C\times C')\cong F(C)\times F(C')$.
Then, we can form the \emph{subscone} of $F$, or category of logical relations 
over $F$, which is Cartesian closed, with a
faithful Cartesian closed functor $\pi_1$ to $\catC$ that forgets about the predicates \cite{johnstone-lack-sobocinski}:
\begin{itemize}
\item objects are pairs $(C,P)$ of an object $C$ of $\catC$ and a predicate 
$P\subseteq FC$;
\item morphisms $(C,P)\to (C',P')$ are $\catC$ morphisms $f:C\to C'$
that respect the predicates in the sense that $F(f)(P)\subseteq P'$;
\item identities and composition are as in $\catC$;
\item $(\terminal, F\terminal)$ is the terminal object, and binary products and exponentials are given by
\begin{align*}(C,P)\times (C',P')&= (C\times C', \set{\alpha\in F(C\times C')\mid 
F(\pi_1)(\alpha)\in P, F(\pi_2)(\alpha)\in P'})\\
(C,P)\Rightarrow (C', P')&= (C\Rightarrow C',
\{\alpha\in F(C\Rightarrow C')\mid 
\forall \gamma\in F((C\Rightarrow C')\times C).\;
\big(F(\pi_1)(\gamma)=\alpha \textnormal{ and } F(\pi_2)(\gamma)\in P\big)\Rightarrow F(\ev)(\gamma)\in P'
\}
).\end{align*}
\end{itemize}

In typical applications, $\catC$ can be the syntactic category of a 
language (like $\Syn$), the codomain of a denotational semantics
$\sem{-}$ (like $\Set$), or a product of the above, if we want to 
consider $n$-ary logical relations.
Typically, $F$ tends to be a hom-functor (which always preserves products), like $\catC(\terminal,-)$ or 
$\catC(C_0,-)$, for some important object $C_0$.
When applied to the syntactic category $\Syn$ and $F=\Syn(\Unit,-)$, the 
formulae for products and exponentials in the subscone clearly
reproduce the usual recipes in traditional, syntactic logical relations arguments.
In this sense, subsconing generalizes standard logical relations methods.

\subsection{Subsconing for Correctness of AD}
We apply the subsconing construction above to
\[
    \begin{array}{lll}
\catC=\Set\times \Sigma_{\Set}\CMon& F=\Set\times \Sigma_{\Set}\CMon((\RR^d,(\RR^d,\cRR^d)),-)&\hspace{-3pt}\textnormal{(forward AD)}\\
\catC=\Set\times \Sigma_{\Set}\CMon^{op}\qquad\qquad & F=\Set\times \Sigma_{\Set}\CMon^{op}((\RR^d,(\RR^d,\cRR^d)),-)\qquad\qquad&\textnormal{(reverse AD)},
    \end{array}
\]
where we note that $\Set$, $\Sigma_{\Set}\CMon$, and $\Sigma_{\Set}\CMon^{op}$
are Cartesian closed (given the arguments of \S\ref{sec:semantics} and \S\ref{sec:self-dualization})
and that the product of Cartesian closed categories is again Cartesian closed.
Let us write $\Gl$ and $\GlRev$, respectively, for the resulting categories of logical 
relations.

Since $\Gl$ and $\GlRev$ are Cartesian closed, we obtain unique 
Cartesian closed functors $\semgl{-}:\Syn\to\Gl$ and $\semglrev{-}:\Syn\to\GlRev$,
by the universal property of $\Syn$ (\S\ref{sec:language}),
once we fix an interpretation of $\reals^n$ and all operations $\op$.
We write $P_{\ty}^f$ and $P_{\ty}^r$, respectively, for the relations 
 $\pi_2\semgl{\ty}$ and $\pi_2\semglrev{\ty}$.
Let us~interpret
\begin{align*}
&\semgl{\reals^n}\defeq (((\RR^n,(\RR^n,\cRR^n)), \set{(f,(g,h))\mid f\textnormal{ is \smooth{}}, f=g\textnormal{ and } 
h=Df
}))\\
&\semglrev{\reals^n}\defeq (((\RR^n,(\RR^n,\cRR^n)),\{(f,(g,h))\mid f\textnormal{ is \smooth{}},  f=g\textnormal{ and } 
h=\transpose{Df}
\}))\\
&\semgl{\op}\defeq (\sem{\op}, (\sem{\op}, \sem{D\op}))\qquad
\semglrev{\op}\defeq (\sem{\op}, (\sem{\op}, \sem{\transpose{D\op}})),
\end{align*}
where we write $Df$ for the semantic derivative of $f$ and $\transpose{(-)}$ for the matrix transpose (see \S\ref{sec:semantics}).

\begin{lemma}
These definitions extend uniquely to define Cartesian closed functors $$\semgl{-}:\Syn\to\Gl\qquad\textnormal{ and }\qquad\semglrev{-}:\Syn\to\GlRev.$$
\end{lemma}
\begin{proof}
This follows from the universal property of $\Syn$ (Prop. \ref{prop:universal-property-syn}) once we verify that $(\sem{\op}, (\sem{\op}, \sem{D\op}))$
and $(\sem{\op}, (\sem{\op}, \sem{\transpose{D\op}}))$
respect the logical relations $P^f$ and $P^r$, respectively.
This preservation of the relations follows immediately from the chain rule for multivariate differentiation, 
provided we have implemented our derivatives correctly for the basic operations~$\op$, in the sense that
\begin{align*}
&\sem{\var;\var[2]\vdash D\op(\var;\var[2])}=D\sem{\op}\;\;\quad\qquad\textnormal{and}\quad\qquad\;\;
\sem{\var;\var[2]\vdash \transpose{D\op}(\var;\var[2])}=\transpose{D\sem{\op}}.
\end{align*}

Writing $\reals^{n_1,...,n_k}\!\defeq\! \reals^{n_1}\t*\cdots\t*\reals^{n_k}$
and $\RR^{n_1,...,n_k}\!\defeq\! \RR^{n_1}\times\cdots\times \RR^{n_k}$,
we compute
\begin{align*}
    &\semgl{\reals^{n_1,...,n_k}}\!=\! ((\RR^{n_1,...,n_k},(\RR^{n_1,...,n_k},\cRR^{n_1,...,n_k})), \set{(f,(g,h))\mid f\textnormal{ is \smooth{}}, f=g, 
    h=Df
    })\\
    &\semglrev{\reals^{n_1,...,n_k}}\!=\! ((\RR^{n_1,...,n_k},(\RR^{n_1,...,n_k},\cRR^{n_1,...,n_k})),  \{(f,(g,h))\mid f\textnormal{ is \smooth{}}, f=g, 
    h=\transpose{Df}\})
\end{align*}
since derivatives of tuple-valued functions are computed component-wise.
(In fact, the corresponding facts hold more generally for any first-order type,
as an iterated product of $\reals^n$.)
Suppose that $(f,(g,h))\in P^f_{\reals^{n_1,...,n_k}}$,
i.e. $f$ is \smooth{}, $g=f$ and $h=Df$.
Then, using the chain rule in the last step, we have
\begin{align*}
 &  (f,(g,h));(\sem{\op},(\sem{\op},\sem{D\op})) \\
 &  = (f,(f,Df));(\sem{\op},(\sem{{\op}},\sem{\var;\var[2]\vdash D\op(\var;\var[2])}))\\
 &  = (f,(f,Df));(\sem{\op},(\sem{\op},D\sem{\op})) \\
 &  = (f;\sem{\op},(f;\sem{\op}, x\mapsto r\mapsto D\sem{\op}(f(x))(Df(x)(r)) ))\\
 &  = (f;\sem{\op},(f;\sem{\op}, D(f;\sem{\op}) ))\in P_{\reals^m}^f.
\end{align*}
Similarly, if $(f, (g,h))\in P^r_{\reals^{n_1,...,n_k}}$, then by the chain rule 
and linear algebra
\begin{align*}
 &  (f,(g,h));(\sem{\op},(\sem{\op},\sem{\transpose{D\op}})) \\
 &  = (f,(f,\transpose{Df}));(\sem{\op},(\sem{{\op}},\sem{\var;\var[2]\vdash \transpose{D\op}(\var;\var[2])})) \\
 &  = (f,(f,\transpose{Df}));(\sem{\op},(\sem{\op},\transpose{D\sem{\op}})) \\
 &  = (f;\sem{\op},(f;\sem{\op}, x\mapsto v\mapsto 
    \transpose{Df}(x)(\transpose{D\sem{\op}}(f(x))(v))
    )) \\
 &  = (f;\sem{\op},(f;\sem{\op}, x\mapsto v\mapsto 
    \transpose{Df(x);D\sem{\op}(f(x))}
    (v))
    ) \\
 &  = (f;\sem{\op},(f;\sem{\op}, \transpose{D(f;\sem{\op})} ))\in P_{\reals^m}^r.
\end{align*}
\noindent Consequently, we obtain our unique Cartesian closed functors $\semgl{-}$ and $\semglrev{-}$.
\end{proof}

Furthermore, observe that $\Sigma_{\sem{-}}\sem{-}(\trm_1,\trm_2)\defeq (\sem{\trm_1},\sem{\trm_2})$
defines a Cartesian closed functor $\Sigma_{\sem{-}}\sem{-}:\Sigma_{\CSyn}\LSyn\to \Sigma_{\Set}\CMon$.
Similarly, we get a Cartesian closed functor
$\Sigma_{\sem{-}}\sem{-}^{op}:\Sigma_{\CSyn}\LSyn^{op}\to \Sigma_{\Set}\CMon^{op}$.
As a consequence,
both squares below commute.
\begin{figure}[!h]
\begin{tikzcd}
    \Syn \arrow[r, "\sPair{\id}{\Dsynsymbol}"] \arrow[dd, "\semgl{-}"'] & \Syn\times \Sigma_{\CSyn}\LSyn \arrow[dd, "\sem{-}\times\Sigma_{\sem{-}}\sem{-}"] & \qquad\qquad & \Syn \arrow[r, "\sPair\id\Dsynrevsymbol"] \arrow[dd, "\semglrev{-}"'] & \Syn\times\Sigma_{\CSyn}\LSyn^{op} \arrow[dd, "\sem{-}\times\Sigma_{\sem{-}}\sem{-}^{op}"] \\
    \\
    \Gl \arrow[r, "\pi_1"']                                            & \Set\times\Sigma_{\Set}\CMon                                               &\qquad\qquad & \GlRev \arrow[r, "\pi_1"']                                           & \Set\times\Sigma_{\Set}\CMon^{op}.                                                       
    \end{tikzcd}\end{figure}

\noindent Indeed, going around the squares in both directions defines Cartesian closed functors
that agree on their action on the generators $\reals^n$ and $\op$ of the Cartesian closed category $\Syn$.
\begin{corollary}\label{cor:respecting-logical-relation}
    For any source language (\S\ref{sec:language}) program $\Gamma\vdash\trm:\ty$, 
    $(\sem{\trm}, (\sem{\Dsyn[\vGamma]{\trm}_1},\sem{\Dsyn[\vGamma]{\trm}_2}))$ is a morphism in $\Gl$
and therefore respects the logical relations $P^f$.
Similarly, $(\sem{\trm}, (\sem{\Dsynrev[\vGamma]{\trm}_1},\sem{\Dsynrev[\vGamma]{\trm}_2}))$ is a morphism in $\GlRev$
and therefore respects the logical relations $P^r$.
\end{corollary}

Most of the work is now in place to show correctness of 
AD.
We finish the proof below.
To ease notation, we work with terms in a context with a single type.
Doing so is not a restriction as our language has products,
and the theorem holds for arbitrary terms between 
first-order types.
\begin{theorem}[Correctness of AD]\label{thm:AD-correctness}
For programs $\Gamma\vdash \trm:\ty[2]$ where $\ty[2]$ and all types $\ty_i$ in $\Gamma=\var_1:\ty_1,\ldots,\var_n:\ty_n$ are first-order types, $\sem{\trm}$ is \smooth{} and
$$
\sem{\Dsyn[\vGamma]{\trm}_1}=\sem{\trm}\qquad \sem{\Dsyn[\vGamma]{\trm}_2}=D\sem{\trm}
\qquad\sem{\Dsynrev[\vGamma]{\trm}_1}=\sem{\trm}\qquad\sem{\Dsynrev[\vGamma]{\trm}_2}=\transpose{D\sem{\trm}},
$$
where we write $D$ and $\transpose{(-)}$ for the usual calculus derivative and matrix transpose.
Hence,
$$
\sem{\Dsyn[\vGamma]{\trm}}=(\sem{\trm},D\sem{\trm})\qquad \text{and}\qquad 
\sem{\Dsynrev[\vGamma]{\trm}}=(\sem{\trm},\transpose{D\sem{\trm}})
$$
\end{theorem}
\begin{proof}
    Since our language has tuples, we may assume without loss of generality that $\Gamma=\var:\ty$.
    We use the logical relations for general $d$ to show differentiability of all programs.
    Next, the case of $d=1$ suffices to show that CHAD computes correct derivatives.

    First, we observe that $\sem{\trm}$ sends \smooth{} functions $\RR^d\to \sem{\ty}$ to \smooth{} functions $\RR^d\to\sem{\ty[2]}$, as $\trm$ respects the logical relations.
    Observing that $\sem{\ty}\cong \RR^N$ for some $N$, as $\ty$ is a first-order type, we can choose $d=N$.
    Then, $P_{\ty}^f$ contains $(f,(f,Df))$ for a \smooth{} isomorphism $f$.
    It therefore follows that $\sem{\trm}$ is \smooth{}.

    Second, we focus on the correctness of forward AD, $\Dsynsymbol$.\\
    Let $x\in \sem{\Dsyn{\ty}_1}=\sem{\ty}\cong \RR^N$ and $v\in\sem{\Dsyn{\ty}_2}\cong \cRR^N$ (for some $N$).
    Then, there is a \smooth{} curve $\gamma:\RR\to \sem{\ty}$ such that
    $\gamma(0)=x$ and $D\gamma(0)(1)=v$.
    Clearly, $(\gamma,(\gamma, D\gamma))\in P_{\ty}^f$ (for $d=1$).

    As $(\sem{\trm}, (\sem{\Dsyn[\vGamma]{\trm}_1},\sem{\Dsyn[\vGamma]{\trm}_2}))$ respects the logical relation $P^f$ by 
    Cor. \ref{cor:respecting-logical-relation},
    we have
    \begin{align*}&(\gamma;\sem{\trm}, (\gamma;\sem{\Dsyn[\vGamma]{\trm}_1},x\mapsto r\mapsto \sem{\Dsyn[\vGamma]{\trm}_2}(\gamma(x))(D\gamma(x)(r))))=
    (\gamma,(\gamma,D\gamma));(\sem{\trm}, (\sem{\Dsyn[\vGamma]{\trm}_1},\sem{\Dsyn[\vGamma]{\trm}_2}))
    \in P^f_{\ty[2]},
    \end{align*}
    where we use the definition of composition in $\Set\times \Sigma_{\Set}\CMon$.
    Therefore, $$\gamma;\sem{\trm}=\gamma;\sem{\Dsyn[\vGamma]{\trm}_1}$$ and, by the chain rule,
    \begin{align*}
        x\mapsto r\mapsto D\sem{\trm}(\gamma(x))(D\gamma(x)(r))
        &=D(\gamma;\sem{\trm})=
     x\mapsto r\mapsto \sem{\Dsyn[\vGamma]{\trm}_2}(\gamma(x))(D\gamma(x)(r)).\end{align*}
    Evaluating the former at $0$ gives $\sem{\trm}(x)=\sem{\Dsyn[\vGamma]{\trm}_1}(x)$.
    Similarly, evaluating the latter at $0$ and $1$ gives $D\sem{\trm}(x)(v)= \sem{\Dsyn[\vGamma]{\trm}_2}(x)(v)$.
\\
\\
    Third, we turn to the correctness of reverse AD, $\Dsynrevsymbol$.\\
    Let $x\in \sem{\Dsynrev{\ty}_1}=\sem{\ty}\cong \RR^N$
    and $v\in\sem{\Dsynrev{\ty[2]}_2}\cong \cRR^M$ (for some $N$ and $M$).
    Let $\gamma_i:\RR\to \sem{\ty}$ be a \smooth{} curve such that
    $\gamma_i(0)=x$ and $D\gamma_i(0)(1)=e_i$, where we write $e_i$ for the $i$-th standard basis vector 
    of $\sem{\Dsynrev{\ty}_2}\cong \cRR^N$.
    Clearly, $(\gamma_i,(\gamma_i, \transpose{D\gamma_i}))\in P_{\ty}^r$ (for $d=1$).

    As $(\sem{\trm}, (\sem{\Dsynrev[\vGamma]{\trm}_1},\sem{\Dsynrev[\vGamma]{\trm}_2}))$ respects the logical relation $P^r$ by Cor. \ref{cor:respecting-logical-relation},
    we have
    \begin{align*}&(\gamma_i;\sem{\trm}, (\gamma_i;\sem{\Dsynrev[\vGamma]{\trm}_1},
    x\mapsto w\mapsto \transpose{D\gamma_i(x)}(\sem{\Dsynrev[\vGamma]{\trm}_2}(\gamma_i(x))(w))))= (\gamma_i,(\gamma_i,\transpose{D\gamma_i}));(\sem{\trm}, (\sem{\Dsynrev[\vGamma]{\trm}_1},\sem{\Dsynrev[\vGamma]{\trm}_2}))\in P^r_{\ty[2]},\end{align*}
    using the definition of composition in $\Set\times \Sigma_{\Set}\CMon^{op}$.
    Consequently, $$\gamma_i;\sem{\trm}=\gamma_i;\sem{\Dsynrev[\vGamma]{\trm}_1}$$ and, by the chain rule,
    \begin{align*}x\mapsto w\mapsto \transpose{D\gamma_i(x)}(\transpose{D\sem{\trm}(\gamma_i(x))}(w))& =\transpose{D(\gamma_i;\sem{\trm})} =
    x\mapsto w\mapsto \transpose{D\gamma_i(x)}(\sem{\Dsynrev[\vGamma]{\trm}_2}(\gamma_i(x))(w)).\end{align*}
    Evaluating the former at $0$ gives $\sem{\trm}(x)=\sem{\Dsynrev[\vGamma]{\trm}_1}(x)$.
    Similarly, evaluating the latter at $0$ and $v$ gives
    $\innerprod{e_i}{\transpose{D\sem{\trm}(x)}(v)}= \innerprod{e_i}{\sem{\Dsynrev[\vGamma]{\trm}_2}(x)(v)}$.
    As this equation holds for all basis vectors $e_i$ of $\sem{\Dsynrev{\ty}_2}$, we find that
    \begin{align*}\transpose{D\sem{\trm}(x)}(v)&=
    \sum_{i=1}^N (\innerprod{e_i}{\transpose{D\sem{\trm}(x)}(v)})\cdot e_i
    =\sum_{i=1}^N (\innerprod{e_i}{\sem{\Dsynrev[\vGamma]{\trm}_2}(x)(v)})\cdot e_i
    = \sem{\Dsynrev[\vGamma]{\trm}_2}(x)(v).\end{align*}
\end{proof}
 \section{Practical Relevance and Implementation in Functional Languages}\label{sec:implementation}
Most popular functional languages, such as Haskell and OCaml, do not 
natively support linear types.
Thus, the transformations described in this paper may seem 
hard to implement. 
However, as we will argue in this section, we can easily implement 
the limited linear types 
used in phrasing the transformations
as abstract data types using only a basic module system, such as that of Haskell.
The key idea is that linear function types $\cty\multimap \cty[2]$ can be represented as 
plain functions $\ty\To\ty[2]$ and copowers $\copower{\ty}{\cty[2]}$ can be represented as lists or arrays of pairs of type $\ty\t*\ty[2]$.

To substantiate that claim, we provide a reference implementation of CHAD operating on strongly typed, deeply embedded DSLs in Haskell at \href{https://github.com/VMatthijs/CHAD}{https://github.com/VMatthijs/CHAD}.
This section explains how that implementation relates to the theoretical development in the rest of this paper.
This section is rather short because our implementation almost exactly follows the 
theoretical development in \S\ref{sec:language}, \ref{sec:minimal-linear-language}, \ref{sec:semantics} and \ref{sec:ad-transformation}.

\subsection{Implementing Linear Functions and Copowers as Abstract Types in Functional Languages}\label{ssec:linear-types-as-abstract-types}
Based on the denotational semantics,
$\cty\multimap\cty[2]$-types should hold (representations of) functions $f$
from $\cty$ to $\cty[2]$
that are homomorphisms of the monoid structures on $\cty$ and $\cty[2]$.
We will see that these types can be implemented using an abstract data type 
that holds certain basic linear functions (extensible as the library evolves)
and is closed under the identity, composition, argument swapping, and currying.
Again, based on the semantics, $\copower{\ty}{\cty[2]}$ should contain (representations of)
finite maps (associative arrays) $\sum_{i=1}^n\copower{\trm_i}{\trm[2]_i}$ of pairs $(\trm_i,\trm[2]_i)$, where $\trm_i$ is of type $\ty$, and 
$\trm[2]_i$ is of type $\cty[2]$, and where we identify $xs+\copower{\trm}{\trm[2]}+\copower{\trm}{\trm[2]'}$
and $xs+\copower{\trm}{(\trm[2]+\trm[2]')}$.

To implement this idea, we consider abstract types $\LinFun{\ty}{\ty[2]}$ of linear functions and $\Map{\ty}{\ty[2]}$ of copowers.
Their programs are generated by the following grammar\\
\begin{syntax}
    \ty, \ty[2], \ty[3] & \gdefinedby & &  \syncat{types}\\     
    &\gor& \ldots                      & \synname{as in \S\ref{sec:language}}\\
   &&&\\
\trm, \trm[2], \trm[3] & \gdefinedby  & & \syncat{terms}             \\
&\gor& \ldots                      & \synname{as in \S\ref{sec:language}}\\
&\gor & \lop(\trm_1,\ldots,\trm_n) & \synname{linear operations}\\
& \gor & \zero_{\ty} & \synname{zero}\\
&\gor & \trm + \trm[2] & \synname{plus}\\
&\gor & \linearid & \synname{linear identity}\\
&\gor & \trm\lcomp \trm[2] & \synname{linear composition}\\
\end{syntax}~\qquad\qquad
  \begin{syntax}
    &\gor \quad\,& \Map{\ty}{\ty[2]} & \synname{copower types}\\ 
  & \gor & \LinFun{\ty}{\ty[2]} & \synname{linear function}\\
  &  & & \\ 
  &\gor & \lswap\,\trm & \synname{swapping args}\\
  &\gor & \leval{\trm} & \synname{linear evaluation}\\
  & \gor & \lsing{\trm} & \synname{singletons}\\
  &\gor & \lcopowfold\trm & \synname{$\MapSym$-elimination}\\
  &\gor & \lFst           & \synname{linear projection}\\
  &\gor & \lSnd           &\synname{linear projection}\\
  &\gor & \lPair{\trm}{\trm[2]} & \synname{linear pairing,}
  \end{syntax}
 \\[2pt]
and their API can be typed according to the rules of Fig. \ref{fig:types-maps}.
\begin{figure}[!t]
    \noindent
\framebox{\begin{minipage}{0.98\linewidth}\noindent\hspace{-24pt}\[
  \begin{array}{c}
    \inferrule{\set{\Ginf {\trm_i} {\reals^{n_i}}\mid i=1,\ldots,k}
    \quad 
    (\lop(\trm_1,\ldots,\trm_k)\in\LOp^{m_1,\ldots,m_r}_{n_1,\ldots,n_k;n'_1,\ldots,n'_l})}{\Ginf {\lop(\trm_1,\ldots,\trm_k)}{\LinFun{\reals^{n'_1}\t*\cdots\t*\reals^{n'_l}}{\reals^{m_1}\t*\cdots\t*\reals^{m_r}}}}
\qquad 
    \inferrule{~}
    {\Ginf {\zero_{\ty}} {\ty}}
    \qquad 
    \inferrule{\Ginf \trm {\ty}
    \quad \Ginf {\trm[2]}{\ty}}
     {\Ginf {\trm+_{\ty}\trm[2]}{\ty}}
\\\\
\inferrule{~}{\Ginf \linearid \LinFun{\ty}{\ty}}
    \qquad
    \inferrule{\Ginf \trm {\LinFun{\ty}{\ty[2]}}
    \quad \Ginf {\trm[2]} {\LinFun{\ty[2]}{\ty[3]}}}
    {\Ginf {\trm\lcomp \trm[2]} {\LinFun{\ty}{\ty[3]}}}
\\\\
    \inferrule{\Ginf \trm {\ty\To \LinFun{\ty[2]}{\ty[3]}}}
    {\Ginf {\lswap\, \trm} {\LinFun{\ty[2]}{\ty\To\ty[3]}}}
    \qquad
    \inferrule{\Ginf \trm {\ty}}
    {\Ginf {\leval{\trm}} {\LinFun{\ty\To\ty[2]}{\ty[2]}}}
    \\\\
    \inferrule{\Ginf \trm \ty}
    {\Ginf {\lsing{\trm}} {\LinFun{\ty[2]}{\Map{\ty}{\ty[2]}}}}
    \qquad 
    \inferrule{\Ginf \trm{\ty\To \LinFun{\ty[2]}{\ty[3]}}}
    {\Ginf {\lcopowfold\trm} {\LinFun {\Map{\ty}{\ty[2]}}{\ty[3]}}}
    \\\\ 
    \inferrule{~}{\Ginf {\lFst}
    {\LinFun{\ty\t*\ty[2]}{\ty}}}\qquad
    \inferrule{~}{\Ginf {\lSnd}
    {\LinFun{\ty\t*\ty[2]}{\ty[2]}}}\qquad 
    \inferrule{\Ginf {\trm}{\LinFun{\ty}{\ty[2]}}\qquad 
    \Ginf{\trm[2]}{\LinFun{\ty}{\ty[3]}}}{\Ginf {\lPair{\trm}{\trm[2]}}
    {\LinFun{\ty}{\ty[2]\t*\ty[3]}}}
\end{array}
\]
 \end{minipage}}
    \caption{Typing rules for the applied target language, to extend the source language.\label{fig:types-maps}}
    \end{figure}

We note that these abstract types give us precisely the functionality
and type safety of the linear function and copower types of our target language 
of \S\ref{sec:minimal-linear-language}.
Indeed, we can define a semantics and type-preserving translation $(-)^T$ from that target language to our source language extended with these $\LinFun{\ty}{\ty[2]}$ and $\Map{\ty}{\ty[2]}$ types, for which $(\copower{\ty}{\cty[2]})^T\defeq \Map{\ty^T}{\cty[2]^T}$,
$(\cty\multimap \cty[2])^T\defeq \LinFun{\cty^T}{\cty[2]^T}$,
$(\creals^n)^T\defeq \reals^n$
and we extend $(-)^T$ structurally recursively, letting it preserve all other type formers.
We then translate
$(\var_1:\ty_1,\ldots,\var_n:\ty_n;\var[2]:\cty[2]\vdash \trm:{\cty[3]})^T\defeq
\var_1:\ty_1^T,\ldots,\var_n:\ty_n^T\vdash \trm^T:{(\cty[2]\multimap\cty[3])^T}$
and $(\var_1:\ty_1,\ldots,\var_n:\ty_n\vdash \trm:{\ty[2]})^T\defeq 
\var_1:\ty_1^T,\ldots,\var_n:\ty_n^T\vdash \trm^T : \ty[2]^T$.
We believe an interested reader can fill in the details.

\subsection{Implementing the API of $\LinFun{\ty}{\ty[2]}$ and $\Map{\ty}{\ty[2]}$ Types}\label{ssec:api}
We can implement the API of $\LinFun{\ty}{\ty[2]}$ and $\Map{\ty}{\ty[2]}$ types
in a language that extends the source language with
types $\List{\ty}$ of lists (or arrays) of elements of type $\ty$.
Indeed, under the hood, we implement $\LinFun{\ty}{\ty[2]}$ as 
$\ty\To\ty[2]$ and $\Map{\ty}{\ty[2]}$ as $\List{\ty\t*\ty[2]}$.
The idea is that $\LinFun{\ty}{\ty[2]}$, which arose as a right adjoint in our linear language, 
is essentially a \emph{subtype} of $\ty\To\ty[2]$. On the other hand, $\Map{\ty}{\ty[2]}$, which arose as a left adjoint, 
is a \emph{quotient type} of $\List{\ty\t*\ty[2]}$.
We achieve the desired subtyping and quotient typing by exposing only the API of Fig. \ref{fig:types-maps} and 
hiding the implementation.
We can then implement this interface as follows.\footnote{Note that the implementation of $\trm +_{\Map{\ty}{\ty[2]}} \trm[2]$ is merely list concatenation written using a fold.}
    \begin{align*}
    &\lop \defeq \overline{\lop}
    \qquad\zero_{\Unit} \defeq \tUnit
    \qquad\trm +_{\Unit} \trm[2] \defeq \tUnit
    \qquad\zero_{\ty\t*\ty[2]} \defeq \tPair{\zero_{\ty}}{\zero_{\ty[2]}}
    \qquad \trm +_{\ty\t*\ty[2]}\trm[2] \defeq \tPair{\tFst\trm+_{\ty}\tFst\trm[2]}{\tSnd\trm+_{\ty[2]}\tSnd\trm[2]}
\\
    &\zero_{\ty\To\ty[2]} \defeq \fun{\_}\zero_{\ty[2]}\qquad \trm +_{\ty\To\ty[2]} \trm[2] \defeq \fun{\var}\trm\,\var +_{\ty[2]} \trm[2]\,\var
    \qquad 
    \zero_{\LinFun{\ty}{\ty[2]}} \defeq \fun{\_}\zero_{\ty[2]}
    \qquad
    \trm +_{\LinFun{\ty}{\ty[2]}} \trm[2] \defeq \fun{\var}\trm\,\var +_{\ty[2]} \trm[2]\,\var
    \\ 
    &\zero_{\Map{\ty}{\ty[2]}}\defeq \EmptyList\qquad \trm +_{\Map{\ty}{\ty[2]}} \trm[2] \defeq \ListFold{\ListCons{\var}{acc}}{\var}{\trm}{acc}{\trm[2]}
    \\ &\linearid \defeq \fun{\var}\var  \qquad \trm\lcomp\trm[2] \defeq \fun{\var}\trm[2]\,(\trm\,\var)\qquad \lswap\,\trm\defeq \fun{\var}\fun{\var[2]}\trm\,\var[2]\,\var
    \qquad \leval{\trm}\defeq \fun{\var}\var\,\trm 
    \\
  & \lsing{\trm}\defeq \fun{\var}\ListCons{\tPair{\trm}{\var}}{\EmptyList}\qquad  \lcopowfold\trm\defeq \fun{\var[3]}\ListFold{\trm\, (\tFst\var)\,(\tSnd\var)+acc}{\var}{\var[3]}{acc}{\zero}
    \\ &\lFst\defeq \fun\var {\tFst\var}\qquad \lSnd\defeq \fun\var{\tSnd\var}
    \qquad \lPair{\trm}{\trm[2]}\defeq \fun\var{\tPair{\trm\,\var}{\trm[2]\,\var}}
    \end{align*}
Here, we write $\overline{\lop}$ for the function that $\lop$ is intended to implement,
$\EmptyList$ for the empty list, $\ListCons{\trm}{\trm[2]}$ for the list consisting 
of $\trm[2]$ with $\trm$ prepended, and $\ListFold{\trm}{\var}{\trm[2]}{acc}{init}$
for (right) folding an operation $\trm$ over a list $\trm[2]$, starting from $init$.
The monoid structure $0_{\ty}$, $+_{\ty}$ can be defined by induction on the structure of types, 
using, for example, type classes in Haskell, OCaml's module system, or reified types in any functional language with algebraic data types.
Furthermore, the implementer of the AD library can determine which linear operations $\lop$ 
to include within the implementation of $\LinFunSym$.
We expect these linear operations to include various forms of dense
and sparse matrix-vector multiplication as well as code that computes 
Jacobian-vector and Jacobian-adjoint products for the operations $\op$ while avoiding the need to 
compute the full Jacobian.
Another option is to simply include two linear operations $D\op$ and $\transpose{D\op}$ for computing the 
derivative and transposed derivative of each operation $\op$.

\subsection{Maintaining Type Safety throughout the Compilation Pipeline in our Reference Implementation}
In a principled approach to building a define-then-run AD library,
one would shield this implementation using the
abstract data types $\Map{\ty}{\ty[2]}$ and $\LinFun{\ty}{\ty[2]}$ as 
we describe, both
for reasons of type safety
and because it conveys the intuition
behind the algorithm and its correctness.
By combining such abstract data types in our Haskell implementation 
with GADTs and type families, we achieve a fully type-safe (well-scoped, well-typed De Bruijn) 
implementation of 
the source and target languages of \S\ref{sec:language} and \S\ref{sec:minimal-linear-language} with their semantics of \S\ref{sec:semantics} 
and statically type-checked code transformations of \S\ref{sec:ad-transformation}.

However, nothing prevents library implementers from exposing the full implementation rather than working with abstract types.
In fact, this seems to be the approach taken in \cite{vytiniotis2019differentiable}.
A downside of that ``exposed'' approach is that the transformations then no longer 
respect equational reasoning principles.
In our reference implementation, we include a compiler from the (linearly typed) target language to a less type-safe ``concrete'' target language (implementing \S\S\ref{ssec:api} as a compilation step): essentially the source language extended with list (or array) types.\footnote{More precisely, the to-concrete compilation step in the implementation does convert copowers and linear functions to lists and regular functions, but retains $\zero$ and $+$ primitives. This brings a major increase in readability of the output. Implementing the inductive definitions for $\zero$ and $+$ is easy when reified types (singletons) are added to the implementation, which itself is an easy change if one modifies the \texttt{LT} type class.}
This demonstrates that CHAD implements a compile-time AD code transformation that takes standard functional code as input and produces standard functional code without any custom semantics.

\subsection{Compiling Away Copowers}
As a final observation on this implementation, we note that while the proposed implementation 
of copowers as lists is generally applicable, more efficient implementation strategies can often be achieved in practice.
In fact, in unpublished follow-up work to this paper led by Tom Smeding, we show that when we implement CHAD for Accelerate \cite{McDonell:2013:OPF:2500365.2500595}, we can optimize away uses of copower types.

\section{Adding Higher-Order Array Primitives}
\label{sec:higher-order-primitives}
The aim of this paper is to answer the foundational question 
of how to perform (reverse) AD at higher types.
The problem of how to perform AD for evaluation and currying is highly 
challenging.
For this reason, we have devoted this paper to explaining a solution to that problem in detail, working with a toy language whose ground types 
are black-box, sized arrays $\reals^n$ with some first-order
operations $\op$.
However, many of the interesting applications only arise 
once we can use higher-order array primitives such as $\tMap$ 
and $\mathbf{fold}$ on $\reals^n$.

Our definitions and correctness proofs extend to this
 setting with standard array processing primitives including $\tMap$, $\mathbf{fold}$, $\mathbf{filter}$, $\mathbf{zipWith}$, $\mathbf{permute}$ (also known as $\mathbf{scatter}$), $\mathbf{backpermute}$ (also known as $\mathbf{gather}$), $\mathbf{generate}$ (also known as $\mathbf{build}$), and array indexing. We
 plan to discuss these primitives as well as CHAD applied to dynamically sized arrays in detail 
in an applied follow-up paper, which will focus on 
an implementation of CHAD for Accelerate \cite{McDonell:2013:OPF:2500365.2500595}.

To illustrate the idea behind such an extension, we briefly discuss the case of $\tMap$ here and leave the rest to future work.
Suppose that we add operations\footnote{We prefer to work with this elementary formulation of maps rather than the usual higher-order formulation of 
$$
\inferrule{\Ginf\trm{\reals\To\reals}\quad 
\Ginf{\trm[2]}{\reals^n}}{\Ginf{\tMap(\trm,\trm[2])}{\reals^n}}
$$
because it makes sense in the wider context of languages 
without function types as well and because it simplifies the CHAD correctness proof. 
Note that both are equivalent in the presence of function types: $\tMap(\trm,\trm[2])=\tMap(\var.\trm\,\var,\trm[2])$ and $\tMap(\var.\trm,\trm[2])=\tMap(\fun{\var}\trm,\trm[2])$.
}
\[
  \inferrule{\Ginf[,\var:\reals]\trm\reals\quad 
  \Ginf{\trm[2]}{\reals^n}}{\Ginf{\tMap(\var.\trm,\trm[2])}{\reals^n}}  
\]
 to the 
source language to ``map'' functions over the black-box arrays.
Then, supposing that we add the following primitives to the target language 
\begin{align*}
\inferrule{\Gamma,\var:\reals;\lvar:\cty\t*\creals\vdash \trm: \creals\quad 
\Gamma\vdash \trm[2]: \reals^n\quad 
\Gamma;\lvar:\cty\vdash \trm[3]:\creals^n
}
{\Gamma;\lvar:\cty\vdash D\tMap(\var.\trm,\trm[2],\trm[3]) :\creals^n}\\
\\
\inferrule{
\Gamma,\var:\reals; \lvar:\creals\vdash \trm:\cty\t* \creals \quad 
\Gamma\vdash \trm[2]: \reals^n
\quad\Gamma;\lvar:\creals^n\vdash \trm[3] : \cty}
{\Gamma;\lvar:\creals^n\vdash \transpose{D\tMap}(\var.\trm,\trm[2],\trm[3]) : \cty}
\end{align*}
we can define 
\begin{flalign*}
&\Dsyn[\vGamma]{\tMap(\var.\trm,\trm[2])} &&\defeq &&
\letin{\var[2]}{\fun{\var}\Dsyn[\vGamma,\var]{\trm}}{
\pletin{\var[3]}{\var[3]'}{\Dsyn[\vGamma]{\trm[2]}}{
\tPair{\tMap(\var.\tFst(\var[2]\,\var),\var[3])}{\lfun{\lvar}D\tMap(\var.\lapp{(\tSnd(\var[2]\,\var))}{\lvar},\var[3],\lapp{\var[3]'}{\lvar})}}}\\
&\Dsynrev[\vGamma]{\tMap(\var.\trm,\trm[2])} &&\defeq &&
\letin{\var[2]}{\fun{\var}\Dsynrev[\vGamma,\var]{\trm}}{
\pletin{\var[3]}{\var[3]'}{\Dsynrev[\vGamma]{\trm[2]}}{
\tPair{\tMap(\var.\tFst(\var[2]\,\var),\var[3])}{\lfun\lvar \transpose{D\tMap}(\var.\lapp{(\tSnd (\var[2]\,\var))}{\lvar},\var[3],\lapp{\var[3]'}{\lvar})}}}.
\end{flalign*}

In our practical API of \S\S\ref{ssec:linear-types-as-abstract-types},
the required target language primitives correspond to 
\begin{align*}
  \inferrule{\Gamma,\var:\reals\vdash \trm: \LinFun{\cty\t*\creals}{\creals}\quad 
  \Gamma\vdash \trm[2]: \reals^n\quad 
  \Gamma\vdash \trm[3]:\LinFun{\cty}{\creals^n}
  }
  {\Gamma\vdash D\tMap(\var.\trm,\trm[2],\trm[3]) :\LinFun{\cty}{\creals^n}}\\
  \\
  \inferrule{
  \Gamma,\var:\reals\vdash \trm:\LinFun{\creals}{\cty\t* \creals} \quad 
  \Gamma\vdash \trm[2]: \reals^n
  \quad\Gamma\vdash \trm[3] : \LinFun{\creals^n}{\cty}}
  {\Gamma\vdash \transpose{D\tMap}(\var.\trm,\trm[2],\trm[3]) : \LinFun{\creals^n}{\cty}.}
  \end{align*}
Extending \S\S\ref{ssec:api}, we can implement the API as 
\begin{align*}
&D\tMap(\var.\trm,\trm[2],\trm[3])\defeq \fun{\var[2]} \tZipWith((\var,\var').\trm\,\tPair{\var[2]}{\var'},\trm[2],\trm[3]\,\var[2])\\
&\transpose{D\tMap}(\var.\trm,\trm[2],\trm[3])\defeq 
\fun{\var[2]}\letin{zs}{\tZipWith((\var,\var').\trm\,\var',\trm[2],\var[2])}{
 \mathbf{sum}(\tMap(w.\tFst\,w,zs)) + \trm[3] \,\tMap(w.\tSnd\,w,zs)
},
\end{align*}
where 
\begin{align*}
  \inferrule{\Ginf[,\var:\ty]\trm\ty[2]\quad 
\Ginf{\trm[2]}{\ty^n}}{\Ginf{\tMap(\var.\trm,\trm[2])}{\ty[2]^n}}  
\qquad 
\inferrule{\Gamma,\var:\ty,\var':\ty[2]\vdash\trm:\ty[3]\quad
\Gamma\vdash \trm[2]:\ty^n\quad 
\Gamma\vdash \trm[3]:\ty[2]^n}
{\Gamma\vdash \tZipWith((\var,\var').\trm,\trm[2],\trm[3]):{\ty[3]}^n}\qquad
\inferrule{\Ginf\trm{{\ty}^n}}{\Ginf{\mathbf{sum}\,\trm}{\ty}}
\end{align*}
are the usual functional programming idioms for mapping a unary function over an array, zipping two arrays with a binary operation, and taking the sum of the elements in an array.
Note that we assume that we have types $\ty^n$ for 
length-$n$ arrays of elements of type $\ty$ here, 
generalizing the arrays $\reals^n$ of elements of type $\reals$.
We present a correctness proof for this implementation of the derivatives in \citeappx{A}.

Applications frequently require AD of
higher-order primitives such as differential and algebraic equation 
solvers,
e.g. for use in pharmacological modelling in 
Stan \cite{tsiros2019population}.
Currently, derivatives of such primitives are derived using
the calculus of variations (and implemented with define-by-run AD) \cite{betancourt2020discrete, hannemann2015adjoint}.
Our proof method provides a more lightweight and formal method for 
calculating derivatives for such higher-order primitives and establishing their correctness.
Indeed, most formalizations of the calculus of variations
use infinite-dimensional vector spaces and are technically involved
\cite{kriegl1997convenient}.
 
\section{Scope of CHAD and Future Work}\label{sec:scope}

\subsection{Memory Use of CHAD's Forward AD}\label{ssec:dual-numbers-forward-ad}
Our formulation makes reverse and forward AD precisely each other's 
categorical dual.
The former first computes the primals in a forward pass and then the cotangents in a reverse pass. 
Dually, the latter first computes the primals in a forward pass and then the tangents in another forward pass.
Since the two forward passes in forward AD have identical control flow, 
it can be advantageous to interleave them and simultaneously compute the primals and tangents.
Such interleaving greatly reduces the memory consumption of the algorithm, because not all primals need to be stored for most of the algorithm.
We present such an interleaved formulation of forward AD in \cite{hsv-fossacs2020}.

Although this formulation is much more memory efficient, it has the conceptual downside of no longer being the mirror image of reverse AD.
Furthermore, these interleaved formulations of forward AD work by operating on dual numbers.
That is, they use an array-of-structs representation, in contrast to the struct-of-arrays representation used to pair primals with tangents in CHAD.
Therefore, an SoA-to-AoS optimization is typically needed to make interleaved implementations of forward AD efficient \cite{shaikhha2019efficient}.

Finally, we note that such interleaving techniques do not apply to reverse AD, 
because the dependency structure of the algorithm requires us to complete the forward primal pass before starting the reverse cotangent pass.

\subsection{Applying CHAD to Richer Source Languages}\label{ssec:chad-rich-source}
The core observations that let us use CHAD for AD on a higher-order language were the following:
\begin{enumerate}
\item there is a class of categories with structure $\mathcal{S}$ (in this case, Cartesian closure)
such that the source language $\Syn$ on which we want to perform AD can be seen as the freely generated 
$\mathcal{S}$-category on the operations $\op$;
\item we identified structure $\mathcal{T}$ that suffices for a $\CMon$-enriched strictly indexed category $\catL:\catC^{op}\to \Cat$ to ensure that $\Sigma_\catC\catL$ and $\Sigma_\catC\catL^{op}$ are $\mathcal{S}$-categories;
\item we gave a description $\LSyn:\CSyn^{op}\to\Cat$ of the freely generated $\CMon$-enriched strictly indexed category with structure $\mathcal{T}$, on the Cartesian operations $\op$ in $\CSyn$ and linear operations $D\op$ and $\transpose{D\op}$ in $\LSyn$; we interpret this linear/non-linear language as the target language of our AD translations;
\item by the universal property of $\Syn$, we obtain unique $\mathcal{S}$-homomorphic AD functors $\Dsynsymbol:\Syn\to\Sigma_\CSyn\LSyn$ and $\Dsynrevsymbol:\Syn\to\Sigma_\CSyn\LSyn^{op}$ such that $\Dsyn{\op}=(\op,D\op)$ and 
$\Dsynrev{\op}=(\op,\transpose{D\op})$, whose correctness proof follows immediately because of the well-known theory of 
subsconing for $\mathcal{S}$-categories.
\end{enumerate}
CHAD applies equally to source languages with other choices of $\mathcal{S}$, provided that we can follow steps 1-4.

In particular, \cite{lucatellivakar2021chad} shows how CHAD applies equally to languages with sum types and (co)inductive types (and tuple and function types). 
In that setting, the category $\LSyn$ is a genuine strictly indexed category over $\CSyn$, to account 
for the fact that the (co)tangent space to a space of varying dimension depends on the chosen base point.
That is, in its most principled formulation, the target language has (linear) dependent types.
However, we can also work with a simply typed target language in this setting, at the cost of some extra type safety.
In fact, our Haskell implementation already supports such a treatment of coproducts.

As discussed in \S\ref{sec:higher-order-primitives}, $\mathcal{S}$ can also be chosen to include various operations for manipulating arrays, such as $\tMap$, $\mathbf{fold}$, $\mathbf{filter}$, $\mathbf{zipWith}$, $\mathbf{permute}$ (also known as $\mathbf{scatter}$), $\mathbf{backpermute}$ (also known as $\mathbf{gather}$), $\mathbf{generate}$ (also known as $\mathbf{build}$), and array indexing.
We plan to describe this application of CHAD to array processing languages
and we are implementing CHAD to operate on the Accelerate parallel array processing language.

In work in progress, we are applying CHAD to partial features such as real conditionals, iteration, recursion and recursive types.
Our Haskell implementation of CHAD already supports real conditionals, iteration and recursion.
The challenge in this setting is to understand the subtle interactions between the $\wCpo$-structure 
needed to model recursion and the commutative monoid structure that CHAD uses to accumulate (co)tangents.

\subsection{CHAD for Other Dynamic Program Analyses}\label{ssec:other-analyses}
As noted by \cite{vytiniotis2019differentiable}, source-code transformation AD 
has many similarities to other dynamic program analyses such as 
dynamic symbolic analysis and provenance analysis.

In fact, as the abstract perspective on CHAD given in \S\S\ref{ssec:chad-rich-source} makes clear, CHAD is in no way tied to automatic differentiation.
In many ways, it is much more general, and can best be seen as a framework for 
applying dynamic program analyses that accumulate data (either by going through the program forward or backward) in a commutative monoid 
to functional languages with expressive features.
In fact, by varying the definitions of $\Dsyn{\mathcal{R}}$ and $\Dsynrev{\mathcal{R}}$ for the ground types $\mathcal{R}$ and the definitions of $\Dsyn{\op}$ and $\Dsynrev{\op}$ for the primitive operations $\op$
(we do not even need to use $\Dsyn{\mathcal{R}}_1=\mathcal{R}$, $\Dsynrev{\mathcal{R}}_1=\mathcal{R}$, $\Dsyn{\op}_1=\op$ or $\Dsynrev{\op}_1=\op$!),
we can completely change the nature of the analysis.
In most cases, as long as a notion of correctness of the analysis can be phrased at the level of a denotational semantics,
we conjecture that our subsconing techniques lead to straightforward correctness proofs of the analysis.

To give one more example application of such an analysis, beyond AD, dynamic symbolic analysis and provenance analysis, note that for a source language $\Syn$, generated from a base type $\mathcal{R}$ that is a commutative semiring, we have a notion of algebraic (or formal) derivative of any polynomial $x_1:\mathcal{R},\ldots,x_n:\mathcal{R}\vdash \op(x_1,\ldots,x_n):\mathcal{R}$ \cite{lang02}.
CHAD can be used to extend and compute this notion of derivative for arbitrary functional programs generated from the polynomials $\op$ as basic operations.
The particular case of such formal derivatives for the Boolean semiring $\mathcal{R}=\mathbb{B}$ is used in
\cite{wilson2021reverse} to feed into a gradient descent algorithm to learn Boolean circuits. 
CHAD makes this method applicable to more general (higher-order) programs over (arrays of) Booleans.

\section{Related Work}\label{sec:related-work}
This work is closely related to \cite{hsv-fossacs2020} and \cite{huot2021higher},
which introduced a similar semantic correctness proof for a dual-numbers version 
of forward mode AD and higher-order forward AD, using a subsconing construction.
A major difference is that this paper also phrases and proves 
correctness of reverse mode AD on a $\lambda$-calculus and relates reverse mode 
to forward mode AD.
Using a syntactic logical relations proof instead, \cite{bcdg-open-logical-relations}
also proves correctness of forward mode AD.
Again, it does not address reverse AD.

\cite{rev-deriv-cat2020} proposes a construction similar to that of
\S\ref{sec:self-dualization}, and it relates it to the
differential $\lambda$-calculus.
This paper develops sophisticated axiomatics for
semantic reverse differentiation. 
However, it neither relates the semantics to a source-code transformation,
nor discusses differentiation of higher-order functions.
Our construction of differentiation with a (biadditive) linear target language 
might remind the reader of differential linear logic \cite{ehrhard2018introduction}.
In differential linear logic, (forward) differentiation is a first-class operation 
in a (biadditive) linear language. 
By contrast, in our treatment, differentiation is a meta-operation.

Importantly, \cite{elliott2018simple} describes and implements
 what are essentially our source-code transformations, though they were restricted to 
first-order functions and scalars.
After completing this work, we realized that 
\cite{vytiniotis2019differentiable} describes an extension of the reverse mode transformation
to higher-order functions in a manner similar to what we propose in this paper, but without the linear or abstract types. Though that paper did not derive the algorithm or show its correctness, it does discuss important practical considerations for its implementation and offers a dependently typed 
variant of the algorithm based on typed closure conversion, inspired by \cite{pearlmutter2008reverse}.

Next, there are various lines of work related to the correctness of 
reverse mode AD that we consider less similar to our work.
For example, \cite{mak-ong2020}
define and prove correct a formulation of reverse mode AD on a higher-order 
language that depends on a non-standard operational semantics, essentially 
a form of symbolic execution. \cite{abadi-plotkin2020} does something similar 
for reverse mode AD on a first-order language extended with conditionals and iteration.
\cite{brunel2019backpropagation} defines a beautifully simple AD algorithm 
on a simply typed $\lambda$-calculus with linear negation
(essentially, a more finely typed version of the continuation-based AD of \cite{hsv-fossacs2020})
and proves it correct using operational techniques.
\cite{mazza2021automatic} extends this work to apply to recursion.
Furthermore, they show with an impressive operational argument that this simple algorithm, surprisingly, corresponds to true reverse mode AD with the correct complexity
under an operational semantics with a ``linear factoring rule''.
While this is a natural operational semantics for a linear $\lambda$-calculus,
it is fundamentally different from normal call-by-value or call-by-name evaluation (under which the generated code has the wrong computational complexity).
For this reason, this reverse AD method requires a custom interpreter or compiler in practice.
Very recently, \cite{krawiec2022provably} specified another purely functional reverse AD algorithm, which appears similar to,
and which we conjecture to be equivalent to, an implementation of the techniques of \cite{mazza2021automatic}.
These formulations of reverse mode AD all depend on non-standard runtimes 
and hence fall into the category of ``define-by-run'' 
formulations of reverse mode AD, for our purposes. 
Meanwhile, we are concerned with ``define-then-run'' formulations:
source-code transformations that produce differentiated code at compile time that can then be optimized during compilation with existing compiler 
toolchains (such as the Accelerate \cite{Chakravarty:2011:AHA:1926354.1926358}, Futhark \cite{henriksen2017futhark} and TensorFlow \cite{abadi2016tensorflow} frameworks for generating high-performance GPU code).
While we can compile such define-by-run transformations together with their interpreter to achieve a source-code transformation (hence a sort of define-then-run transformation), the resulting code recursively 
traverses an AST, so it does not seem obviously suitable for generating, via existing toolchains, optimized machine code 
for the usual parallel hardware that we use as targets for AD, such as GPUs and TPUs.

Finally, there is a long history of work on reverse mode AD,
though almost none of it applies the technique to higher-order 
functions.
A notable exception is \cite{pearlmutter2008reverse}, which 
gives an impressive source-code 
transformation implementation of reverse AD in Scheme.
While very efficient, this implementation crucially uses mutation.
Moreover, the transformation is complex and correctness is not considered.
More recently, \cite{wang2018demystifying} describes a much simpler 
implementation of a reverse AD code transformation, also very efficient.
However, the transformation is quite different from the one 
considered in this paper as it relies on a combination of delimited continuations 
and mutable state.
Correctness is not considered, perhaps because of the semantic complexities
introduced by impurity.

Our work adds to the existing literature by presenting a novel, generally applicable 
method for compositional source-code transformation (forward and) reverse AD on expressive functional 
languages without requiring a non-standard runtime, by giving a method for compositional correctness proofs of such AD algorithms, 
and by observing that the CHAD method and its correctness proof are not limited to AD but 
apply generally to dynamic program analyses that accumulate data in a commutative monoid.
\section*{Acknowledgements}
This project has received funding from the European Union’s Horizon 2020 research and innovation
programme under the Marie Skłodowska-Curie grant agreement No. 895827.
We thank Michael Betancourt, Philip de Bruin, Bob Carpenter, Mathieu Huot, Danny de Jong, Ohad Kammar, Gabriele Keller, 
Pieter Knops, Fernando Lucatelli Nunes,
Curtis Chin Jen Sem, Amir Shaikhha,
and Sam Staton for helpful discussions about automatic differentiation.
We are grateful to the anonymous reviewers who gave excellent comments
on earlier versions of this paper that prompted various much-needed rewrites.

\clearpage
\bibliography{bibliography}

\appendix
\clearpage
\section{CHAD Correctness for Higher-order Operations such as Map}\label{appx:map-fold}
We extend the proofs of Lemmas \ref{lem:fundamental-fwd}
and \ref{lem:fundamental-rev} to apply to the $\tMap$-constructs of \S\ref{sec:higher-order-primitives}.

\subsection{The Semantics of $\tMap$ and its Derivatives}
First, we observe that
\begin{align*}
\sem{\tMap(\var.\trm,\trm[2])} : &\sem{\Gamma}\to \RR^n\\
&\gamma\mapsto \left(\sem{\trm}(\gamma, \pi_1(\sem{\trm[2]}(\gamma))),\ldots, \sem{\trm}(\gamma, \pi_n(\sem{\trm[2]}(\gamma)))\right).
\end{align*}
Similarly, 
\begin{align*}
    \sem{D\tMap(\var.\trm,\trm[2],\trm[3])} : &\sem{\Gamma}\to \sem{\cty}\multimap \cRR^n\\
&\gamma\mapsto v\mapsto \left(\sem{\trm}(\gamma,\pi_1(\sem{\trm[2]}(\gamma)))(v,\pi_1(\sem{\trm[3]}(\gamma)(v))),\ldots, \sem{\trm}(\gamma,\pi_n(\sem{\trm[2]}(\gamma)))(v,\pi_n(\sem{\trm[3]}(\gamma)(v)))\right)\\
\sem{\transpose{D\tMap}(\var.\trm,\trm[2],\trm[3])} : &\sem{\Gamma}\to \cRR^n\multimap \sem{\cty}\\
&\gamma\mapsto v\mapsto \pi_1(\sem{\trm}(\gamma,\pi_1(\sem{\trm[2]}(\gamma)))(\pi_1(v)))+\ldots + \pi_1(\sem{\trm}(\gamma,\pi_n(\sem{\trm[2]}(\gamma)))(\pi_n(v)))+\\
&\hspace{40pt}\sem{\trm[3]}(\gamma)(\pi_2(\sem{\trm}(\gamma,\pi_1(\sem{\trm[2]}(\gamma)))(\pi_1(v))),\ldots, \pi_2(\sem{\trm}(\gamma,\pi_n(\sem{\trm[2]}(\gamma)))(\pi_n(v))))
\end{align*}
This implies that 
\begin{align*}
&\pi_1(\sem{\Dsyn[\vGamma]{\tMap(\var.\trm,\trm[2])}}(\gamma))=\left(\pi_1(\sem{\Dsyn[\vGamma,\var]{\trm}}(\gamma,\pi_1(\pi_1(\sem{\Dsyn[\vGamma]{\trm[2]}}(\gamma))))),\ldots, \pi_1(\sem{\Dsyn[\vGamma,\var]{\trm}}(\gamma,\pi_n(\pi_1(\sem{\Dsyn[\vGamma]{\trm[2]}}(\gamma)))))\right)\\
&\pi_2(\sem{\Dsyn[\vGamma]{\tMap(\var.\trm,\trm[2])}}(\gamma))(v)=\Big(
\pi_2(\sem{\Dsyn[\vGamma,\var]{\trm}}(\gamma,\pi_1(\pi_1(\sem{\Dsyn[\vGamma]{\trm[2]}}(\gamma)))))(v,\pi_1(\pi_2(\sem{\Dsyn[\vGamma]{\trm[2]}}(\gamma))(v)))
,\ldots,\\
&\hspace{150pt}\pi_2(\sem{\Dsyn[\vGamma,\var]{\trm}}(\gamma,\pi_n(\pi_1(\sem{\Dsyn[\vGamma]{\trm[2]}}(\gamma)))))(v,\pi_n(\pi_2(\sem{\Dsyn[\vGamma]{\trm[2]}}(\gamma))(v))) \Big)\\
&\pi_1(\sem{\Dsynrev[\vGamma]{\tMap(\var.\trm,\trm[2])}}(\gamma))=\left(\pi_1(\sem{\Dsynrev[\vGamma,\var]{\trm}}(\gamma,\pi_1(\pi_1(\sem{\Dsynrev[\vGamma]{\trm[2]}}(\gamma))))),\ldots, \pi_1(\sem{\Dsynrev[\vGamma,\var]{\trm}}(\gamma,\pi_n(\pi_1(\sem{\Dsynrev[\vGamma]{\trm[2]}}(\gamma)))))\right)\\
&\pi_2(\sem{\Dsynrev[\vGamma]{\tMap(\var.\trm,\trm[2])}}(\gamma))(v)=
\pi_1(\pi_2(\sem{\Dsynrev[\vGamma,\var]{\trm}}(\gamma,\pi_1(\pi_1(\sem{\Dsynrev[\vGamma]{\trm[2]}}(\gamma)))))(
\pi_1(v)
))+\cdots\\
&\hspace{115pt}+ \pi_1(\pi_2(\sem{\Dsynrev[\vGamma,\var]{\trm}}(\gamma,\pi_n(\pi_1(\sem{\Dsynrev[\vGamma]{\trm[2]}}(\gamma)))))(
    \pi_n(v)
    ))\\
    &\hspace{115pt}+\pi_2(\sem{\Dsynrev[\vGamma]{\trm[2]}}(\gamma))(
        \pi_2(\pi_2(\sem{\Dsynrev[\vGamma,\var]{\trm}}(\gamma,\pi_1(\pi_1(\sem{\Dsynrev[\vGamma]{\trm[2]}}(\gamma)))))(
            \pi_1(v)
            )),\ldots,\\
            &\hspace{185pt}\pi_2(\pi_2(\sem{\Dsynrev[\vGamma,\var]{\trm}}(\gamma,\pi_n(\pi_1(\sem{\Dsynrev[\vGamma]{\trm[2]}}(\gamma)))))(
                \pi_n(v)
                ))
    )\\
&\phantom{\pi_2(\sem{\Dsynrev[\vGamma]{\tMap(\var.\trm,\trm[2])}}(\gamma))(v)}=\sum_{i=1}^n \Big(
        \pi_1(\pi_2(\sem{\Dsynrev[\vGamma,\var]{\trm}}(\gamma,\pi_i(\pi_1(\sem{\Dsynrev[\vGamma]{\trm[2]}}(\gamma)))))(
        \pi_i(v)
        ))\\
        &\hspace{115pt}+\pi_2(\sem{\Dsynrev[\vGamma]{\trm[2]}}(\gamma))(0,\ldots,0,\pi_2(\pi_2(\sem{\Dsynrev[\vGamma,\var]{\trm}}(\gamma,\pi_i(\pi_1(\sem{\Dsynrev[\vGamma]{\trm[2]}}(\gamma)))))(
                    \pi_i(v)
                    )),0,\ldots,0
        )\Big),
\end{align*}
where the last equation holds by linearity of $\pi_2(\sem{\Dsynrev[\vGamma]{\trm[2]}}(\gamma))$.

\subsection{Extending the Induction Proof of the Fundamental Lemma for Forward CHAD}
First, we focus on extending the induction proof of the fundamental lemma for forward CHAD to apply to maps.
Assume the induction hypothesis that 
$\trm$ and $\trm[2]$ respect the logical relation.
We show that $\tMap(\var.\trm,\trm[2])$ does as well.
We assume all terms are well-typed.
Suppose that $(f,(g,h))\in P_{\Gamma}$.
We want to show that 
$$(f',(g',h'))\!\defeq\!(f;\sem{\tMap(\var.\trm,\trm[2])},
    (g; \sem{\Dsyn[\vGamma]{\tMap(\var.\trm,\trm[2])}} ; \pi_1, x\mapsto r\mapsto \pi_2(\sem{\Dsyn[\vGamma]{\tMap(\var.\trm,\trm[2])}}(g(x)))(h(x)(r))
    ))\in P_{\reals^n}.$$
Note that $f'= (f'_1,\ldots, f'_n)$, $g'=(g'_1,\ldots,g'_n)$ and $h'(x)=(h'_1(x),\ldots,h'_n(x))$; as derivatives are computed componentwise, it is equivalent to show that $(f'_i,(g'_i, h'_i))\in P_\reals$ for $i=1,\ldots,n$.
That is, we need to show that 
\begin{align*}
&(x\mapsto \sem{\trm}(f(x), \pi_i(\sem{\trm[2]}(f(x)))),
(x\mapsto \pi_1(\sem{\Dsyn[\vGamma,\var]{\trm}}(g(x), \pi_i(\pi_1(\sem{\Dsyn[\vGamma]{\trm[2]}}(g(x)))))),\\
&\qquad x\mapsto r\mapsto\pi_2(\sem{\Dsyn[\vGamma,\var]{\trm}}(g(x),\pi_i(\pi_1(\sem{\Dsyn[\vGamma]{\trm[2]}}(g(x))))))(h(x)(r),\pi_i(\pi_2(\sem{\Dsyn[\vGamma]{\trm[2]}}(g(x)))(h(x)(r))))
))\in P_{\reals}
\end{align*}
As $\trm$ respects the logical relation by our induction hypothesis, it is enough to show that 
\begin{align*}
&(x\mapsto (f(x), \pi_i(\sem{\trm[2]}(f(x)))),
(x\mapsto (g(x), \pi_i(\pi_1(\sem{\Dsyn[\vGamma]{\trm[2]}}(g(x)))))\\
&\qquad x\mapsto r\mapsto (h(x)(r),\pi_i(\pi_2(\sem{\Dsyn[\vGamma]{\trm[2]}}(g(x)))(h(x)(r))))
))\in P_{\Gamma,\var:\reals}.
\end{align*}
Since $(f,(g,h))\in P_\Gamma$ by assumption, it is enough, by definition of 
$P_{\Gamma,\var: \reals}$, to show that
\begin{align*}
    &(x\mapsto \pi_i(\sem{\trm[2]}(f(x))),
    (x\mapsto \pi_i(\pi_1(\sem{\Dsyn[\vGamma]{\trm[2]}}(g(x)))),\\
    &\qquad x\mapsto r\mapsto \pi_i(\pi_2(\sem{\Dsyn[\vGamma]{\trm[2]}}(g(x)))(h(x)(r)))
    ))\in P_{\reals}.
    \end{align*}
By definition of $P_{\reals^n}$, it is enough to show that 
\begin{align*}
    &(x\mapsto \sem{\trm[2]}(f(x)),
    (x\mapsto \pi_1(\sem{\Dsyn[\vGamma]{\trm[2]}}(g(x))),\\
    &\qquad x\mapsto r\mapsto \pi_2(\sem{\Dsyn[\vGamma]{\trm[2]}}(g(x)))(h(x)(r)))
    )\in P_{\reals^n}.
    \end{align*}
Since $\trm[2]$ respects the logical relation by our induction hypothesis, it is enough to show that 
\begin{align*}
    &(x\mapsto f(x),
    (x\mapsto g(x),\\
    &\qquad x\mapsto r\mapsto h(x)(r))
    )\in P_{\Gamma},
    \end{align*}
which is true by assumption.

\subsection{Extending the Induction Proof of the Fundamental Lemma for Reverse CHAD}
Next, we extend the fundamental lemma for reverse CHAD to apply to maps.
Assume the induction hypothesis that 
$\trm$ and $\trm[2]$ respect the logical relation.
We show that $\tMap(\var.\trm,\trm[2])$ does as well.
We assume all terms are well-typed.
Suppose that $(f,(g,h))\in P_{\Gamma}$.
We want to show that 
$$(f',(g',h'))\!\defeq\!
(f;\sem{\tMap(\var.\trm,\trm[2])},
(g; \sem{\Dsynrev[\vGamma]{\tMap(\var.\trm,\trm[2])}} ; \pi_1,{x}\mapsto {v}\mapsto
h(x)( \pi_2(\sem{\Dsynrev[\vGamma]{\tMap(\var.\trm,\trm[2])}}(g(x)))( v))))\in P_{\reals^n}.$$
By basic multivariate calculus, elements $(f',(g',h'))\in P_{\reals^n}$ are all of the form 
$f'=(f'_1,\ldots,f'_n)$, $g'=(g'_1,\ldots,g'_n)$ and $h'(x)(v)=h'_1(x)(\pi_1(v))+\cdots+h'_n(x)(\pi_n(v))$ where $(f'_i,(g'_i,h'_i))\in P_{\reals}$ for $i=1,\ldots,n$.
That is, we need to show (by linearity of $h(x)$) that 
\begin{align*}
  & (x\mapsto \sem{\trm}(f(x),\pi_i(\sem{\trm[2]}(f(x)))),
   (x\mapsto \pi_1(\sem{\Dsynrev[\vGamma,\var]{\trm}}(g(x),\pi_i(\pi_1(\sem{\Dsynrev[\vGamma]{\trm[2]}}(g(x)))))),\\
  &\qquad x\mapsto v\mapsto 
  h(x)( 
    \pi_1(\pi_2(\sem{\Dsynrev[\vGamma,\var]{\trm}}(g(x),\pi_i(\pi_1(\sem{\Dsynrev[\vGamma]{\trm[2]}}(g(x))))))(
        v
        ))\\
        &\qquad\;\;\;+\pi_2(\sem{\Dsynrev[\vGamma]{\trm[2]}}(g(x)))(0,\ldots,0,\pi_2(\pi_2(\sem{\Dsynrev[\vGamma,\var]{\trm}}(g(x),\pi_i(\pi_1(\sem{\Dsynrev[\vGamma]{\trm[2]}}(g(x))))))(
                    v
                    )),0,\ldots,0
        )
  )
     )   )\in P_{\reals}
\end{align*}
As $\trm$ respects the logical relation by our induction hypothesis, it is enough to show that
\begin{align*}
    & (x\mapsto (f(x),\pi_i(\sem{\trm[2]}(f(x)))),
     (x\mapsto (g(x),\pi_i(\pi_1(\sem{\Dsynrev[\vGamma]{\trm[2]}}(g(x))))),\\
    &\qquad x\mapsto v\mapsto 
    h(x)(\pi_1(v)+\pi_2(\sem{\Dsynrev[\vGamma]{\trm[2]}}(g(x)))(0,\ldots,0,\pi_2(v),0,\ldots,0      )        )
    ))\in P_{\Gamma,\var:\reals}
  \end{align*}
Since $(f,(g,h))\in P_\Gamma$ by assumption, we merely need to check the following, by definition of $P_{\Gamma,\var:\reals}$:
\begin{align*}
    & (x\mapsto \pi_i(\sem{\trm[2]}(f(x))),
     (x\mapsto \pi_i(\pi_1(\sem{\Dsynrev[\vGamma]{\trm[2]}}(g(x)))),\\
    &\qquad x\mapsto v\mapsto 
    h(x)(\pi_2(\sem{\Dsynrev[\vGamma]{\trm[2]}}(g(x)))(0,\ldots,0,v,0,\ldots,0      ))        
    ))\in P_{\reals}
  \end{align*}
By definition of $P_{\reals^n}$ and linearity of $h(x)$, it is enough to show that 
\begin{align*}
    & (x\mapsto \sem{\trm[2]}(f(x)),
     (x\mapsto \pi_1(\sem{\Dsynrev[\vGamma]{\trm[2]}}(g(x))),\\
    &\qquad x\mapsto v\mapsto 
    h(x)(\pi_2(\sem{\Dsynrev[\vGamma]{\trm[2]}}(g(x)))(v      ))        
    ))\in P_{\reals^n}
  \end{align*}
Since $\trm[2]$ respects the logical relation by our induction hypothesis,
it is enough to show that 
\begin{align*}
    (f,(g,h))\in P_{\Gamma},
  \end{align*}
  which holds by assumption. \clearpage
\section{Term simplifications in the implementation}\label{appx:impl-simpl}

Our implementation\footnote{As also mentioned in \S\ref{sec:implementation}, the implementation is available at \href{https://github.com/VMatthijs/CHAD}{https://github.com/VMatthijs/CHAD}.} of the AD macros described in \S\ref{sec:implementation} includes a number of simplification rules on the concrete target language whose only purpose is to make the produced code more readable and easier to follow (without changing its asymptotic runtime cost).
The motivation for these rules is to generate legible code when applying the AD macros to example programs.
In this appendix, we list these simplification rules explicitly and show the implementation's output on the four example programs in Figs. \ref{fig:first-order-ad-example} and \ref{fig:higher-order-ad-example} under these simplification rules.
We do this to illustrate that:

\begin{enumerate}
\item
    the simplifications given here are evidently meaning-preserving, given the $\beta\eta+$ rules in Figs. \ref{fig:beta-eta} and \ref{fig:minimal-linear-beta-eta}, and are standard rules that any optimizing compiler would apply;
\item
    the resulting simplified output of the AD macros from \S\ref{sec:ad-transformation} is indeed equivalent to the differentiated programs in Figs. \ref{fig:first-order-ad-example} and \ref{fig:higher-order-ad-example}.
\end{enumerate}

The simplification rules in question are given below in Table~\ref{tab:impl-simpls}.
In the implementation, these are (at the time of writing) implemented in the simplifier for the concrete target language.\footnote{\href{https://github.com/VMatthijs/CHAD/blob/eedd6b12f224ed28ef9ca8650718d901c2b5e6a3/src/Concrete/Simplify.hs}{https://github.com/VMatthijs/CHAD/blob/eedd6b12f224ed28ef9ca8650718d901c2b5e6a3/src/Concrete/Simplify.hs}}

\newcommand\apslet[2]{\mathbf{let}\ #1\ \mathbf{in}\ #2}
\newcommand\apsfst[1]{\mathbf{fst}\ #1}
\newcommand\apssnd[1]{\mathbf{snd}\ #1}
\newcommand\apsplus{\mathbf{plus}}
\newcommand\apszero{\mathbf{zero}}
\newcommand\apsmap{\mathbf{map}}
\newcommand\apssum{\mathbf{sum}}
\newcommand\apszip{\mathbf{zip}}
\newcommand\apstabhline{\hline}

\begin{table}[h!]
\begin{tabular}{@{}l|l@{\ \ $\rightsquigarrow$\ \ }l|l}
\textbf{Name} & \multicolumn{2}{l|}{\textbf{Rule}} & \textbf{Justification} \\\hline
lamAppLet
    & $(\lambda x.\ e)\ a$
    & $\apslet{x = a}e$
    & lambda subst., let subst. \\\apstabhline
letRotate
    & $\apslet{x = (\apslet{y = a}b)}e$
    & $\apslet{y = a}{\apslet{x = b}e}$
    & let substitution \\\apstabhline
letPairSplit
    & $\apslet{x = (a, b)}e$
    & $\apslet{x_1 = a}{\apslet{x_2 = b}{\subst{e}{\sfor{x}{\langle x_1,x_2 \rangle}}}}$
    & let substitution \\\apstabhline
letInline
    & $\apslet{x = a}e$
    & $\subst{e}{\sfor{x}{a}}$
    & let substitution \\
& \multicolumn{2}{l|}{\qquad (if $a$ is cheap or used at most once in e)} & \\\apstabhline
pairProj${}_1$
    & $\apsfst{\langle a, b \rangle}$
    & $a$
    & $\beta$ pair \\\apstabhline
pairProj${}_2$
    & $\apssnd{\langle a, b \rangle}$
    & $b$
    & $\beta$ pair \\\apstabhline
pairEta
    & $\langle \apsfst a, \apssnd a \rangle$
    & $a$
    & $\eta$ pair \\\apstabhline
letProj${}_1$
    & $\apsfst{(\apslet{x = a}e)}$
    & $\apslet{x = a}{\apsfst e}$
    & let substitution \\\apstabhline
letProj${}_2$
    & $\apssnd{(\apslet{x = a}e)}$
    & $\apslet{x = a}{\apssnd e}$
    & let substitution \\\apstabhline
plusZero${}_1$
    & $\apsplus\ \apszero\ a$
    & $a$
    & equational rule \\\apstabhline
plusZero${}_2$
    & $\apsplus\ a\ \apszero$
    & $a$
    & equational rule \\\apstabhline
plusPair
    & $\apsplus\ \langle a, b \rangle\ \langle c, d \rangle$
    & $\langle \apsplus\ a\ c, \apsplus\ b\ d \rangle$
    & equational rule \\\apstabhline
plusLet${}_1$
    & $\apsplus\ (\apslet{x = e}a)\ b$
    & $\apslet{x = e}{\apsplus\ a\ b}$
    & let substitution \\\apstabhline
plusLet${}_2$
    & $\apsplus\ a\ (\apslet{x = e}b)$
    & $\apslet{x = e}{\apsplus\ a\ b}$
    & let substitution \\\apstabhline
algebra
    & $0 * x$, $\zero * x$
    & $0$ \qquad (etc.)
    & basic algebra \\\apstabhline
letLamPairSplit
    & $\apslet{f = \lambda x.\ \langle a, b \rangle}e$
    & \makecell[l]{
        $\mathbf{let}\ f_1 = \lambda x.\ a\ \mathbf{in}\ \mathbf{let}\ f_2 = \lambda x.\ b$ \\
        $\mathbf{in}\ \subst{e}{\sfor{f}{\lambda x.\langle f_1\,x,f_2\,x \rangle}}$
    }
    & $\eta$ lambda, let subst. \\\apstabhline
mapPairSplit
    & $\apsmap\ (\lambda x.\ (b, c))\ a$
    & \makecell[l]{
        $\mathbf{let}\ a' = a$ \\[-0.2em]
        $\mathbf{in}\ \langle \apsmap\ (\lambda x.\ b)\ a', \apsmap\ (\lambda x.\ c)\ a' \rangle$
    }
    & equational rule \\\apstabhline
mapZero
    & $\apsmap\ (\lambda x.\ \apszero)\ a$
    & $\apszero$
    & equational rule \\\apstabhline
sumZip
    & $\apssum\ (\apszip\ a\ b)$
    & $\langle \apssum\ a, \apssum\ b \rangle$
    & equational rule \\\apstabhline
sumZero
    & $\apssum\ \apszero$
    & $\apszero$
    & equational rule \\\apstabhline
sumSingleton
    & $\apssum\ (\apsmap\ (\lambda x.\ [x])\ e)$
    & $e$
    & equational rule
\end{tabular}
\caption{\label{tab:impl-simpls}
    The simplification rules that aid legibility and are available in the CHAD implementation in Haskell on the concrete target language.
}
\end{table}

The last column in the table shows the justification for the simplification rule: ``let substitution'', ``$\beta$ pair'', ``$\eta$ pair'', ``lambda substitution'' and ``$\eta$ lambda'' refer to the corresponding rules in Fig.~\ref{fig:beta-eta}.
The equational rules are either from Fig.~\ref{fig:minimal-linear-beta-eta} in the case of $\trm + \zero = \trm$ and its symmetric variant, from the type-based translation rules in \S\S\ref{ssec:api} in the case of plus on pairs, or otherwise general laws that hold for $\mathbf{zero}$ (i.e.\ $\zero$) and/or the array combinators in question.

Note that all the rules preserve the time complexity of the program through careful sharing of values with let-bindings.
These let-bindings could increase work only if the value is used only once in the body of the $\mathbf{let}$ -- but in that case, the `letInline' rule will eliminate the let-binding anyway.

\subsection{First-order example programs}

The output of our implementation for the forward derivative of Fig.~\ref{fig:first-order-ad-example} (a) and the reverse derivative of Fig.~\ref{fig:first-order-ad-example} (b) is shown below in Fig.~\ref{fig:impl-fo-ad-output}.

\begin{figure}[h!]
    \begin{subfigure}[b]{0.45\linewidth}
\begin{lstlisting}[mathescape=true]
let two = 2.0
    y   = two * x
    z   = x * y in
$\langle\langle\langle$y, z$\rangle$, cos z$\rangle$,
 $\lambda$x'.
   let y' = two * snd x'
       z' = plus (x * y')
                 (y * snd x') in
   $\langle\langle$y', z'$\rangle$
  $\,\,$,(0.0 - sin z) * z'$\rangle\rangle$
\end{lstlisting}
    \caption{
        The implementation's forward AD output on Fig.~\ref{fig:first-order-ad-example} (a).
        Compare this with Fig.~\ref{fig:first-order-ad-example} (c).
    }
    \end{subfigure}\hspace{0.08\linewidth}
    \begin{subfigure}[b]{0.45\linewidth}
\begin{lstlisting}[mathescape=true]
let two = 2.0
    y   = x1 * x4 + two * x2
    w   = y * x3 + x4 in
$\langle$sin w
,$\lambda$v'.
   let w' = cos w * v'
       y' = x3 * w' in
   $\langle\langle\langle\langle$zero
   $\hspace{1.1em}$,x4 * y'$\rangle$
   $\hspace{0.73em}$,two * y'$\rangle$
   $\hspace{0.34em}$,y * w'$\rangle$
   ,plus w' (x1 * y')$\rangle\rangle$
\end{lstlisting}
    \caption{
        The implementation's reverse AD output on Fig.~\ref{fig:first-order-ad-example} (b).
        Compare this with Fig.~\ref{fig:first-order-ad-example} (d).
    }
    \end{subfigure}

    \caption{\label{fig:impl-fo-ad-output}
        Output of our Haskell implementation when executed on the first-order example programs in Fig.~\ref{fig:first-order-ad-example} (a) and (b).
    }
\end{figure}

The simplification rules listed above have already been applied (otherwise the output would indeed be much less readable).
The only change we made to the literal text output of the implementation is formatting and variable renaming.

For both programs, we note that in the implementation, environments are encoded using snoc-lists: that is, the environment $\Gamma = \var_1:\reals, \var_2:\reals, \var_3:\reals$ is represented as $((\epsilon, \var_1:\reals), \var_2:\reals), \var_3:\reals$.
Hence, for this $\Gamma$, $\Dsyn{\Gamma}_2$ would be represented as $((\lUnit \t* \creals) \t* \creals) \t* \creals$.
This affects the code in Fig.~\ref{fig:impl-fo-ad-output}, where in subfigure (a), the variable \lstinline{x'} has type $\Unit \t* \reals$ rather than the $\reals$ type it had in Fig.~\ref{fig:first-order-ad-example} (c).
Furthermore, the output in Fig.~\ref{fig:impl-fo-ad-output} (b), which is the cotangent (i.e.\ adjoint) of the environment, has type $(((\Unit \t* \reals) \t* \reals) \t* \reals) \t* \reals$, meaning that the `zero' term in the result has type $\Unit$ and is thus equal to $\tUnit$.

It should be evident to the reader that these outputs are equivalent to the programs given in Fig.~\ref{fig:first-order-ad-example} (c) and (d).

\subsection{Second-order example program Fig.~\ref{fig:higher-order-ad-example} (a)}

The implementation's forward derivative of Fig.~\ref{fig:higher-order-ad-example} (a) is shown below in Fig.~\ref{fig:impl-ho-ad-output-a}.
This version contains a let-bound function `g' that does not occur in the code of Fig.~\ref{fig:higher-order-ad-example} (c).
However, inlining this function in the two places where it is used does not increase work, because pair projection and the equational rules concerning `zero' and `plus' leave only one half of the `plus' expression in `g' at each invocation site of `g'.
A version with `g' manually inlined and simplified using the stated rules is shown in Fig.~\ref{fig:impl-ho-ad-output-a2}.
(Our automatic simplifier cannot yet prove that inlining `g' does not increase work, and hence keeps it let-bound.)

\begin{figure}
\begin{subfigure}[b]{0.85\textwidth}
\begin{lstlisting}[mathescape=true,basicstyle=\small]
let f' = $\lambda$z. $\lambda$d. plus (x * snd d) (z * snd (fst d))
    zs = vreplicate x
    g = $\lambda$z. let dfun = f' z in
                (x * z + 1.0
                ,$\lambda$denv. plus (snd (fst (fst denv)) z)
                             (dfun (zero, snd denv)))
in (vmap ($\lambda$z. fst (g z)) zs
   ,$\lambda$x'.
      let zs' = vreplicate (snd x')
       in plus
            (vzipWith ($\lambda$z. $\lambda$z'. snd (g z) (zero, z'))
                      zs zs')
            (vmap ($\lambda$z. snd (g z) (((x', $\lambda$z. f' z (x', zero)), zs'), zero))
                  zs))
\end{lstlisting}
\end{subfigure}

    \caption{\label{fig:impl-ho-ad-output-a}
        Output of our Haskell implementation of the forward AD macro when executed on Fig.~\ref{fig:higher-order-ad-example} (a).
        The implementation writes operations on scalar arrays ($\reals^n$) with a `v' prefix.
    }
\end{figure}

\begin{figure}
\begin{subfigure}[b]{0.6\textwidth}
\begin{lstlisting}[mathescape=true,basicstyle=\small]
let f' = $\lambda$z. $\lambda$v. plus (x * snd v) (z * snd (fst v))
    zs = vreplicate x in
(vmap ($\lambda$z. x * z + 1.0) zs
,$\lambda$x'.
   let zs' = vreplicate (snd x') in
       plus
         (vzipWith ($\lambda$z. $\lambda$z'. f' z (zero, z'))
                   zs zs')
         (vmap ($\lambda$z. f' z (x', zero)) zs))
\end{lstlisting}
\end{subfigure}

    \caption{\label{fig:impl-ho-ad-output-a2}
        Manually simplified code from Fig.~\ref{fig:impl-ho-ad-output-a}, as described in the text.
    }
\end{figure}

First, note that the type of the variable \lstinline{x'} here is $\Unit \t* \creals$ instead of $\creals$ because of the snoc-list representation of environments, as was discussed in the previous subsection.
Because of this, the expression \lstinline{snd x'} in Fig.~\ref{fig:impl-ho-ad-output-a2} is equivalent to the expression \lstinline{x'} in Fig.~\ref{fig:higher-order-ad-example} (c).
Knowing this, let us work out how the implementation produced this output code, and how it is equivalent to the code in Fig.~\ref{fig:higher-order-ad-example} (c).

For the purposes of this explanation, the most important component of the source of Fig.~\ref{fig:higher-order-ad-example} (a) is its first line: `\lstinline[mathescape=true]{let f = $\lambda$z. x * z + 1 in }...'.
Recall from the forward AD macro from \S\S\ref{ssec:forward-chad-defs}:
\begin{flalign*}
& \Dsyn[\vGamma]{\letin{x}{\trm}{\trm[2]}}
    \qquad \defeq \qquad
    \pletin{x}{x'}{\Dsyn[\vGamma]{\trm}}{
        \pletin{y}{y'}{\Dsyn[\vGamma,x]{\trm[2]}}{
            \tPair{y}
                  {\lfun\lvar\lapp{y'}{\tPair{\lvar}{\lapp{x'}{\lvar}}}}}} &
\end{flalign*}
Hence, the binding of $\trm$ is transformed to a binding of $\Dsyn[\vGamma]{\trm}$, which is then used in the body.
Since $\trm$ is the term `\lstinline[mathescape=true]{$\lambda$z. x * z + 1}' here, and since the macro rule for lambda abstraction is as follows:
\begin{flalign*}
& \Dsyn[\vGamma]{\fun x \trm} \qquad\qquad\quad \defeq \qquad
    \letin{y}{\fun x\Dsyn[\vGamma,x]{\trm}}{
        \tPair{\fun x
                  \pletin{z}{z'}{y\,x}
                     {\tPair{z}{\lfun\lvar\lapp{z'}{\tPair{\zero}{\lvar}}}}}
              {\lfun\lvar\fun x\lapp{(\tSnd(y\,x))}{\tPair{\lvar}{\zero}}}
    } &
\end{flalign*}
we get the following result for $\Dsyn[\vGamma]{\fun z x * z + 1}$ with $\Gamma = \epsilon, x : \reals$:
\begin{align*}
\Dsyn[\epsilon,x]{\fun z x * z + 1} = \;
&\mathbf{let}\ y = \fun z \langle x * z + 1, \lfun{\langle\langle\tUnit,x'\rangle,z'\rangle} x * z' + z * x'\rangle \\[-0.4em]
&\mathbf{in}\ \langle\fun z \mathbf{let}\ \langle u, u'\rangle = y\ z\ \mathbf{in}\ \langle u, \lfun \lvar \lapp{u'}{\langle\zero, \lvar\rangle}\rangle,
    \lfun \lvar \fun z \lapp{(\tSnd (y\ z))}{\langle \lvar, \zero\rangle}\rangle
\end{align*}
This is the term that appears on the right-hand side of a let-binding in the forward AD transformed version of the code from Fig.~\ref{fig:higher-order-ad-example} (a).
Inlining of $y$ and some further simplification yields:
\begin{align*}
\Dsyn[\epsilon,x]{\fun z x * z + 1} = \;
&\langle\fun z \langle x * z + 1, \lfun \lvar x * \lvar\rangle,
    \lfun \lvar \fun z z * \tSnd \lvar\rangle
\end{align*}
where we recognize \lstinline{f} and \lstinline{f'} from Fig.~\ref{fig:higher-order-ad-example} (c).

The implementation instead simplifies $\Dsyn[\epsilon,x]{\fun z x * z + 1}$ by splitting the lambda bound to $y$ using `letLamPairSplit'.
The second of the resulting two lambda functions is $\fun z \lfun{\langle\langle\tUnit, x'\rangle, z'\rangle} x * z' + z * x'$, or by desugaring pattern matching, $\fun z \lfun \lvar x * \tSnd \lvar + z * \tSnd (\tFst \lvar)$.
We recognize this expression as the right-hand side of the \lstinline{f'} binding in Fig.~\ref{fig:impl-ho-ad-output-a2}.

We leave it to the reader to show that the code in Fig.~\ref{fig:impl-ho-ad-output-a2} is equivalent to the code in Fig.~\ref{fig:higher-order-ad-example} (c) under forward and reverse let-substitution.

\subsection{Second-order example program Fig.~\ref{fig:higher-order-ad-example} (b)}

When the implementation performs reverse AD on the code in Fig.~\ref{fig:higher-order-ad-example} (b) and simplifies the result using the simplification rules in Table~\ref{tab:impl-simpls}, the result is the code shown below in Fig.~\ref{fig:impl-ho-ad-output-b}.

\begin{figure}[h!]
\begin{subfigure}[b]{0.66\textwidth}
\begin{lstlisting}[mathescape=true]
$\langle$vsum (vmap ($\lambda$x2i. x1 * x2i) x2)
,$\lambda$w'.
   let ys' = vreplicate w' in
   $\langle\langle$zero, sum (map ($\lambda$p. evalOp EScalProd p)
                  $\hphantom{\langle\langle}$(zip (toList x2) (toList ys')))$\rangle$
  $\;\:$,vzipWith ($\lambda$x2i. $\lambda$y'. x1 * y') x2 ys'$\rangle\rangle$
\end{lstlisting}
\end{subfigure}

    \caption{\label{fig:impl-ho-ad-output-b}
        Output of our Haskell implementation of the reverse AD macro when executed on Fig.~\ref{fig:higher-order-ad-example} (b).
        The implementation writes operations on scalar arrays ($\reals^n$) with a `v' prefix; list operations are written without a prefix.
        Contrast this with Fig.~\ref{fig:higher-order-ad-example}, where all arrays are implemented as lists and hence there are no pure-array operations that would have been written using a `v' prefix.
    }
\end{figure}

First, note the \lstinline{evalOp EScalProd}.
Since scalar multiplication is implemented as an operation ($\op$), it takes a pair of the two scalars to multiply rather than two separate arguments.
Due to the application of the `pairEta' simplification rule from Table~\ref{tab:impl-simpls}, the argument to the multiplication was reduced from \lstinline[mathescape=true]{$\langle$fst p, snd p$\rangle$} to just `p', preventing the pretty-printer from showing \lstinline{fst p * snd p}; instead, this becomes a bare operation application to the argument `p'.
In the code in Fig.~\ref{fig:higher-order-ad-example} (d), this lambda is the argument to the `map' in the cotangents block, meaning that `p' stands for \lstinline[mathescape=true]{$\langle$x2i, y'$\rangle$}.

In this example, a copower structure is created because the code to be differentiated using reverse AD uses a function abstraction.
Here, this copower is interpreted using lists as described in \S\S\ref{ssec:api}.
The `toList' function converts an \emph{array} of scalars (i.e.\ a value of type $\reals^n$) to a \emph{list} of scalars.
The code in Fig.~\ref{fig:higher-order-ad-example} (d) goes further and interprets all arrays as lists, effectively removing the distinction between arrays originating from arrays in the source program and lists originating from copower values.
Once one recognizes the snoc-list representation of environments, inlining some let-bindings in Fig.~\ref{fig:higher-order-ad-example} (d) suffices to arrive at code equivalent to the code shown here.

\end{document}